\documentclass[sigconf,screen]{acmart}

\newif\ifanonymous
\anonymousfalse
\relax

\AtBeginDocument{}

\usepackage[utf8]{inputenc} 
\usepackage[T1]{fontenc}    
\usepackage{xcolor}         
\definecolor{linkdarkblue}{rgb}{0, 0.08, 0.45}    
\usepackage{url}            
\usepackage{array}          
\usepackage{colortbl}       
\usepackage{makecell}       
\usepackage{tabularx}       
\usepackage{bigstrut}       
\usepackage{multirow}       
\usepackage{colortbl}       
\usepackage{nicefrac}       
\usepackage{microtype}      
\usepackage{graphicx}       
\usepackage{bm}             
\usepackage[capitalise,noabbrev]{cleveref} 
\usepackage{subcaption}     
\usepackage{booktabs}       
\usepackage{wrapfig}        
\usepackage{pifont}         
\usepackage{algorithm}      
\usepackage{mathtools}      
\usepackage{enumitem}       
\usepackage{pifont}         
\usepackage{algorithm}      
\usepackage{algorithmicx}   
\usepackage{algpseudocode}  
\usepackage{thm-restate}    
\usepackage{fontawesome}    
\usepackage{comment}        
\usepackage{ulem}[normalem] 
\usepackage{amsmath,amsfonts,bm}

\def\eqref#1{equation~\ref{#1}}

\def\1{\bm{1}}

\DeclareMathAlphabet{\mathsfit}{\encodingdefault}{\sfdefault}{m}{sl}
\SetMathAlphabet{\mathsfit}{bold}{\encodingdefault}{\sfdefault}{bx}{n}

\DeclareMathOperator*{\argmin}{arg\,min}

\allowdisplaybreaks

\newcommand{\noindentparnoline}[1]{\noindent\textbf{#1} }
\newcommand{\noindentpar}[1]{\noindent\textbf{\uline{#1}} }

\numberwithin{equation}{section}    
\numberwithin{algorithm}{section}   

\theoremstyle{plain}
\newtheorem{theorem}{Theorem}[section]
\newtheorem*{theorem*}{Theorem}
\newtheorem{lemma}[theorem]{Lemma}
\newtheorem*{lemma*}{Lemma}

\newtheorem*{axiom*}{Axiom}

\newtheorem*{conjecture*}{Conjecture}

\newtheorem*{assumption*}{Assumption}

\newtheorem*{proposition*}{Proposition}

\newtheorem*{corollary*}{Corollary}

\theoremstyle{definition}

\newtheorem*{definition*}{Definition}

\newtheorem*{example*}{Example}

\newtheorem*{problem*}{Problem}

\newtheorem*{question*}{Question}

\newtheorem*{exercise*}{Exercise}

\newtheorem*{remark*}{Remark}

\definecolor{mydarkgreen}{rgb}{0.0, 0.5, 0.0}
\definecolor{skyblue}{rgb}{0.53, 0.81, 0.98}

\newcolumntype{Y}{>{\centering\arraybackslash}X}

\newcommand{\ie}{\text{i.e., }}         
\newcommand{\eg}{\text{e.g., }}         
\newcommand{\wrt}{\text{w.r.t. }}       
\newcommand{\cf}{\text{cf. }}           
\newcommand{\aka}{\text{a.k.a. }}       
\newcommand{\iid}{\text{i.i.d. }}       
\newcommand{\cdf}{\text{c.d.f. }}       
\newcommand{\versus}{\text{vs. }}       
\newcommand{\diff}{\mathrm{d}}          

\copyrightyear{2025}
\acmYear{2025}
\setcopyright{acmlicensed}
\acmConference[KDD '25] {Proceedings of the 31st ACM SIGKDD Conference on Knowledge Discovery and Data Mining V.2}{August 3--7, 2025}{Toronto, ON, Canada.}
\acmBooktitle{Proceedings of the 31st ACM SIGKDD Conference on Knowledge Discovery and Data Mining V.2 (KDD '25), August 3--7, 2025, Toronto, ON, Canada}
\acmISBN{979-8-4007-1454-2/25/08}

\begin{document}

\title{
    \texorpdfstring{Breaking the Top-$K$ Barrier: Advancing Top-$K$ Ranking Metrics Optimization in Recommender Systems}{Breaking the Top-K Barrier: Advancing Top-K Ranking Metrics Optimization in Recommender Systems}
}


\author{Weiqin Yang}
    \orcid{0000-0002-5750-5515}
    \affiliation{
        \institution{Zhejiang University}
        \city{Hangzhou}
        \country{China}
    }
    \email{tinysnow@zju.edu.cn}
\authornotemark[2]
\authornotemark[3]

\author{Jiawei Chen}
    \orcid{0000-0002-4752-2629}
    \affiliation{
        \institution{Zhejiang University}
        \city{Hangzhou}
        \country{China}
    }
    \email{sleepyhunt@zju.edu.cn}
\authornote{Corresponding author.}
\authornote{State Key Laboratory of Blockchain and Data Security, Zhejiang University.}
\authornote{College of Computer Science and Technology, Zhejiang University.}
\authornote{Hangzhou High-Tech Zone (Binjiang) Institute of Blockchain and Data Security.}

\author{Shengjia Zhang}
    \orcid{0009-0004-0209-2276}
    \affiliation{
        \institution{Zhejiang University}
        \city{Hangzhou}
        \country{China}
    }
    \email{shengjia.zhang@zju.edu.cn}
\authornotemark[2]
\authornotemark[3]

\author{Peng Wu}
    \orcid{0000-0001-7154-8880}
    \affiliation{
        \institution{Beijing Technology and Business University}
        \city{Beijing}
        \country{China}
    }
    \email{pengwu@btbu.edu.cn}
\authornote{School of Mathematics and Statistics, Beijing Technology and Business University.}

\author{Yuegang Sun}
    \orcid{0009-0009-2701-4641}
    \affiliation{
        \institution{Intelligence Indeed}
        \city{Hangzhou}
        \country{China}
    }
    \email{bulutuo@i-i.ai}

\author{Yan Feng}
    \orcid{0000-0002-3605-5404}
    \affiliation{
        \institution{Zhejiang University}
        \city{Hangzhou}
        \country{China}
    }
    \email{fengyan@zju.edu.cn}
\authornotemark[2]
\authornotemark[3]

\author{Chun Chen}
    \orcid{0000-0002-6198-7481}
    \affiliation{
        \institution{Zhejiang University}
        \city{Hangzhou}
        \country{China}
    }
    \email{chenc@zju.edu.cn}
\authornotemark[2]
\authornotemark[3]

\author{Can Wang}
    \orcid{0000-0002-5890-4307}
    \affiliation{
        \institution{Zhejiang University}
        \city{Hangzhou}
        \country{China}
    }
    \email{wcan@zju.edu.cn}
\authornotemark[2]
\authornotemark[4]

\ifanonymous
\else
    \renewcommand{\shortauthors}{Weiqin Yang et al.}
\fi



\begin{abstract}
    In the realm of recommender systems (RS), Top-$K$ ranking metrics such as NDCG@$K$ are the gold standard for evaluating recommendation performance. However, during the training of recommendation models, optimizing NDCG@$K$ poses significant challenges due to its inherent discontinuous nature and the intricate Top-$K$ truncation. Recent efforts to optimize NDCG@$K$ have either overlooked the Top-$K$ truncation or suffered from high computational costs and training instability. To overcome these limitations, we propose \textbf{SoftmaxLoss@$K$ (SL@$K$)}, a novel recommendation loss tailored for NDCG@$K$ optimization. Specifically, we integrate the quantile technique to handle Top-$K$ truncation and derive a smooth upper bound for optimizing NDCG@$K$ to address discontinuity. The resulting SL@$K$ loss has several desirable properties, including theoretical guarantees, ease of implementation, computational efficiency, gradient stability, and noise robustness. Extensive experiments on four real-world datasets and three recommendation backbones demonstrate that SL@$K$ outperforms existing losses with a notable average improvement of \textbf{6.03\%}. The code is available at \url{https://github.com/Tiny-Snow/IR-Benchmark}.
\end{abstract}


\keywords{Recommender systems; Surrogate loss; \texorpdfstring{NDCG@$K$}{NDCG@K} optimization}



\maketitle

\newcommand\kddavailabilityurl{https://github.com/Tiny-Snow/IR-Benchmark}

\ifdefempty{\kddavailabilityurl}{}{
\begingroup\small\noindent\raggedright\textbf{Artifacts Link:}\\
The source code of this paper has been made publicly available at \url{\kddavailabilityurl}.
\endgroup
}



\section{Introduction} \label{sec:introduction}

Recommender systems (RS) \citep{ko2022survey,wang2024distributionally,wang2024llm4dsr,chen2021autodebias,gao2023cirs,gao2023alleviating} have been widely applied in various personalized services \citep{nie2019multimodal,ren2017social}. The primary goal of RS is to model users' preferences on items and subsequently retrieve a select number of items that users are most likely to interact with \citep{liu2009learning,li2020sampling,hurley2011novelty}. In practice, RS typically display only $K$ top-ranked items to users based on their preference scores. Therefore, \emph{Top-$K$ ranking metrics}, \eg NDCG@$K$ \citep{he2017neural}, are commonly used to evaluate recommendation performance. Unlike \emph{full-ranking metrics}, \eg NDCG \citep{jarvelin2017ir}, which assess the entire ranking list, Top-$K$ ranking metrics focus on the quality of the items ranked within the Top-$K$ positions, making them more aligned with practical requirements.

\begin{figure*}[t]
    \centering
    \begin{subfigure}[b]{0.51\textwidth}
        \centering
        \includegraphics[width=\textwidth]{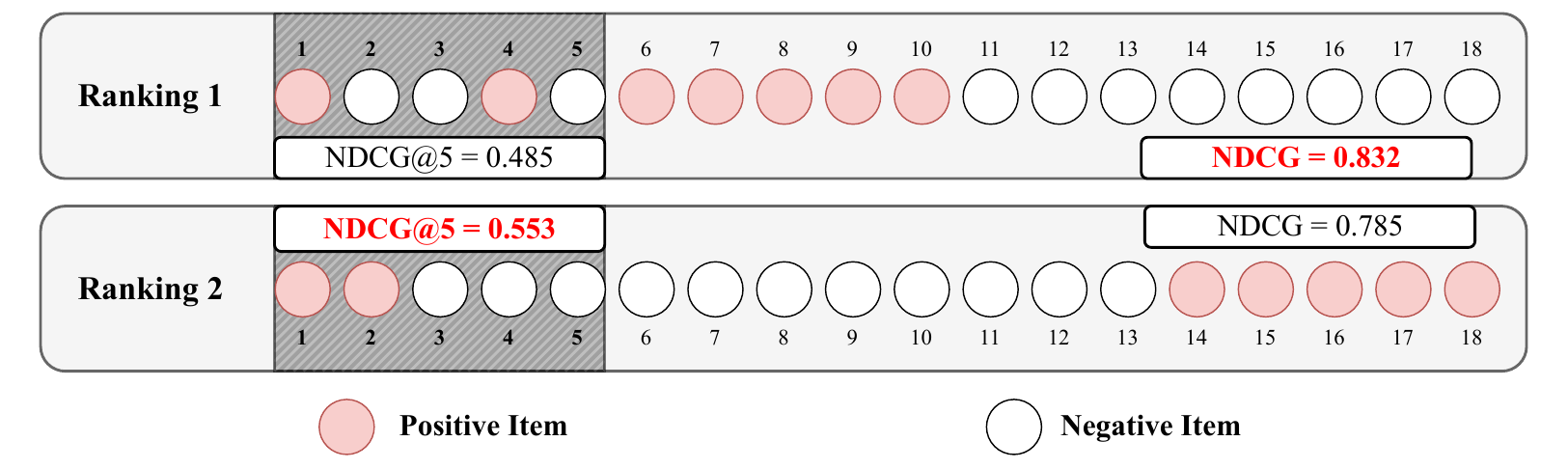}
        \caption{Inconsistency between NDCG and NDCG@$K$.}
        \label{fig:ndcg_ndcgatk}
    \end{subfigure}
    \begin{subfigure}[b]{0.24\textwidth}
        \centering
        \includegraphics[width=\textwidth]{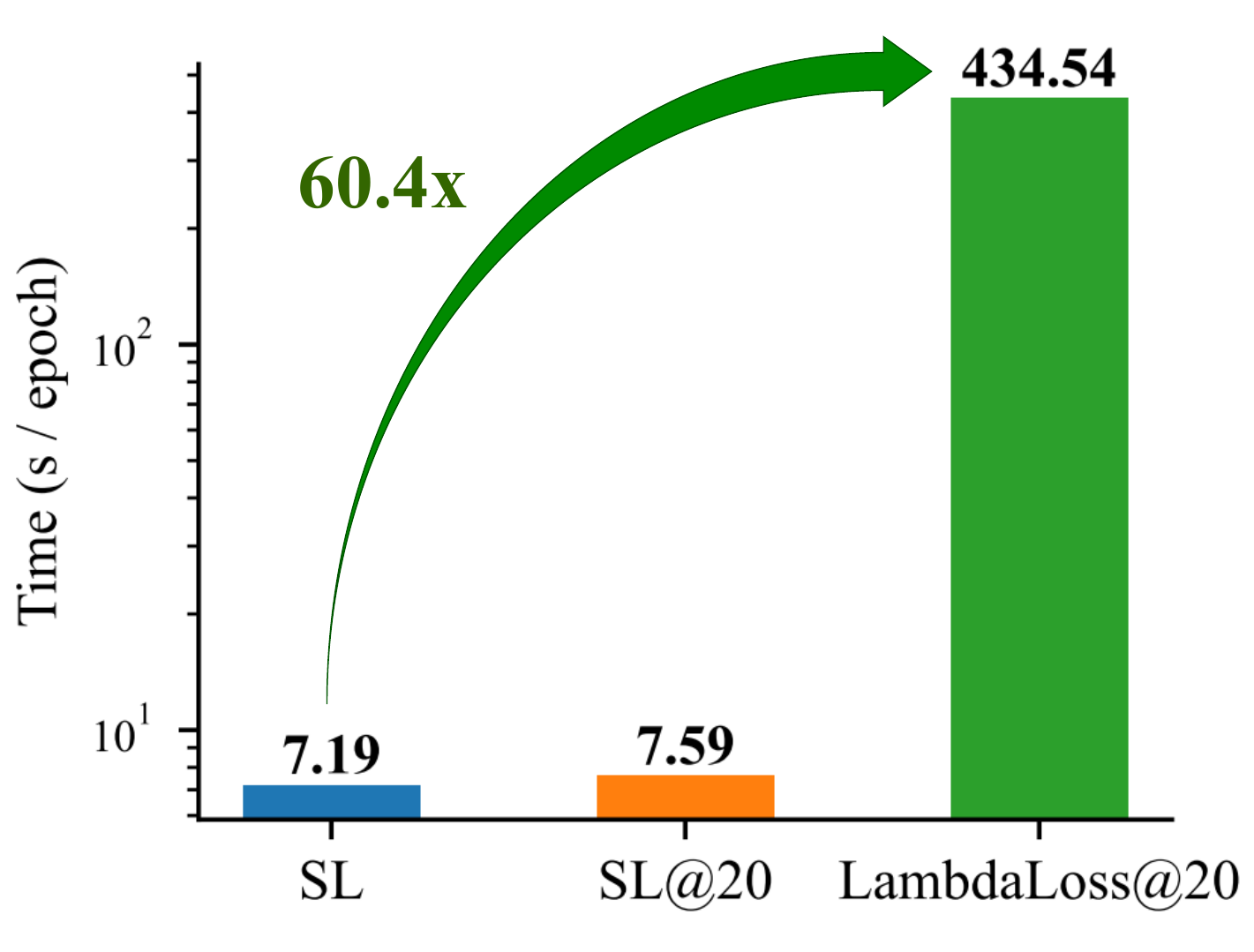}
        \caption{Execution time.}
        \label{fig:execution_time_comparison}
    \end{subfigure}
    \begin{subfigure}[b]{0.24\textwidth}
        \centering
        \includegraphics[width=\textwidth]{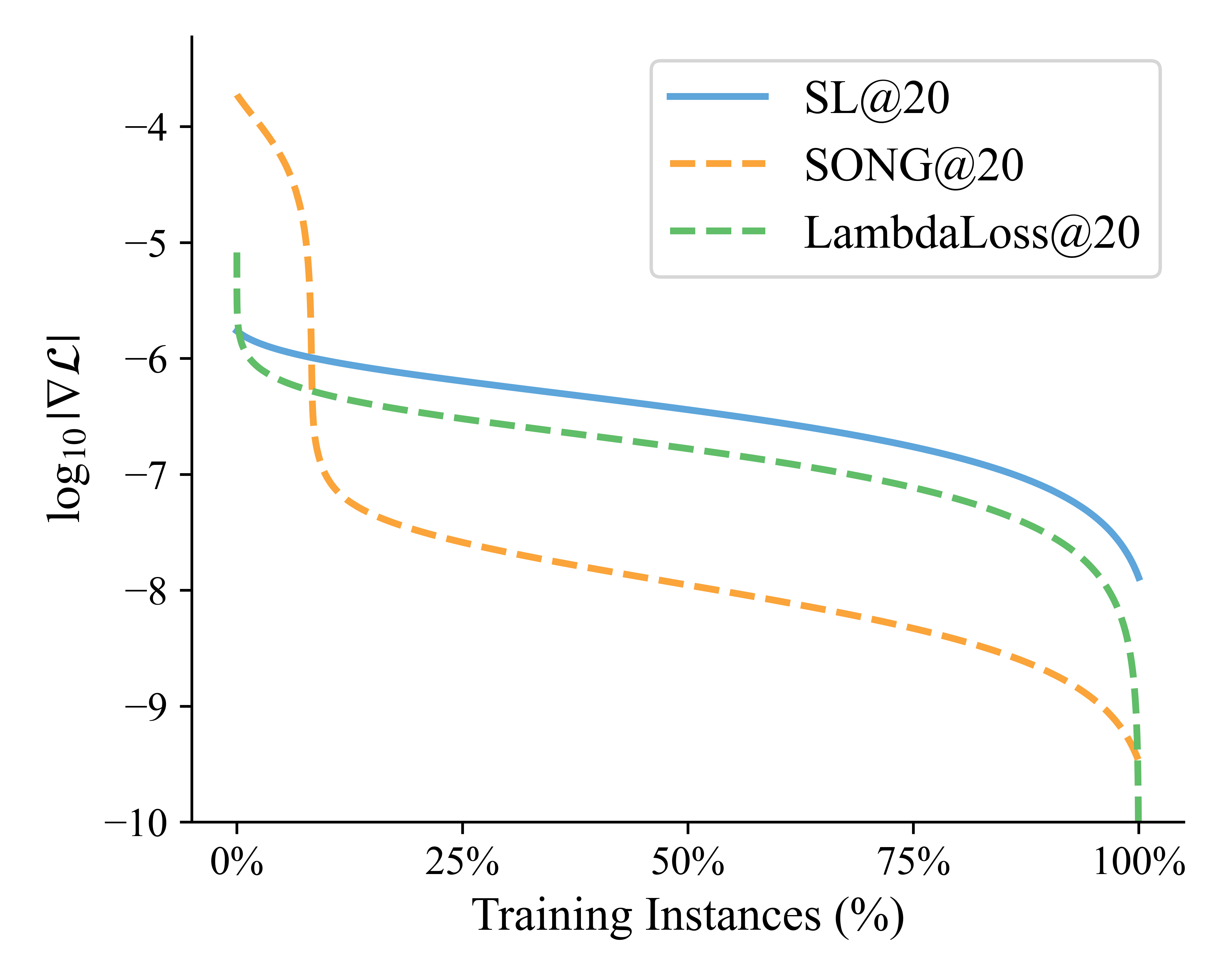}
        \caption{Gradient distribution.}
        \label{fig:gradient_distribution_comparison}
    \end{subfigure}
    \caption{(a) Inconsistency between NDCG and NDCG@$K$. Ranking 1 and Ranking 2 represent two different ranking lists of the same set of items, where red/white circles denote positive/negative items, respectively. While Ranking 1 has a better NDCG than Ranking 2, it has worse NDCG@5. (b) Execution time comparison. LambdaLoss@$K$ incurs a significantly higher (60 times) computational overhead compared to SL and SL@$K$ on the Electronic dataset (8K items). (c) Gradient distribution comparison. LambdaLoss@$K$ and SONG@$K$ exhibit skewed long-tailed gradient distributions, where top-5\% samples contribute over 90\% of the overall gradients. In contrast, SL@$K$ achieves a more moderate gradient distribution, where top-5\% samples contribute less than 15\% of the overall gradients. This leads to better data utilization and training stability.}
    \Description{Illustration of inconsistency between NDCG and NDCG@$K$, execution time comparison, and gradient distribution comparison.}
    \label{fig:introduction-fig}
\end{figure*}

\noindentpar{Challenges.}
Despite the widespread adoption of the NDCG@$K$ metric, its optimization presents two fundamental challenges:

\begin{itemize}[topsep=3pt,leftmargin=10pt,itemsep=0pt]
    \item \emph{Top-$K$ truncation}: NDCG@$K$ involves truncating the ranking list, requiring the identification of which items occupy the Top-$K$ positions. This necessitates sorting the entire item list, imposing significant computational costs and optimization complexities.
    \item \emph{Discontinuity}: NDCG@$K$ is inherently discontinuous or flat everywhere in the space of model parameters, which severely impedes the effectiveness of gradient-based optimization methods.
\end{itemize}

\noindentpar{Existing works.}
Recent studies have introduced two types of NDCG@$K$ surrogate losses to tackle these challenges. However, these approaches still exhibit significant limitations:

\begin{itemize}[topsep=3pt,leftmargin=10pt,itemsep=0pt]
    \item A prominent line of work focuses on \emph{optimizing full-ranking metrics} like NDCG, without accounting for the complex Top-$K$ truncation. Notable among these is Softmax Loss (SL) \citep{wu2024effectiveness}, which serves as an upper bound for optimizing NDCG and demonstrates state-of-the-art performance \citep{wu2023bsl,yang2024psl,bruch2019analysis,wang2025msl}. Moreover, SL enjoys practical advantages in terms of formulation simplicity and computational efficiency. However, we argue that NDCG is \emph{inconsistent} with NDCG@$K$ --- NDCG@$K$ focuses exclusively on a few top-ranked items, while NDCG evaluates the entire ranking list. This discrepancy means that optimizing NDCG does not always yield improvements in NDCG@$K$ and may even lead to performance degradation, as illustrated in \cref{fig:ndcg_ndcgatk}. Therefore, without incorporating Top-$K$ truncation, these NDCG surrogate losses could inherently encounter performance bottlenecks. 
    \item Few studies have explored \emph{incorporating Top-$K$ truncation} into NDCG@$K$ optimization. For example, LambdaLoss@$K$ \citep{jagerman2022optimizing} incorporates truncation-aware lambda weights \citep{burges2006learning,wang2018lambdaloss} based on ranking positions to optimize NDCG@$K$, exhibiting superior performance compared to full-ranking surrogate losses like SL \citep{wu2024effectiveness} and LambdaLoss \citep{wang2018lambdaloss} in learning to rank tasks \citep{liu2009learning}. Another notable work is SONG@$K$ \citep{qiu2022large}, which employs a ingenious bilevel compositional optimization strategy \citep{wang2017stochastic} to optimize NDCG@$K$ with provable guarantees. While these methods have proven effective in other tasks, we find them \emph{ineffective} for recommendation due to the large-scale and sparse nature of RS data. Specifically, LambdaLoss@$K$ requires sorting the entire item list to calculate lambda weights, which is computationally impractical in real-world RS (\cf \cref{fig:execution_time_comparison}). Additionally, both LambdaLoss@$K$ and SONG@$K$ exhibit a highly skewed gradient distribution in RS --- a few instances dominate the gradients, while the majority contribute negligibly (\cf \cref{fig:gradient_distribution_comparison}). This severely hinders effective data utilization and model training.
\end{itemize}

\noindentpar{Our method.}
Given the critical importance of optimizing NDCG@$K$ and the inherent limitations of existing losses in RS, it is essential to devise a more effective NDCG@$K$ surrogate loss. In this paper, we propose \textbf{SoftmaxLoss@$K$ (SL@$K$)}, incorporating the following two key strategies to address the aforementioned challenges:

\begin{itemize}[topsep=3pt,leftmargin=10pt,itemsep=0pt]
    \item To address the \emph{Top-$K$ truncation} challenge, we employ the \emph{quantile technique} \citep{koenker2005quantile,boyd2012accuracy}. Specifically, we introduce a Top-$K$ quantile for each user as a threshold score that separates the Top-$K$ items from the remainder. This technique transforms the complex Top-$K$ truncation into a simpler comparison between item scores and quantiles, which circumvents the need for explicit calculations of ranking positions. We further develop a Monte Carlo-based quantile estimation strategy that achieves both computational efficiency and theoretical precision guarantees.
    \item To overcome the \emph{discontinuity} challenge, we derive an upper bound for optimizing NDCG@$K$ and relax it into a smooth surrogate loss --- SL@$K$. Our analysis proves that SL@$K$ serves as a tight upper bound for $-\log$ NDCG@$K$, ensuring its theoretical effectiveness in Top-$K$ recommendation.
\end{itemize}

Beyond its theoretical foundations, SL@$K$ offers several practical advantages: (i) \emph{Ease of implementation}: Compared to SL, SL@$K$ only adds a quantile-based weight for each positive instance, making it easy to implement and integrate into existing RS. (ii) \emph{Computational efficiency}: The adoption of quantile estimation and relaxation techniques incurs minimal additional computational overhead over SL (\cf \cref{fig:execution_time_comparison}). (iii) \emph{Gradient stability}: SL@$K$ exhibits more moderate gradient distribution characteristics during training (\cf \cref{fig:gradient_distribution_comparison}), promoting effective data utilization and improving model training stability. (iv) \emph{Noise robustness}: SL@$K$ demonstrates enhanced robustness against false positive noise \citep{chen2023bias,wen2019leveraging}, \ie interactions arising from extraneous factors rather than user preferences.

Finally, to empirically validate the effectiveness of SL@$K$, we conduct extensive experiments on four real-world recommendation datasets and three typical recommendation backbones. Experimental results demonstrate that SL@$K$ achieves impressive performance improvements of \textbf{6.03\%} on average. Additional experiments, including an exploration of varying hyperparameter $K$ and robustness evaluations, confirm that SL@$K$ is not only well-aligned with NDCG@$K$, but also exhibits superior resistance to noise. Moreover, since SL@$K$ is essentially a general ranking loss, it can be seamlessly applied to other information retrieval (IR) tasks. We extend our work to three different IR tasks, including learning to rank (LTR) \citep{pobrotyn2021neuralndcg}, sequential recommendation (SeqRec) \citep{kang2018self}, and link prediction (LP) \citep{li2023evaluating}. Empirical results validate the versatility and effectiveness of SL@$K$ across diverse IR tasks.

\noindentpar{Contributions.}
In summary, our contributions are as follows:  
\begin{itemize}[topsep=3pt,leftmargin=10pt,itemsep=0pt]  
    \item We highlight the significance of optimizing the Top-$K$ ranking metric NDCG@$K$ in recommendation and reveal the limitations of existing losses.
    \item We propose a novel loss function, SL@$K$, tailored for Top-$K$ recommendation by integrating the quantile technique and analyzing the upper bound of NDCG@$K$.  
    \item We conduct extensive experiments on various real-world datasets and backbones, demonstrating the superiority of SL@$K$ over existing losses, achieving an average improvement of 6.03\%.
    \item We extend SL@$K$ to three different IR tasks, validating its versatility and effectiveness beyond conventional recommendation.
\end{itemize}


\section{Preliminaries} \label{sec:preliminaries}

In this section, we first present the task formulation (\cref{subsec:task_formulation}), then highlight the challenges in optimizing NDCG@$K$ (\cref{subsec:formulation_of_ndcgatk}), and finally introduce Softmax Loss (SL) \citep{wu2024effectiveness} while discussing its limitations in optimizing NDCG@$K$ (\cref{subsec:softmax_loss}).

\subsection{\texorpdfstring{Top-$K$ Recommendation}{Top-K Recommendation}} \label{subsec:task_formulation}

In this work, we focus on the Top-$K$ recommendation from implicit feedback, a widely-used scenario in recommender systems (RS) \citep{su2009survey,zhu2019improving}. Specifically, given an RS with a user set $\mathcal{U}$ and an item set $\mathcal{I}$, let $\mathcal{D} = \{y_{ui} : u \in \mathcal{U}, i \in \mathcal{I}\}$ denote the historical interactions between users and items, where $y_{ui} = 1$ indicates that user $u$ has interacted with item $i$, and $y_{ui} = 0$ indicates no interaction. For each user $u$, we denote $\mathcal{P}_u = \{ i \in \mathcal{I} : y_{ui} = 1\}$ as the set of positive items, and $\mathcal{N}_u = \mathcal{I} \setminus \mathcal{P}_u$ as the set of negative items. The recommendation task can be formulated as follows: learning user preferences from dataset $\mathcal{D}$ and recommending the Top-$K$ items that users are most likely to interact with.

Formally, modern RS typically infer user preferences for items with a learnable model $s_{ui} = f_{\Theta}(u, i)$, where $f_{\Theta}(u, i): \mathcal{U} \times \mathcal{I} \to \mathbb{R}$ can be any flexible recommendation backbone with parameters $\Theta$, mapping user/item features (\eg IDs) into their preference scores $s_{ui}$. Subsequently, the Top-$K$ items with the highest scores $s_{ui}$ are retrieved as recommendations. In this work, we focus not on model architecture design but instead on exploring the recommendation loss. Given that the loss function guides the optimization direction of models, its importance cannot be overemphasized \citep{rendle2009bpr}.


\subsection{\texorpdfstring{NDCG@$K$}{NDCG@K} Metric} \label{subsec:formulation_of_ndcgatk}

\noindentpar{Formulation of NDCG@$K$.}
Given the Top-$K$ nature of RS, Top-$K$ ranking metrics have been widely used to evaluate the recommendation performance. This work focuses on the most representative Top-$K$ ranking metric, \ie NDCG@$K$ (Normalized Discounted Cumulative Gain with Top-$K$ truncation) \citep{he2017neural,jarvelin2017ir}. Formally, for each user $u$, NDCG@$K$ can be formulated as follows:
\begin{equation} \label{eq:dcg_idcg}
    \mathrm{NDCG}@K(u) = \frac{\mathrm{DCG}@K(u)}{\mathrm{IDCG}@K(u)}, 
    \ \ 
    \mathrm{DCG}@K(u) = \sum_{i \in \mathcal{P}_u} \frac{\mathbb{I}(\pi_{ui} \leq K)}{\log_2(\pi_{ui} + 1)} ,
\end{equation}
where $\mathbb{I}(\cdot)$ is the indicator function, $\pi_{ui} = \sum_{j \in \mathcal{I}} \mathbb{I}(s_{uj} \geq s_{ui})$ is the ranking position of item $i$ for user $u$, and IDCG@$K$ is a normalizing constant representing the optimal DCG@$K$  with an ideal ranking. 

As observed, NDCG@$K$ not only evaluates the number of positive items within the Top-$K$ recommendations (similar to other Top-$K$ metrics, \eg Recall@$K$ and Precision@$K$), but also accounts for their ranking positions, \ie higher-ranked items contribute more to NDCG@$K$. This makes NDCG@$K$ a more practical metric for recommendation. Therefore, this work focuses on NDCG@K, while we also observe that effectively optimizing NDCG@$K$ can bring improvements on other Top-$K$ metrics like Recall@K (\cf \cref{tab:slatk_versus_sota}). 

\noindentpar{Challenges in optimizing NDCG@$K$.}
While NDCG@$K$ is widely applied, directly optimizing it presents significant challenges:

\begin{itemize}[topsep=3pt,leftmargin=10pt,itemsep=0pt]  
    \item \emph{Challenge 1: Top-$K$ truncation.} NDCG@$K$ involves truncating the ranking list, as indicated by the term $\mathbb{I}(\pi_{ui} \leq K)$ in \cref{eq:dcg_idcg}. This implies the need to determine whether an item is situated within the Top-$K$ positions. Directly computing this involves sorting all items for each user, which is computationally impractical for RS. Moreover, this truncation introduces highly complex gradient signals, complicating the optimization process.
    \item \emph{Challenge 2: Discontinuity.} NDCG@$K$ is a discontinuous metric as it incorporates the indicator function and the ranking positions. Furthermore, this metric exhibits flat characteristics across most regions of the parameter space, \ie the metric remains unchanged with minor perturbations of $s_{ui}$ almost everywhere. This results in the gradient being undefined or vanishing, posing substantial challenges to the effectiveness of existing gradient-based optimization methods \citep{robbins1951stochastic}. Consequently, a smooth surrogate for NDCG@K is required to facilitate optimization.
\end{itemize}


\subsection{Softmax Loss} \label{subsec:softmax_loss}

\textbf{Softmax Loss (SL)} \citep{wu2024effectiveness} has achieved remarkable success in RS. Specifically, SL integrates a contrastive learning paradigm \citep{liu2021self}. It normalizes the preference scores to a multinomial distribution \citep{casella2024statistical} by Softmax operator, augmenting the scores of positive items as compared to the negative ones \citep{cao2007l2r}. Formally, SL is defined as:
\begin{equation} \label{eq:sl}
    \mathcal{L}_{\text{SL}}(u)
    = -\sum_{i \in \mathcal{P}_u} \log \frac{\exp(s_{ui} / \tau)}{\displaystyle\sum_{j \in \mathcal{I}} \exp(s_{uj} / \tau)}
    = \sum_{i \in \mathcal{P}_u} \log\left(\sum_{j \in \mathcal{I}} \exp(d_{uij} / \tau)\right) ,
\end{equation}
where $d_{uij} = s_{uj} - s_{ui}$ is the negative-positive score difference, and $\tau$ is a temperature coefficient controlling the sharpness of the Softmax distribution.

\noindentpar{Underlying rationale of SL.}
The success of SL can be attributed to two main aspects: (i) \emph{Theoretical guarantees}: SL has been proven to serve as an upper bound of $-\log$ NDCG \citep{bruch2019analysis,yang2024psl}, ensuring that optimizing SL is consistent with optimizing NDCG, leading to state-of-the-art (SOTA) performance \citep{wu2023bsl}.
(ii) \emph{Computational efficiency}: SL does not require accurately calculating the ranking positions, which is time-consuming. In fact, SL can be efficiently estimated through negative sampling \citep{wu2023bsl}. That is, the sum of item $j$ over the entire item set $\mathcal{I}$ in \cref{eq:sl} can be approximated by sampling a few negative items through uniform \citep{gutmann2012noise,yang2024psl} or in-batch \citep{wu2024effectiveness,ji2019invariant} sampling. These advantages make SL a practical and effective choice for NDCG optimization, demonstrating superior performance and efficiency over other NDCG surrogate methods, including ranking-based (\eg Smooth-NDCG \citep{chapelle2010gradient}), Gumbel-based (\eg NeuralSort \citep{grover2019stochastic}), and neural-based (\eg GuidedRec \citep{rashed2021guided}) methods. Nowadays, SL has been extensively applied in practice, attracting considerable research exploration with a substantial amount of follow-up work.

\noindentpar{Limitations of SL.}
While SL serves as an effective surrogate loss for NDCG, a significant gap remains between NDCG and NDCG@$K$, which limits its performance. As \cref{fig:ndcg_ndcgatk} shows, optimizing NDCG does not consistently improve NDCG@$K$ and sometimes even leads to performance drops. This limitation still exists in more advanced SL-based losses, \eg AdvInfoNCE \citep{zhang2024empowering}, BSL \citep{wu2023bsl}, and PSL \citep{yang2024psl}. Therefore, how to bridge this gap and effectively model the Top-$K$ truncation in recommendation loss remains an open challenge.



\section{Methodology} \label{sec:methodology}

To bridge the gap towards NDCG@$K$ optimization, we propose \textbf{SoftmaxLoss@$K$ (SL@$K$)}, a novel NDCG@$K$ surrogate loss. In this section, we first present the derivations and implementation details of SL@$K$ (\cref{subsec:softmaxlossatk}). Then, we analyze its properties and discuss its advantages over existing losses (\cref{subsec:slatk_properties}).

\subsection{Proposed Loss: \texorpdfstring{SoftmaxLoss@$K$}{SoftmaxLoss@K}} \label{subsec:softmaxlossatk}

The primary challenges in optimizing NDCG@$K$, as discussed in \cref{subsec:formulation_of_ndcgatk}, are the \emph{Top-$K$ truncation} and the \emph{discontinuity}. To tackle these challenges, we introduce the following two techniques.

\subsubsection{Quantile-based Top-$K$ Truncation} \label{subsubsec:quantile_based_topk_truncation}

To address the \emph{Top-$K$ truncation} challenge, we need to estimate the Top-$K$ truncation term $\mathbb{I}(\pi_{ui} \leq K)$, which involves estimating the ranking position $\pi_{ui}$ for each interaction $(u, i)$. However, directly estimating $\pi_{ui}$ is particularly challenging. Sorting all items for each user to calculate $\pi_{ui}$ will incur a computational cost of $O(|\mathcal{U}||\mathcal{I}| \log |\mathcal{I}|)$, which is impractical for real-world RS with immense user and item scales.

To overcome this, we borrow the \textbf{quantile technique} \citep{koenker2005quantile,hao2007quantile}. Specifically, we introduce a \emph{Top-$K$ quantile} $\beta_{u}^{K}$ for each user $u$, \ie
\begin{equation} \label{eq:quantile}
    \beta_{u}^{K} \coloneqq \inf \{s_{ui} : \pi_{ui} \leq K\} .
\end{equation}
This quantile acts as a threshold score that separates the Top-$K$ items from the rest. Specifically, if an item's score is larger than the quantile, \ie $s_{ui} \geq \beta_{u}^{K}$, then item $i$ is Top-$K$ ranked; conversely, $s_{ui} < \beta_{u}^{K}$ implies that item $i$ is outside the Top-$K$ positions. Therefore, the Top-$K$ truncation term can be rewritten as:
\begin{equation} \label{eq:quantile-delta}
    \mathbb{I}(\pi_{ui} \leq K) = \mathbb{I}(s_{ui} \geq \beta_{u}^{K}) .
\end{equation}

This transformation reduces the complex truncation to a simple comparison between the preference score $s_{ui}$ and the quantile $\beta_{u}^{K}$, thus avoiding the need to directly estimate $\pi_{ui}$. This makes the Top-$K$ truncation both computationally efficient and easy to optimize. To handle the complexities of quantile estimation, we further propose a simple Monte Carlo-based quantile estimation strategy in \cref{subsubsec:quantile_estimation}, which guarantees both high efficiency and precision. Notably, while quantile-based techniques have been explored in previous works -- e.g., AATP \citep{boyd2012accuracy} employs quantiles to optimize Top-$K$ accuracy, and SONG@$K$ \citep{qiu2022large} adopts quantile-related thresholds for bilevel compositional optimization -- we adapt this approach specifically for NDCG@$K$ optimization in the context of recommendation. Specifically, we propose a novel tailored recommendation loss and a dedicated quantile estimation strategy, which address the unique challenges in RS. Readers can refer to \cref{subsubsec:comparison_with_existing_losses,sec:related_work} for a detailed comparison between our proposed loss and existing methods, as well as a discussion of their limitations.

\subsubsection{Smooth Surrogate for NDCG@$K$} \label{subsubsec:smooth_surrogate_for_ndcgatk}

To tackle the \emph{discontinuity} challenge, we proceed to relax the discontinuous NDCG@$K$ into a smooth surrogate. Specifically, our approach focuses on deriving a smooth upper bound for $-\log$ DCG@$K$, since optimizing this upper bound is equivalent to lifting NDCG@$K$\footnote{Note that optimizing DCG@$K$ and NDCG@$K$ is equivalent, as the normalization term IDCG@$K$ is a constant.} \citep{yang2024psl,wu2024effectiveness}. To ensure mathematical well-definedness, we make a simple assumption that DCG@$K$ is non-zero, which is practical in optimization\footnote{This assumption is conventional in RS \citep{yang2024psl,wu2024effectiveness,bruch2019analysis}. Note that DCG@$K = 0$ suggests the worst result. During training, the scores of positive instances are rapidly elevated. As a result, there is almost always at least one positive item within the Top-$K$ positions, ensuring DCG@$K > 0$. (\cf \cref{subapp:proof_slatk} for empirical validation).}.

\noindentpar{Upper bound derivation.}
While several successful examples (\eg SL) of relaxing full-ranking metric DCG exist as references \citep{wu2024effectiveness,yang2024psl,wang2018lambdaloss}, special care must be taken to account for the differences in DCG@$K$ introduced by the Top-$K$ truncation. Based on the quantile technique and some specific relaxations, we can derive an upper bound for $-\log$ DCG@$K$ as follows:
\begin{subequations} \label{eq:slatk-derivation}
    \begin{align}
        & -\log \mathrm{DCG}@K(u)
            \nonumber
        & \\
        & \stackrel{(\ref{eq:quantile-delta})}{=}
            -\log \left( \sum_{i \in \mathcal{P}_u} \mathbb{I}(s_{ui} \geq \beta_{u}^{K}) \frac{1}{\log_2(\pi_{ui} + 1)} \right)
            \label{eq:slatk-derivation-1}
        & \\
        & \stackrel{\text{\ding{172}}}{\leq}
            -\log \left( \sum_{i \in \mathcal{P}_u} \mathbb{I}(s_{ui} \geq \beta_{u}^{K}) \frac{1}{\pi_{ui}} \right)
            \label{eq:slatk-derivation-2}
        & \\
        & =
            -\log \left( \sum_{i \in \mathcal{P}_u} \frac{\mathbb{I}(s_{ui} \geq \beta_{u}^{K})}{H_{u}^{K}} \frac{1}{\pi_{ui}} \right) - \log H_{u}^{K}
            \label{eq:slatk-derivation-3}
        & \\
        & \stackrel{\text{\ding{173}}}{\leq}
            \sum_{i \in \mathcal{P}_u} \frac{\mathbb{I}(s_{ui} \geq \beta_{u}^{K})}{H_{u}^{K}} \left( -\log \frac{1}{\pi_{ui}} \right) - \log H_{u}^{K}
            \label{eq:slatk-derivation-4}
        & \\
        & \stackrel{\text{\ding{174}}}{\leq}
            \sum_{i \in \mathcal{P}_u} \mathbb{I}(s_{ui} \geq \beta_{u}^{K}) \log \pi_{ui} , 
            \label{eq:slatk-derivation-5}
        & 
    \end{align}
\end{subequations}
where $H_{u}^{K} = \sum_{i \in \mathcal{P}_u} \mathbb{I}(s_{ui} \geq \beta_{u}^{K})$ is the number of Top-$K$ positive items (\aka hits) for user $u$. \cref{eq:slatk-derivation-3} is well-defined since $H_{u}^{K} \geq 1$ due to our non-zero assumption\footnote{Since DCG@$K$ $> 0$, there is at least one Top-$K$ hit $i$ such that $s_{ui} \geq \beta_{u}^{K}$.}. Several important relaxations are applied in \cref{eq:slatk-derivation}: \ding{172} is due to $\log_2(\pi_{ui} + 1) \leq \pi_{ui}$; \ding{173} is due to Jensen's inequality \citep{jensen1906fonctions}; \ding{174} is due to $H_{u}^{K} \geq 1$. 

The motivation behind the relaxations \ding{172} and \ding{173} is to simplify the DCG term $1 / \log_2(\pi_{ui} + 1)$, which includes the ranking position $\pi_{ui}$ in the denominator. It is important to note that the ranking position $\pi_{ui}$ is intricate and challenging to estimate accurately. Retaining $\pi_{ui}$ in the denominator could exacerbate the optimization difficulty, potentially leading to high estimation errors and numerical instability. SONG@$K$ \citep{qiu2022large} is a representative example. Although SONG@$K$ utilizes a sophisticated compositional optimization technique \citep{wang2017stochastic}, it still performs poorly in RS due to its highly skewed gradient distributions (\cf \cref{fig:gradient_distribution_comparison,tab:slatk_versus_sota}). Therefore, we follow the successful paths of SL \citep{wu2024effectiveness} and PSL \citep{yang2024psl}, aiming to simplify this complex structure. This significantly facilitates gradient-based optimization and supports sampling-based estimation. Moreover, in relaxation \ding{174}, we drop the term $H_{u}^{K}$ to reduce computational complexity. While retaining this term could potentially lead to improved performance, we empirically find that the gains are marginal, whereas the additional computational overhead is significant.

Furthermore, we can express the above upper bound in terms of the preference scores. Given the Heaviside step function $\delta(x) = \mathbb{I}(x \geq 0)$ \citep{yang2024psl}, recall that $\pi_{ui} = \sum_{j \in \mathcal{I}} \mathbb{I}(s_{uj} \geq s_{ui}) = \sum_{j \in \mathcal{I}} \delta(d_{uij})$, where $d_{uij} = s_{uj} - s_{ui}$, we can rewrite the upper bound (\ref{eq:slatk-derivation-5}) as:
\begin{equation} \label{eq:slatk-final}
    (\text{\ref{eq:slatk-derivation-5}}) = \sum_{i \in \mathcal{P}_u} \delta(s_{ui} - \beta_{u}^{K}) \cdot \log \left( \sum_{j \in \mathcal{I}} \delta(d_{uij}) \right) .
\end{equation}

\noindentpar{Smoothing Heaviside function.}
Note that \cref{eq:slatk-final} is still discontinuous due to the Heaviside step function $\delta(\cdot)$. To address this, following the conventional approach, we approximate $\delta(\cdot)$ by two continuous activation functions $\sigma_w(\cdot)$ and $\sigma_d(\cdot)$, resulting in the following recommendation loss --- \textbf{SoftmaxLoss@$K$ (SL@$K$)}:
\begin{equation} \label{eq:slatk}
    \mathcal{L}_{\text{SL@}K}(u) = 
    \sum_{i \in \mathcal{P}_u} \underbrace{\colorbox{pink!50}{$\sigma_w(s_{ui} - \beta_{u}^{K})$}}_{\textcolor{red}{\text{weight term:\ } w_{ui}}} \cdot \underbrace{\colorbox{skyblue!50}{$\log \left( \displaystyle\sum_{j \in \mathcal{I}} \sigma_d(d_{uij}) \right)$}}_{\textcolor{blue}{\text{SL term:\ } \mathcal{L}_{\text{SL}}(u, i)}} .
\end{equation}
To approximate the Heaviside step function $\delta(\cdot)$, two conventional activation functions are widely adopted --- the exponential function $e^{x / \tau_d}$ and the sigmoid function $1 / (1 + e^{-x / \tau_w})$, where $\tau_d$ and $\tau_w$ are temperature hyperparameters. The exponential function serves as an upper bound of $\delta(\cdot)$ and has been employed in SL, while the sigmoid function provides a tighter approximation of $\delta(\cdot)$ and has been utilized in BPR \citep{rendle2009bpr}. Here we select $\sigma_d$ as exponential and $\sigma_w$ as sigmoid in \cref{eq:slatk}. This configuration guarantees that SL@$K$ serves as a tight upper bound for $-\log$ DCG@$K$ (\cf \cref{thm:slatk} in \cref{subsec:slatk_properties}). In contrast, if both activations are chosen as sigmoid, the upper bound relation does not hold; if both are chosen as exponential, the bound is not as tight as in our setting. For detailed discussions, please refer to \cref{subapp:activation_functions}.

As shown in \cref{eq:slatk}, SL@$K$ can be interpreted as a specific \emph{weighted Softmax Loss}, where each positive interaction $(u, i)$ in SL (\cf \cref{eq:sl}) is assigned a \emph{quantile-based weight} $w_{ui}$. Intuitively, $w_{ui}$ serves to assign larger weights to positive instances with higher scores $s_{ui}$, emphasizing those ranked within the Top-$K$ positions during optimization (\ie those whose scores exceed the quantile). This aligns with the principle of Top-$K$ ranking metrics.

\subsubsection{\texorpdfstring{Top-$K$}{Top-K} Quantile Estimation} \label{subsubsec:quantile_estimation}

Now the question lies in how to estimate the Top-$K$ quantile $\beta_u^K$ efficiently and accurately. While quantile estimation \citep{koenker2005quantile,hao2007quantile,bickel2015mathematical} has been extensively studied in the field of statistics, these methods may not be appropriate in our scenarios. Given that the quantile evolves during training and the large scale of item set in RS, SL@$K$ places high demands on estimation efficiency. To address this, our work develops a simple Monte Carlo-based estimation strategy. Specifically, we randomly sample a small set of $N$ items for each user and estimate the Top-$K$ quantile among these sampled items. The computational complexity of this method is $O(|\mathcal{U}| N \log N)$, as it only requires sorting the sampled items, which significantly reduces the computational overhead compared to sorting the entire item set (\ie $O(|\mathcal{U}| |\mathcal{I}| \log |\mathcal{I}|)$).

\noindentpar{Theoretical guarantees.}
Despite its simplicity, our quantile estimation strategy has theoretical guarantees. To ensure rigor and facilitate generalization to the continuous case, we follow the conventional definition of the \emph{$p$-th quantile} \citep{bickel2015mathematical}. In the context of RS, the $p$-th quantile is exactly the Top-$(1 - p)|\mathcal{I}|$ quantile. We have:

\begin{restatable}[Monte Carlo quantile estimation]{theorem}{RestatableThmQuantileEstimation} \label{thm:quantile_estimation}
    Given the cumulative distribution function (\cdf) $F_u(s)$ of the preference scores $s_{ui}$ for user $u$, for any $p \in (0, 1)$, the $p$-th quantile is defined as $\theta^p_u \coloneqq F_u^{-1}(p) = \inf\{s : F_u(s) \geq p\}$. In Monte Carlo quantile estimation, we randomly sample $N$ preference scores $\{s_{uj}\}_{j = 1}^{N} \overset{\iid}{\sim} F_u(s)$. The estimated $p$-th quantile is defined as $\hat \theta_u^p \coloneqq \hat{F}_u^{-1}(p)$, where $\hat{F}_u(s) = \frac{1}{N}\sum_{j = 1}^{N} \mathbb{I}(s_{uj} \leq s)$ is the empirical \cdf of the sampled scores. Then, for any $\varepsilon > 0$, we have
    \begin{equation} \label{eq:quantile_estimation_error}
        \Pr\left( \left|\hat \theta_u^p - \theta_u^p\right| > \varepsilon \right) \leq 4e^{-2N\delta_{\varepsilon}^2} ,
    \end{equation}
    where $\delta_\varepsilon = \min\{F_u(\theta_u^p + \varepsilon) - p, p - F_u(\theta_u^p - \varepsilon)\}$. Specifically, in the discrete RS scenarios, the Top-$K$ quantile $\beta_{u}^{K}$ is exactly $\theta_u^{1 - K / |\mathcal{I}|}$.
\end{restatable}
The proof is provided in \cref{subapp:quantile_estimation_error_bound}. \cref{thm:quantile_estimation} provides the theoretical foundation for sampling-based quantile estimation --- the error between the estimated and ideal quantile is bounded by a function that decreases exponentially with the sample size $N$. This implies that the Top-$K$ quantile $\beta_{u}^{K}$ can be estimated to arbitrary precision given a sufficiently large $N$.

\begin{figure}[t]
    \centering
    \begin{subfigure}[b]{0.235\textwidth}
        \centering
        \includegraphics[width=\textwidth]{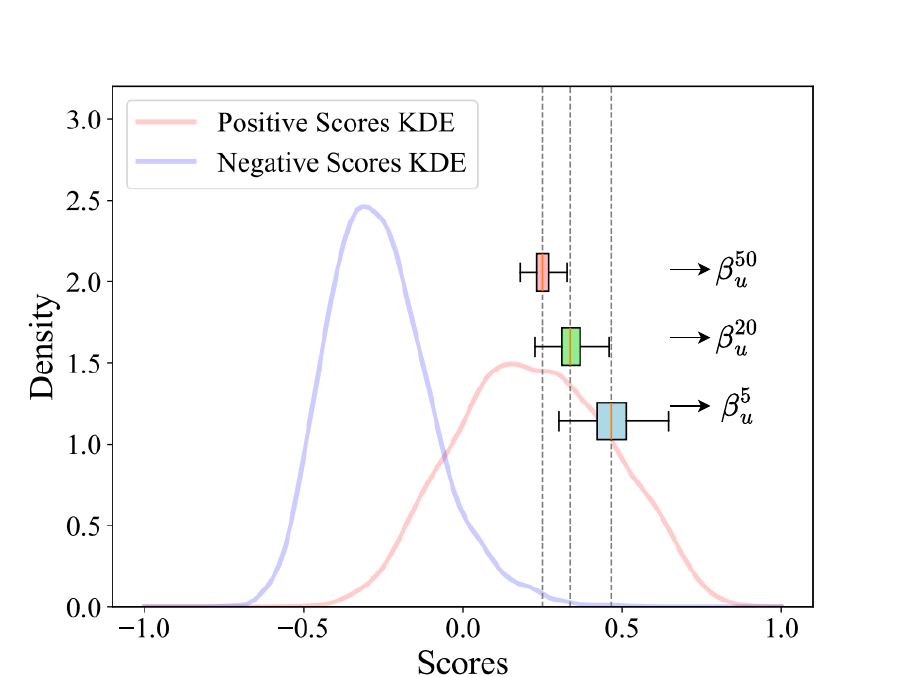}
        \caption{Quantile distribution.}
        \label{fig:quantile_distribution}
    \end{subfigure}
    \begin{subfigure}[b]{0.235\textwidth}
        \centering
        \includegraphics[width=\textwidth]{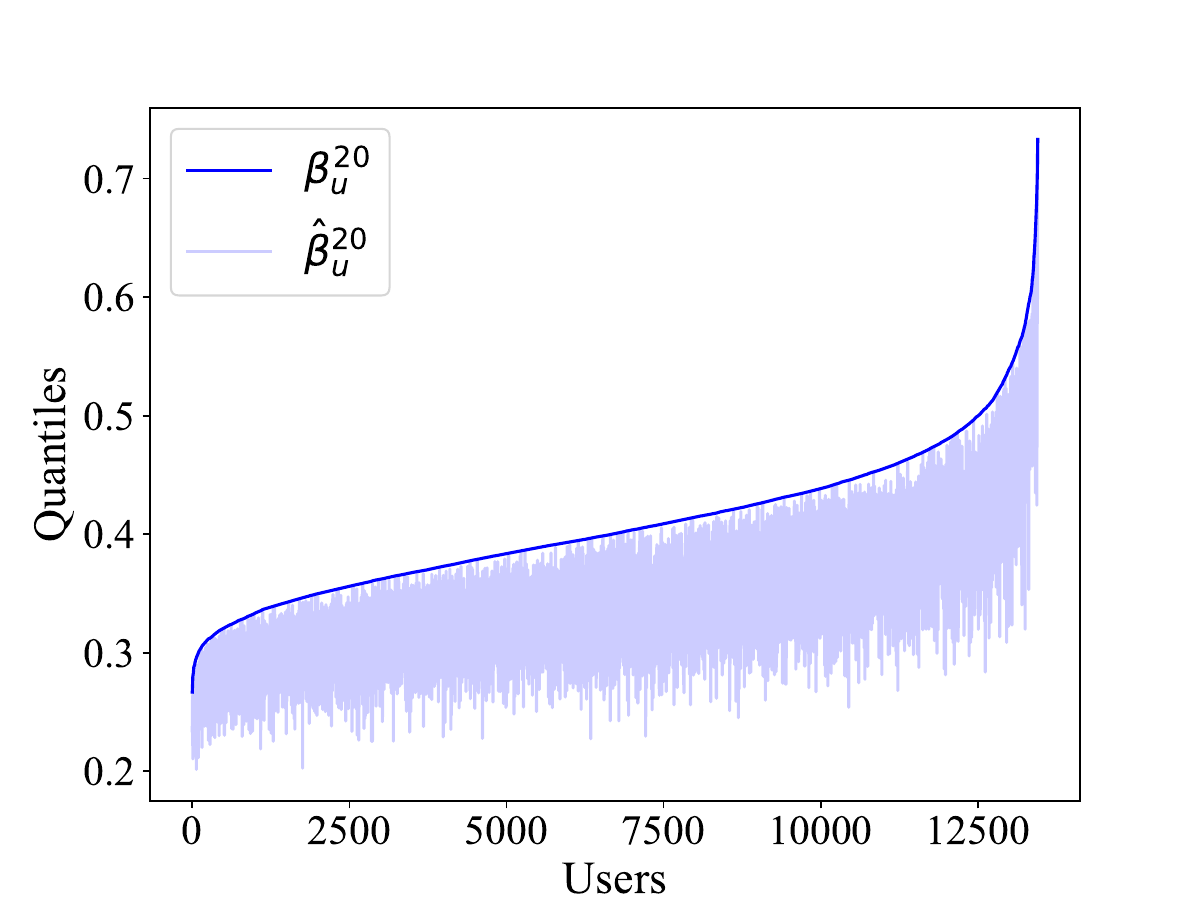}
        \caption{Quantile estimation.}
        \label{fig:sample_quantile_estimation}
    \end{subfigure}
    \caption{(a) Quantile distribution. The distributions of ideal quantiles $\beta_{u}^{20}$ and the positive/negative scores are illustrated using Kernel Density Estimation (KDE) \citep{parzen1962estimation}. (b) Quantile estimation. The estimated quantile $\hat \beta_{u}^{20}$ and ideal quantile $\beta_{u}^{20}$ are illustrated. The estimation error is 0.06 $\pm$ 0.03.}
    \Description{Quantile estimation.}
    \label{fig:sample_quantile}
\end{figure}

\noindentpar{Practical strategies.}
In practice, our Monte Carlo-based quantile estimation strategy can be further improved by leveraging the properties of RS. As shown in \cref{fig:quantile_distribution}, the scores of positive items are typically much higher than those of negative items, and the Top-$K$ quantile is often located within the range of positive item scores. Therefore, it is more effective to retain all positive instances and randomly sample a small set of negative instances for quantile estimation. This strategy, though simple, yields more accurate results. \cref{fig:sample_quantile_estimation} provides an example of estimated quantiles across users on the Electronic dataset, with a sample size of $N = 1000$. The estimated quantile $\hat \beta_{u}^{20}$ closely matches the ideal quantile $\beta_{u}^{20}$, with an average deviation of only 0.06. Further analyses and results can be found in \cref{app:sample_quantile_estimation}. The overall optimization process for SL@$K$ is also summarized in \cref{alg:slatk}.


\subsection{Analyses of \texorpdfstring{SL@$K$}{SL@K}} \label{subsec:slatk_properties}

\subsubsection{Properties of \texorpdfstring{SL@$K$}{SL@K}} \label{subsubsec:theoretical_properties_of_slatk}

Our proposed SL@$K$ offers several desirable properties (P), as summarized below:

\noindentpar{(P1) Theoretical guarantees.}
We establish a theoretical connection between SL@$K$ and NDCG@$K$ as follows:
\begin{restatable}[NDCG@$K$ surrogate]{theorem}{RestatableThmSlatk} \label{thm:slatk}
    For any user $u$, if the Top-$K$ hits $H_{u}^{K} > 1$\footnote{The assumption $H_{u}^{K} > 1$ is commonly satisfied in practice, as the training process tends to increase the scores of positive items, making them typically larger than those of negative items. \cref{subapp:proof_slatk} provides further empirical validation.}, SL@$K$ serves as an upper bound of $-\log \mathrm{DCG}@K$, \ie
    \begin{equation} \label{eq:slatk_bound}
        -\log \mathrm{DCG}@K(u) \leq \mathcal{L}_{\emph{\textrm{SL@}}K}(u) .
    \end{equation}
    When the Top-$K$ hits $H_{u}^{K} = 1$, a marginally looser yet effective bound holds, \ie $-\frac{1}{2} \log \mathrm{DCG}@K(u) \leq \mathcal{L}_{\emph{\textrm{SL@}}K}(u)$.
\end{restatable}
The proof is provided in \cref{subapp:proof_slatk}. \cref{thm:slatk} reveals that minimizing SL@$K$ leads to improved NDCG@$K$, ensuring the theoretical effectiveness of SL@$K$ in Top-$K$ recommendation.

\noindentpar{(P2) Ease of implementation.}
Compared to SL, SL@$K$ introduces only a quantile-based weight $w_{ui}$. Given the widespread adoption of SL in RS, SL@$K$ can be seamlessly integrated into existing recommendation frameworks with minimal modifications.

\noindentpar{(P3) Computational efficiency.}
The utilization of the Monte Carlo strategy for quantile estimation in SL@$K$ (\cf \cref{subsubsec:quantile_estimation}) ensures computational efficiency. The conventional SL has a time complexity of $O(|\mathcal{U}| \bar{P} N)$, where $\bar{P}$ denotes the average number of positive items per user, and $N$ denotes the sample size satisfying $N \ll |\mathcal{I}|$. Compared to SL, SL@$K$ only introduces an additional complexity of $O(|\mathcal{U}| N \log N)$ for quantile estimation, which is typically negligible in practice (\cf \cref{fig:execution_time_comparison}).

\noindentpar{(P4) Gradient stability.}
SL@$K$ exhibits a moderate gradient distribution comparable to that of SL (\cf \cref{fig:gradient_distribution_comparison}), which contributes to its training stability and data utilization effectiveness. This property is mainly attributed to the bounded weight $w_{ui} \in (0.1, 1)$ with sigmoid temperature $\tau_w \geq 1$, thus not significantly amplifying gradient variance. In contrast, other NDCG@$K$ surrogate losses, including LambdaLoss@$K$ \citep{jagerman2022optimizing} and SONG@$K$ \citep{qiu2022large}, are usually hindered by the excessively long-tailed gradients (\cf \cref{fig:gradient_distribution_comparison}).

\noindentpar{(P5) Noise robustness.}
\emph{False positive noise} \citep{chen2023bias} is prevalent in RS, arising from various factors such as clickbait \citep{wang2021clicks}, item position bias \citep{hofmann2014effects}, or accidental interactions \citep{adamopoulos2014unexpectedness}. Recent studies have shown that such noise can significantly mislead model training and degrade performance \citep{wen2019leveraging}. Interestingly, the introduction of weight $w_{ui}$ in SL@$K$ helps mitigate this issue. In fact, the false positives, which often resemble negative instances, tend to have lower preference scores $s_{ui}$ than the true positives. As a result, these noisy instances typically receive smaller weights $w_{ui}$ (which are positively correlated with $s_{ui}$) and contribute less in model training. This enhances the model's robustness against false positive noise, as demonstrated in the gradient analysis in \cref{subapp:gradient_analysis}.

\subsubsection{Comparison with Existing Losses} \label{subsubsec:comparison_with_existing_losses}

In this subsection, we delve into the connections and differences between SL@$K$ and other closely related losses to provide further insights:

\noindentpar{SL@$K$ \versus Softmax Loss (SL).}
As discussed in \cref{subsubsec:smooth_surrogate_for_ndcgatk}, SL@$K$ can be viewed as a specific weighted SL \citep{wu2024effectiveness}. Although SL demonstrates theoretical advantages due to its close connection with NDCG, as well as practical benefits such as concise formulation and computational efficiency, it does not account for the Top-$K$ truncation. Our SL@$K$ bridges this gap by accompanying each term of SL with a quantile-based weight. As such, SL@$K$ inherits the advantages of SL, while introducing additional merits, \eg theoretical connections to NDCG@$K$ and robustness to false positive noise.

\noindentpar{SL@$K$ \versus LambdaLoss@$K$ and SONG@$K$.}
LambdaLoss@$K$ \citep{jagerman2022optimizing} and SONG@$K$ \citep{qiu2022large} take into account the Top-$K$ truncation and have shown promising results in other fields like document retrieval \citep{liu2009learning}. However, we find that their effectiveness in RS is compromised, particularly given the large item space and sparse interactions. Specifically, both of them suffer from the issue of long-tailed gradients due to their inherent design. The gradients are dominated by a few instances, while the majority of instances have negligible contributions, which may lead to data utilization inefficiency and optimization instability. In contrast, SL@$K$ exhibits moderate gradients by leveraging the quantile technique and appropriate relaxations, which addresses these issues and achieves superior performance (\cf \cref{fig:gradient_distribution_comparison}). 

Beyond the gradient instability, LambdaLoss@$K$ also faces additional challenges on computational efficiency. Specifically, it requires calculating exact rankings, which is computationally impractical in RS. Even worse, the skewed gradient distribution hinders the sampling-based strategy to reduce computational overhead, since the gradients of sampled instances may be either vanishingly small or excessively large, leading to unstable optimization. In contrast, SL@$K$ is both theoretically sound and computationally efficient, making it a more suitable choice for RS. 

Notably, while SONG@$K$ also employs a threshold to tackle Top-$K$ truncation similar to SL@$K$'s quantile, SL@$K$ differs significantly from SONG@$K$ in two aspects: (i) we follow the successful approaches of SL \citep{wu2024effectiveness} and PSL \citep{yang2024psl} to simplify and smooth NDCG@$K$, which facilitates sampling-based estimation and optimization, while SONG@$K$ employs a compositional optimization technique, which may not be effective in RS. (ii) we employ a simple sampling-based strategy to estimate the threshold (quantile) with theoretical guarantees, as opposed to the complex bilevel optimization in SONG@$K$. These differences contribute to the significant superiority of SL@$K$ over SONG@$K$ in terms of both recommendation performance and practical applicability. \cref{app:gradient_vanishing} gives a detailed discussion on these two losses.



\begin{table}[t]
    \centering
    \caption{Dataset statistics. Refer to \cref{subapp:datasets} for details.}
    \label{tab:datasets-statistics-maintext}
    \begin{tabular}{l|rrrr}
    \Xhline{1.2pt}
    \multicolumn{1}{c|}{\textbf{Dataset}} & \textbf{\#Users} & \textbf{\#Items} & \textbf{\#Interactions} & \textbf{Density} \bigstrut\\
    \Xhline{1pt}
    Health & 1,974  & 1,200  & 48,189  & 0.02034 \bigstrut[t]\\
    Electronic & 13,455  & 8,360  & 234,521  & 0.00208 \\
    Gowalla & 29,858  & 40,988  & 1,027,464  & 0.00084 \\
    Book  & 135,109  & 115,172  & 4,042,382  & 0.00026 \\
    \Xhline{1.2pt}
    \end{tabular}
\end{table}

\begin{table*}[t]
    \centering
    \caption{Top-20 recommendation performance comparison of SL@$K$ with existing losses. The best results are highlighted in bold, and the best baselines are underlined. "\textcolor{red}{Imp.}" denotes the improvement of SL@$K$ over the best baseline.}
    \label{tab:slatk_versus_sota}
    \small
    \begin{tabularx}{0.98\textwidth}{c|l|YY|YY|YY|YY}
        \Xhline{1.2pt}
        \multirow{2}[4]{*}{\textbf{Backbone}} & \multicolumn{1}{c|}{\multirow{2}[4]{*}{\textbf{Loss}}} & \multicolumn{2}{c|}{\textbf{Health}} & \multicolumn{2}{c|}{\textbf{Electronic}} & \multicolumn{2}{c|}{\textbf{Gowalla}} & \multicolumn{2}{c}{\textbf{Book}} \bigstrut\\
        \cline{3-10}    &   & \textbf{Recall@20} & \textbf{NDCG@20} & \textbf{Recall@20} & \textbf{NDCG@20} & \textbf{Recall@20} & \textbf{NDCG@20} & \textbf{Recall@20} & \textbf{NDCG@20} \bigstrut\\
        \Xhline{1.0pt}
        \multirow{10}[6]{*}{MF}
            & BPR               & 0.1627  & 0.1234  & 0.0816  & 0.0527  & 0.1355  & 0.1111  & 0.0665  & 0.0453  \bigstrut[t]\\
            & GuidedRec         & 0.1568  & 0.1093  & 0.0644  & 0.0385  & 0.1135  & 0.0863  & 0.0518  & 0.0361  \\
            & SONG@20           & 0.0874  & 0.0650  & 0.0708  & 0.0444  & 0.1237  & 0.0970  & 0.0747  & 0.0542  \\
            & LLPAUC            & 0.1644  & 0.1209  & 0.0821  & 0.0499  & 0.1610  & 0.1189  & 0.1150  & 0.0811  \\
            & SL                & \uline{0.1719}  & 0.1261  & 0.0821  & 0.0529  & 0.2064  & 0.1624  & 0.1559  & 0.1210  \\
            & AdvInfoNCE        & 0.1659  & 0.1237  & 0.0829  & 0.0527  & 0.2067  & 0.1627  & 0.1557  & 0.1172  \\
            & BSL               & \uline{0.1719}  & 0.1261  & 0.0834  & 0.0530  & 0.2071  & 0.1630  & 0.1563  & 0.1212 \\
            & PSL               & 0.1718  & \uline{0.1268}  & \uline{0.0838}  & \uline{0.0541}  & \uline{0.2089}  & \uline{0.1647}  & \uline{0.1569}  & \uline{0.1227}  \\
            & \textbf{SL@20 (Ours)}    & \textbf{0.1823} & \textbf{0.1390} & \textbf{0.0901} & \textbf{0.0590} & \textbf{0.2121} & \textbf{0.1709} & \textbf{0.1612} & \textbf{0.1269} \bigstrut[b]\\
            \cline{2-10}    & \textcolor{red}{\textbf{Imp. \%}} & \textcolor{red}{\textbf{+6.05\%}} & \textcolor{red}{\textbf{+9.62\%}} & \textcolor{red}{\textbf{+7.52\%}} & \textcolor{red}{\textbf{+9.06\%}} & \textcolor{red}{\textbf{+1.53\%}} & \textcolor{red}{\textbf{+3.76\%}} & \textcolor{red}{\textbf{+2.74\%}} & \textcolor{red}{\textbf{+3.42\%}} \bigstrut\\
        \Xhline{1.0pt}
        \multirow{10}[6]{*}{LightGCN}
            & BPR               & 0.1618  & 0.1203  & 0.0813  & 0.0524  & 0.1745  & 0.1402  & 0.0984  & 0.0678  \bigstrut[t]\\
            & GuidedRec         & 0.1550  & 0.1073  & 0.0657  & 0.0393  & 0.0921  & 0.0686  & 0.0468  & 0.0310  \\
            & SONG@20           & 0.1353  & 0.0960  & 0.0816  & 0.0511  & 0.1261  & 0.0968  & 0.0820  & 0.0573  \\
            & LLPAUC            & 0.1685  & 0.1207  & \uline{0.0831}  & 0.0507  & 0.1616  & 0.1192  & 0.1147  & 0.0810  \\
            & SL                & 0.1691  & 0.1235  & 0.0823  & 0.0526  & 0.2068  & 0.1628  & 0.1567  & 0.1220  \\
            & AdvInfoNCE        & \uline{0.1706}  & 0.1264  & 0.0823  & 0.0528  & 0.2066  & 0.1625  & 0.1568  & 0.1177  \\
            & BSL               & 0.1691  & 0.1236  & 0.0823  & 0.0526  & 0.2069  & 0.1628  & 0.1568  & 0.1220  \\
            & PSL               & 0.1701  & \uline{0.1270}  & 0.0830  & \uline{0.0536}  & \uline{0.2086}  & \uline{0.1648}  & \uline{0.1575}  & \uline{0.1233}  \\
            & \textbf{SL@20 (Ours)}    & \textbf{0.1783} & \textbf{0.1371} & \textbf{0.0903} & \textbf{0.0591} & \textbf{0.2128} & \textbf{0.1729} & \textbf{0.1625} & \textbf{0.1280} \bigstrut[b]\\
            \cline{2-10}            & \textcolor{red}{\textbf{Imp. \%}} & \textcolor{red}{\textbf{+4.51\%}} & \textcolor{red}{\textbf{+7.95\%}} & \textcolor{red}{\textbf{+8.66\%}} & \textcolor{red}{\textbf{+10.26\%}} & \textcolor{red}{\textbf{+2.01\%}} & \textcolor{red}{\textbf{+4.92\%}} & \textcolor{red}{\textbf{+3.17\%}} & \textcolor{red}{\textbf{+3.81\%}} \bigstrut\\
        \Xhline{1.0pt}
        \multirow{10}[6]{*}{XSimGCL}
            & BPR               & 0.1496  & 0.1108  & 0.0777  & \uline{0.0508}  & 0.1966  & 0.1570  & 0.1269  & 0.0905  \bigstrut[t]\\
            & GuidedRec         & 0.1539  & 0.1088  & 0.0760  & 0.0473  & 0.1685  & 0.1277  & 0.1275  & 0.0951  \\
            & SONG@20           & 0.1378  & 0.0948  & 0.0525  & 0.0320  & 0.1367  & 0.0985  & 0.1281  & 0.0964  \\
            & LLPAUC            & 0.1519  & 0.1083  & 0.0781  & 0.0481  & 0.1632  & 0.1200  & 0.1363  & 0.1008  \\
            & SL                & 0.1534  & 0.1113  & 0.0772  & 0.0490  & 0.2005  & 0.1570  & 0.1549  & 0.1207  \\
            & AdvInfoNCE        & 0.1499  & 0.1072  & 0.0776  & 0.0489  & 0.2010  & 0.1564  & 0.1568  & 0.1179  \\
            & BSL               & \uline{0.1649}  & \uline{0.1201}  & 0.0800  & 0.0507  & \uline{0.2037}  & \uline{0.1597}  & 0.1550  & 0.1207  \\
            & PSL               & 0.1579  & 0.1143  & \uline{0.0801}  & 0.0507  & \uline{0.2037}  & 0.1593  & \uline{0.1571}  & \uline{0.1228}  \\
            & \textbf{SL@20 (Ours)}    & \textbf{0.1753} & \textbf{0.1332} & \textbf{0.0869} & \textbf{0.0571} & \textbf{0.2095} & \textbf{0.1717} & \textbf{0.1624} & \textbf{0.1277} \bigstrut[b]\\
            \cline{2-10}            & \textcolor{red}{\textbf{Imp. \%}} & \textcolor{red}{\textbf{+6.31\%}} & \textcolor{red}{\textbf{+10.91\%}} & \textcolor{red}{\textbf{+8.49\%}} & \textcolor{red}{\textbf{+12.40\%}} & \textcolor{red}{\textbf{+2.85\%}} & \textcolor{red}{\textbf{+7.51\%}} & \textcolor{red}{\textbf{+3.37\%}} & \textcolor{red}{\textbf{+3.99\%}} \bigstrut\\
        \Xhline{1.2pt}
    \end{tabularx}
\end{table*}

\section{Experiments} \label{sec:experiments}

We aim to answer the following research questions (RQs):

\begin{itemize}[topsep=3pt,leftmargin=10pt,itemsep=0pt]
    \item \textbf{RQ1}: How does SL@$K$ perform compared with existing losses?
    \item \textbf{RQ2}: Does SL@$K$ exhibit consistent improvements across different NDCG@$K$ metrics with varying $K$?
    \item \textbf{RQ3}: Does SL@$K$ exhibit robustness against false positive noise?
    \item \textbf{RQ4}: Can SL@$K$ be effectively applied to other information retrieval (IR) tasks?
\end{itemize}

\begin{table*}[t]
    \centering
    \caption{NDCG@$K$ (D@$K$) comparisons with varying $K$ on Health and Electronic datasets and MF backbone. The best results are highlighted in bold, and the best baselines are underlined. "\textcolor{red}{Imp.}" denotes the improvement of SL@$K$ over the best baseline.}
    \label{tab:slatk_versus_sota_mf}
    \small
    \setlength{\tabcolsep}{4.0pt}
    \begin{tabular}{l|cccccc|cccccc}
        \Xhline{1.2pt}
        \multicolumn{1}{c|}{\multirow{2}[4]{*}{\textbf{Method}}} & \multicolumn{6}{c|}{\textbf{Health}} & \multicolumn{6}{c}{\textbf{Electronic}} \bigstrut\\
        \cline{2-13}
        & \textbf{D@5}  & \textbf{D@10}  & \textbf{D@20}  & \textbf{D@50}  & \textbf{D@75}  & \textbf{D@100}  
        & \textbf{D@5}  & \textbf{D@10}  & \textbf{D@20}  & \textbf{D@50}  & \textbf{D@75}  & \textbf{D@100} \bigstrut\\
        \Xhline{1.0pt}
        BPR             & \uline{0.0940}  & 0.1037  & 0.1234  & \uline{0.1621}  & \uline{0.1804}  & \uline{0.1925}  
                        & 0.0345  & 0.0419  & 0.0527  & 0.0690  & 0.0777  & 0.0845 \bigstrut[t]\\
        GuidedRec       & 0.0769  & 0.0881  & 0.1093  & 0.1484  & 0.1671  & 0.1811  
                        & 0.0228  & 0.0294  & 0.0385  & 0.0551  & 0.0635  & 0.0703 \\
        SONG            & 0.0353  & 0.0392  & 0.0488  & 0.0709  & 0.0834  & 0.0930  
                        & 0.0316  & 0.0393  & 0.0493  & 0.0661  & 0.0744  & 0.0803 \\
        SONG@$K$        & 0.0503  & 0.0535  & 0.0650  & 0.0896  & 0.1037  & 0.1135  
                        & 0.0276  & 0.0349  & 0.0444  & 0.0581  & 0.0651  & 0.0706 \\
        LLPAUC          & 0.0887  & 0.0996  & 0.1209  & 0.1592  & 0.1765  & 0.1892  
                        & 0.0305  & 0.0388  & 0.0499  & 0.0686  & 0.0778  & \uline{0.0848} \\
        SL              & 0.0922  & 0.1037  & 0.1261  & 0.1620  & 0.1791  & 0.1924  
                        & 0.0353  & 0.0430  & 0.0529  & 0.0696  & 0.0783  & 0.0845 \\
        AdvInfoNCE      & 0.0926  & 0.1038  & 0.1237  & 0.1608  & 0.1789  & 0.1920  
                        & 0.0341  & 0.0423  & 0.0527  & 0.0697  & 0.0782  & 0.0843 \\
        BSL             & 0.0922  & 0.1037  & 0.1261  & 0.1620  & 0.1791  & 0.1924  
                        & 0.0344  & 0.0425  & 0.0530  & 0.0691  & 0.0776  & 0.0843 \\
        PSL             & \uline{0.0940}  & \uline{0.1048}  & \uline{0.1268}  & 0.1613  & 0.1789  & 0.1912  
                        & \uline{0.0356}  & \uline{0.0434}  & \uline{0.0541}  & \uline{0.0700}  & \uline{0.0784}  & 0.0845 \\
        \textbf{SL@$K$ (Ours)} & \textbf{0.1080} & \textbf{0.1190} & \textbf{0.1390} & \textbf{0.1736} & \textbf{0.1916} & \textbf{0.2035}  
                                & \textbf{0.0402} & \textbf{0.0484} & \textbf{0.0590} & \textbf{0.0760} & \textbf{0.0844} & \textbf{0.0908} \bigstrut[b]\\
        \hline
        \textcolor{red}{\textbf{Imp. \%}} & \textcolor{red}{\textbf{+14.89\%}} & \textcolor{red}{\textbf{+13.55\%}} & \textcolor{red}{\textbf{+9.62\%}} & \textcolor{red}{\textbf{+7.09\%}} & \textcolor{red}{\textbf{+6.21\%}} & \textcolor{red}{\textbf{+5.71\%}}  
                                           & \textcolor{red}{\textbf{+12.92\%}} & \textcolor{red}{\textbf{+11.52\%}} & \textcolor{red}{\textbf{+9.06\%}} & \textcolor{red}{\textbf{+8.57\%}} & \textcolor{red}{\textbf{+7.65\%}} & \textcolor{red}{\textbf{+7.08\%}} \bigstrut\\
        \Xhline{1.2pt}
    \end{tabular}
\end{table*}

\begin{table*}[t]
    \centering
    \caption{Performance exploration of SL@$K$ on NDCG@$K'$ with varying $K$ and $K'$. The best results are highlighted in bold.}
    \label{tab:consistency_atk}
    \small
    \setlength{\tabcolsep}{4.0pt}
    \begin{tabular}{l|cccccc|cccccc}
        \Xhline{1.2pt}
        \multicolumn{1}{l|}{\multirow{2}[4]{*}{\textbf{SL@$K$}}} & \multicolumn{6}{c|}{\textbf{Health}} & \multicolumn{6}{c}{\textbf{Electronic}} \bigstrut\\
        \cline{2-13}
        & \textbf{D@5} & \textbf{D@10} & \textbf{D@20} & \textbf{D@50} & \textbf{D@75} & \textbf{D@100} & \textbf{D@5} & \textbf{D@10} & \textbf{D@20} & \textbf{D@50} & \textbf{D@75} & \textbf{D@100} \bigstrut\\
        \Xhline{1.0pt}
        SL@5   & \cellcolor{skyblue!50}\textbf{0.1080} & 0.1180 & 0.1379 & 0.1724 & 0.1906 & 0.2032 & \cellcolor{skyblue!50}\textbf{0.0402} & 0.0480 & 0.0583 & 0.0753 & 0.0839 & 0.0900 \bigstrut\\
        SL@10  & 0.1077 & \cellcolor{skyblue!50}\textbf{0.1190} & 0.1377 & 0.1734 & 0.1909 & 0.2028 & 0.0400 & \cellcolor{skyblue!50}\textbf{0.0484} & 0.0583 & 0.0755 & 0.0839 & 0.0901 \bigstrut\\
        SL@20  & 0.1076 & 0.1188 & \cellcolor{skyblue!50}\textbf{0.1390} & 0.1733 & 0.1909 & 0.2029 & 0.0400 & 0.0483 & \cellcolor{skyblue!50}\textbf{0.0590} & 0.0759 & 0.0837 & 0.0900 \bigstrut\\
        SL@50  & 0.1062 & 0.1167 & 0.1364 & \cellcolor{skyblue!50}\textbf{0.1736} & 0.1901 & 0.2020 & 0.0398 & 0.0481 & 0.0587 & \cellcolor{skyblue!50}\textbf{0.0760} & 0.0842 & 0.0907 \bigstrut\\
        SL@75  & 0.1073 & 0.1179 & 0.1387 & 0.1734 & \cellcolor{skyblue!50}\textbf{0.1916} & 0.2031 & 0.0397 & 0.0481 & 0.0587 & 0.0759 & \cellcolor{skyblue!50}\textbf{0.0844} & 0.0907 \bigstrut\\
        SL@100 & 0.1071 & 0.1177 & 0.1375 & 0.1727 & 0.1904 & \cellcolor{skyblue!50}\textbf{0.2035} & 0.0399 & 0.0481 & 0.0587 & 0.0759 & 0.0843 & \cellcolor{skyblue!50}\textbf{0.0908} \bigstrut\\
        \hline
        SL (@$\infty$) & \cellcolor{gray!20}0.0922 & \cellcolor{gray!20}0.1037 & \cellcolor{gray!20}0.1261 & \cellcolor{gray!20}0.1620 & \cellcolor{gray!20}0.1791 & \cellcolor{gray!20}0.1924 & \cellcolor{gray!20}0.0353 & \cellcolor{gray!20}0.0430 & \cellcolor{gray!20}0.0529 & \cellcolor{gray!20}0.0696 & \cellcolor{gray!20}0.0783 & \cellcolor{gray!20}0.0845 \bigstrut\\
        \Xhline{1.2pt}
    \end{tabular}
\end{table*}

\subsection{Experimental Setup} \label{subsec:experimental_setup}


\noindentpar{Datasets.}
To ensure fair comparisons, our experimental setup closely follows the prior work of \citet{wu2023bsl} and \citet{yang2024psl}. We conduct experiments on four widely-used datasets: Health \citep{he2016ups,mcauley2015image}, Electronic \citep{he2016ups,mcauley2015image}, Gowalla \citep{cho2011friendship}, and Book \citep{he2016ups,mcauley2015image}. Additionally, given the inefficiency of LambdaLoss@$K$ \citep{jagerman2022optimizing} in handling these large datasets, we further evaluate its performance on two relatively smaller datasets, \ie MovieLens \citep{harper2015movielens} and Food \citep{majumder2019generating}. Detailed dataset descriptions can be found in \cref{tab:datasets-statistics-maintext,subapp:datasets}.

\noindentpar{Recommendation backbones.}
Following the settings in \citet{yang2024psl}, we evaluate the proposed losses on three backbones: \uline{MF} \citep{koren2009matrix} (classic Matrix Factorization model), \uline{LightGCN} \citep{he2020lightgcn} (SOTA graph-based model), and \uline{XSimGCL} \citep{yu2023xsimgcl} (SOTA contrastive-based model). The implementation details can be found in \cref{subapp:recommendation_backbones}.

\noindentpar{Baseline losses.}
We compare SL@$K$ with the following baselines: (i) Pairwise loss (\uline{BPR} \citep{rendle2009bpr}); (ii) NDCG surrogate losses (\uline{GuidedRec} \citep{rashed2021guided} and \uline{Softmax Loss (SL)} \citep{wu2024effectiveness}); (iii) NDCG@$K$ surrogate losses (\uline{LambdaLoss@$K$} \citep{jagerman2022optimizing} and \uline{SONG@$K$} \citep{qiu2022large}); (iv) Partial AUC surrogate loss (\uline{LLPAUC} \citep{shi2024lower}); (v) Advanced SL-based losses (\uline{AdvInfoNCE} \citep{zhang2024empowering}, \uline{BSL} \citep{wu2023bsl}, and \uline{PSL} \citep{yang2024psl}). Refer to \cref{subapp:compared_methods_hyperparameters} for details.

\noindentpar{Hyperparameter settings.}
For fair comparisons, SL@$K$ adopts the same temperature parameter $\tau_d$ as the optimal $\tau$ in SL. SL@$K$ also uses the same negative sampling strategy as SL for both training and quantile estimation with sample size $N = 1000$. For all baselines, we follow the hyperparameter settings provided in original papers and further tune them to achieve the best performance. We provide the details in \cref{subapp:compared_methods_hyperparameters}, optimal hyperparameters in \cref{subapp:optimal_hyperparameters}, and supplementary results in \cref{app:supplementary_experimental_results}.


\noindentpar{Information Retrieval Tasks.}
To extend SL@$K$ to other fields, we adapt it to three different IR tasks: (i) \textbf{Learning to rank (LTR)}, aiming to order a list of candidate items according to their relevance to a given query; (ii) \textbf{Sequential recommendation (SeqRec)}, focusing on next item prediction in a user's interaction sequence; and (iii) \textbf{Link prediction (LP)}, predicting links between two nodes in a graph. We closely follow the experimental settings in prior work \citep{pobrotyn2021neuralndcg, kang2018self, li2023evaluating} and provide the details in \cref{subapp:experimental_setup_ir}.


\subsection{Performance Comparison} \label{subsec:performance_comparison}


\noindentpar{SL@$K$ \versus Baselines (RQ1).}
\cref{tab:slatk_versus_sota} presents the performance comparison of SL@$K$ against existing losses. As shown, SL@$K$ consistently outperforms all competing losses across various datasets and backbones. The improvements are substantial, with an average increase of 6.03\% over the best baselines. This improvement can be attributed to the closer alignment of SL@$K$ with NDCG@$K$, highlighting the importance of explicitly modeling Top-$K$ truncation during optimization, as opposed to NDCG surrogate losses. Notably, SL@$K$ also demonstrates strong performance on Recall@$K$. This is because optimizing NDCG@$K$ naturally increases the positive hits in Top-$K$ positions, thereby enhancing Recall@$K$ performance. 

\noindentpar{SL@$K$ \versus NDCG@$K$ surrogate losses (RQ1).}
We further compare SL@$K$ with existing NDCG@$K$ surrogate losses, \ie SONG@$K$ and LambdaLoss@$K$, in \cref{tab:slatk_versus_sota,tab:slatk_versus_lambda}. Although these losses are also designed to optimize NDCG@$K$, our experiments show that SL@$K$ consistently outperforms them, with significant improvements of over 70\% and 13\% in NDCG@20 compared to SONG@$K$ and LambdaLoss@$K$, respectively. The unsatisfactory performance of these surrogate losses can be attributed to their unstable and ineffective optimization process, as discussed in \cref{subsubsec:comparison_with_existing_losses}. Moreover, LambdaLoss@$K$ incurs significantly higher computational costs compared to SL@$K$. While sampling strategies could be employed to accelerate LambdaLoss@$K$ (\ie LambdaLoss@$K$-S in \cref{tab:slatk_versus_lambda}), they lead to substantial performance degradation (over 30\%).



\noindentpar{NDCG@$K$ performance with varying $K$ (RQ2).}  
\cref{tab:slatk_versus_sota_mf} illustrates the NDCG@$K$ performance across different values of $K$. Experimental results show that SL@$K$ consistently outperforms the baseline methods in all NDCG@$K$ metrics. We also observe that as $K$ increases, the magnitude of the improvements decreases, which aligns with our intuition. Specifically, the Top-$K$ truncation has a greater impact when $K$ is small. As $K$ increases, NDCG@$K$ degrades to the full-ranking metric NDCG. Consequently, the advantage of optimizing for NDCG@$K$ diminishes as $K$ grows.

\noindentpar{Top-$K$ recommendation consistency (RQ2).} \cref{tab:consistency_atk} presents the performance of NDCG@$K'$ for SL@$K$ with varying values of $K, K'$ in $\{5, 10, 20, 50, 75, 100\}$. We observe that the best NDCG@$K'$ performance is always achieved when $K' = K$ in SL@$K$. This consistency aligns with our theoretical analysis in \cref{subsubsec:theoretical_properties_of_slatk}, \ie SL@$K$ is oriented towards optimizing NDCG@$K$ rather than other NDCG@$K'$ when $K \neq K'$. For instance, SL@20 achieves the best NDCG@20 performance, but its performance on NDCG@50 is lower compared to SL@50. Nonetheless, SL@$K$ always outperforms SL(@$\infty$), emphasizing the effectiveness of SL@$K$ in real-world RS.



\noindentpar{Noise Robustness (RQ3).}
In \cref{fig:slatk_noise}, we assess the robustness of SL@$K$ to false positive instances. Following \citet{wu2023bsl}, we manually introduce a certain ratio of negative instances as noisy positive instances during training. As the noise ratio increases, SL@$K$ demonstrates greater improvements over SL (up to 24\%), indicating superior robustness to false positive noise. This finding is consistent with our analysis in \cref{subsubsec:theoretical_properties_of_slatk}.


\noindentpar{Application to other IR tasks (RQ4).} We adapt SL@$K$ to three different IR tasks: LTR (\cref{tab:ltr-results}), SeqRec (\cref{tab:seqrec-results}), and LP (\cref{tab:lp-results}). Results show that SL@$K$ consistently outperforms baseline ranking losses (\eg LambdaLoss@$K$ \citep{jagerman2022optimizing} and NeuralNDCG \citep{pobrotyn2021neuralndcg}) and classification losses (\eg BCE \citep{kang2018self} and SL \citep{wu2024effectiveness}) across all tasks, demonstrating its versatility for general IR applications.


\begin{figure}[t]
    \centering
    \includegraphics[width=0.47\textwidth]{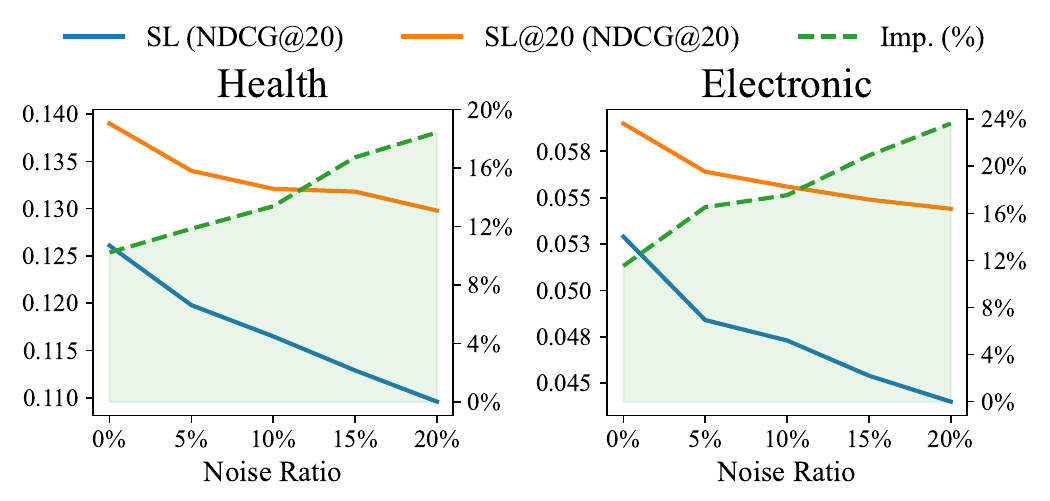}
    \caption{NDCG@20 performance of SL@$K$ compared with SL under varying ratios of imposed false positive instances.}
    \label{fig:slatk_noise}
    \Description{Noise robustness of SL@$K$.}
\end{figure}

\begin{table}[t]
    \centering
    \caption{LTR results on WEB10K, WEB30K \citep{DBLP:journals/corr/QinL13}, and Istella \citep{dato2016fast} datasets (metrics: NDCG@5).}
    \label{tab:ltr-results}
    \begin{tabular}{l|ccc}
    \Xhline{1.2pt}
    \multicolumn{1}{c|}{\textbf{Loss}} & 
    \textbf{WEB10K} & 
    \textbf{WEB30K} & 
    \textbf{Istella} \bigstrut\\
    \Xhline{1pt}
    ListMLE \citep{xia2008listwise} & 0.4145 & 0.4433 & 0.5671 \\
    ListNet \cite{cao2007learning} & 0.4225 & \uline{0.4594} & 0.6290 \bigstrut[t]\\
    RankNet \citep{burges2005learning} & 0.4253 & 0.4426 & 0.6189 \\
    LambdaLoss@5 \citep{jagerman2022optimizing} & 0.4320 & 0.4496 & 0.5860 \\
    NeuralNDCG \citep{pobrotyn2021neuralndcg} & \uline{0.4338} & 0.4524 & 0.5823 \\
    SL \citep{wu2024effectiveness} & 0.4310 & 0.4552 & \uline{0.6327} \\
    \textbf{SL@5 (Ours)} & \textbf{0.4633} & \textbf{0.4895} & \textbf{0.6412} \bigstrut[b]\\
    \hline
    \textcolor{red}{\textbf{Imp. \%}}   & \textcolor{red}{\textbf{+6.80\%}} & \textcolor{red}{\textbf{+6.55\%}} & \textcolor{red}{\textbf{+1.34\%}} \bigstrut\\
    \Xhline{1.2pt}
    \end{tabular}
\end{table}


\section{Related Work} \label{sec:related_work}

\noindentpar{Recommendation losses.}
Recommendation losses play a vital role in recommendation models optimization. The earliest works treat recommendation as a simple regression or binary classification problem, utilizing pointwise losses such as MSE \citep{he2017nfm} and BCE \citep{he2017ncf}. However, due to neglecting the ranking essence in RS, these pointwise losses usually result in inferior performance. To address this, pairwise losses such as BPR \citep{rendle2009bpr,lin2025recommendation} have been proposed. BPR aims to learn a partial order among items and serves as a surrogate for AUC. Following BPR, Softmax Loss (SL) \citep{wu2024effectiveness} extends the pairwise ranking to listwise by introducing the Plackett-Luce models \citep{luce1959individual,plackett1975analysis} or contrastive learning principles \citep{oord2018representation,chen2020simple}. SL has been proven to be an NDCG surrogate and achieves SOTA performance \citep{bruch2019analysis,yang2024psl}.

Recent works have further improved ranking losses from various approaches. For example, robustness enhancements to SL have been explored via Distributionally Robust Optimization (DRO) \citep{shapiro2017distributionally}, as seen in AdvInfoNCE \citep{zhang2024empowering}, BSL \citep{wu2023bsl} and PSL \citep{yang2024psl}. Other approaches directly optimize NDCG, including LambdaRank \citep{burges2006learning}, LambdaLoss \citep{wang2018lambdaloss}, SONG \citep{qiu2022large}, and PSL \citep{yang2024psl}. There are also works focusing on alternative surrogate approaches for NDCG, including ranking-based \citep{chapelle2010gradient}, Gumbel-based \citep{grover2019stochastic}, and neural-based \citep{rashed2021guided} methods.

Despite recent advancements, most ranking losses struggle in practical Top-$K$ recommendation, where only the top-ranked items are retrieved. Losses ignoring Top-$K$ truncation may face performance bottlenecks. To address this, LambdaLoss@$K$ and SONG@$K$ optimize NDCG@$K$ using elegant lambda weights and compositional optimization, respectively, but their performance in RS remains unsatisfactory, as discussed in \cref{subsec:performance_comparison}. Other methods, such as AATP \cite{boyd2012accuracy}, LLPAUC \citep{shi2024lower}, and OPAUC \citep{shi2023theories}, target metrics like Precision@$K$ and Recall@$K$, yet their theoretical connections to NDCG@$K$ remain unclear. While AATP employs a quantile technique, it lacks a theoretical foundation and suffers from inefficiency issues, making it impractical for RS. LLPAUC and OPAUC rely on complex adversarial training, potentially limiting their effectiveness and applicability.


\begin{table}[t]
    \centering
    \caption{SeqRec results on Beauty and Games \citep{he2016ups,mcauley2015image} datasets.}
    \label{tab:seqrec-results}
    \begin{tabular}{l|cc|cc}
    \Xhline{1.2pt}
    \multicolumn{1}{c|}{\multirow{2}[4]{*}{\textbf{Loss}}} & 
    \multicolumn{2}{c|}{\textbf{Beauty}} & 
    \multicolumn{2}{c}{\textbf{Games}} \bigstrut\\
    \cline{2-5}
    & Hit@20 & NDCG@20 & Hit@20 & NDCG@20 \bigstrut\\
    \Xhline{1pt}
    BCE \citep{kang2018self}    & 0.1130 & 0.0484 & 0.1577 & 0.0671 \bigstrut[t]\\
    SL \citep{wu2024effectiveness}     & \uline{0.1578} & \uline{0.0766} & \uline{0.2243} & \uline{0.1024} \\
    \textbf{SL@20 (Ours)} & \textbf{0.1586} & \textbf{0.0780} & \textbf{0.2283} & \textbf{0.1045} \bigstrut[b]\\
    \hline
    \textcolor{red}{\textbf{Imp. \%}} & \textcolor{red}{\textbf{+0.51\%}} & \textcolor{red}{\textbf{+1.82\%}} & \textcolor{red}{\textbf{+1.78\%}} & \textcolor{red}{\textbf{+2.05\%}} \bigstrut\\
    \Xhline{1.2pt}
    \end{tabular}
\end{table}

\begin{table}[t]
    \centering
    \caption{LP results on Cora and Citeseer \citep{sen2008collective} datasets.}
    \label{tab:lp-results}
    \begin{tabular}{l|cc|cc}
    \Xhline{1.2pt}
    \multicolumn{1}{c|}{\multirow{2}[4]{*}{\textbf{Loss}}} & 
    \multicolumn{2}{c|}{\textbf{Cora}} & 
    \multicolumn{2}{c}{\textbf{Citeseer}} \bigstrut\\
    \cline{2-5}
    & Hit@20 & MRR & Hit@20 & MRR \bigstrut\\
    \Xhline{1pt}
    BCE \citep{kang2018self}     & 0.3643 & 0.1482 & 0.3560 & 0.1424 \bigstrut[t]\\
    SL \citep{wu2024effectiveness}     & \uline{0.4668} & \uline{0.1772} & \uline{0.4989} & \uline{0.1942} \\
    \textbf{SL@20 (Ours)} & \textbf{0.4839} & \textbf{0.1812} & \textbf{0.5099} & \textbf{0.1963} \bigstrut[b]\\
    \hline
    \textcolor{red}{\textbf{Imp. \%}} & \textcolor{red}{\textbf{+3.65\%}} & \textcolor{red}{\textbf{+2.25\%}} & \textcolor{red}{\textbf{+2.20\%}} & \textcolor{red}{\textbf{+1.08\%}} \bigstrut\\
    \Xhline{1.2pt}
    \end{tabular}
\end{table}

\section{Conclusion and Future Directions} \label{sec:conclusion_future_work}

This work introduces SoftmaxLoss@$K$ (SL@$K$), a novel recommendation loss tailored for optimizing NDCG@$K$. SL@$K$ employs a quantile-based technique to address the Top-$K$ truncation challenge and derives a smooth approximation to tackle the discontinuity issue. We establish a tight bound between SL@$K$ and NDCG@$K$, demonstrating its theoretical effectiveness. Beyond its theoretical soundness, SL@$K$ offers a concise form, introducing only quantile-based weights atop the conventional Softmax Loss, making it both easy to implement and computationally efficient. 

Looking ahead, a promising direction is to develop incremental quantile estimation methods to further improve the efficiency of SL@$K$ and enable incremental learning in RS. Additionally, since Top-$K$ metrics are widely used, further exploring the application of SL@$K$ in other domains, such as multimedia retrieval, question answering, and anomaly detection, is valuable.



\begin{acks}
    This work is supported by the Zhejiang Province "JianBingLingYan+X" Research and Development Plan (2025C02020).
\end{acks}



\clearpage
\bibliographystyle{ACM-Reference-Format}
\balance
\bibliography{references}


\clearpage
\appendix


\section{Analysis of \texorpdfstring{NDCG@$K$}{NDCG@K} Surrogate Losses} \label{app:gradient_vanishing}

In this section, we provide additional analysis of NDCG@$K$ surrogate losses, including LambdaLoss@$K$ \citep{jagerman2022optimizing} and SONG@$K$ \citep{qiu2022large}. We investigate their gradient distributions, which may shed light on their training instability and underperformance in RS.

\noindentpar{LambdaLoss@$K$ \citep{jagerman2022optimizing}.}
LambdaLoss@$K$ incorporates the lambda weights \citep{burges2006learning,wang2018lambdaloss} to optimize NDCG@$K$, and has shown promising results in document retrieval \citep{liu2009learning}. In recommendation scenarios, LambdaLoss@$K$ can be viewed as the following BPR-like \citep{rendle2009bpr} loss:
\begin{equation} \label{eq:lambdalossatk}
    \mathcal{L}_{\text{LambdaLoss@}K}(u) = \sum_{i \in \mathcal{P}_u, j \in \mathcal{N}_u} \mu_{uij} \cdot \sigma_{\text{softplus}}(d_{uij}) ,
\end{equation}
where $\sigma_{\text{softplus}}(x) = \log(1 + e^x)$ is the softplus activation function, and $\mu_{uij}$ is the lambda weight, which is defined as
\begin{equation} \label{eq:lambdalossatk-weight}
    \mu_{uij} = \begin{cases}
        \eta_{uij} \left( 1 - \dfrac{1}{\log_2(\pi_{uij} + 1)} \right)^{-1} &, \text{if } \pi_{uij} \coloneqq \max\{\pi_{ui}, \pi_{uj}\} > K , \\
        \eta_{uij} &, \text{otherwise} ,
    \end{cases}
\end{equation}
and
\begin{equation} \label{eq:lambdalossatk-eta}
    \eta_{uij} = \frac{1}{\log_2(|\pi_{ui} - \pi_{uj}| + 1)} - \frac{1}{\log_2(|\pi_{ui} - \pi_{uj}| + 2)} .
\end{equation}
The term $\eta_{uij}$ represents the difference between the reciprocals of adjacent discount factors $1 / \log_2(\cdot)$, causing the lambda weight $\mu_{uij}$ to rapidly decay to $0$ as $|\pi_{ui} - \pi_{uj}|$ increases, \ie as the ranking positions of the two items diverge. Consequently, during training, only negative items ranked close to positive items receive significant gradients, while most negative items are undertrained. This behavior is counter-intuitive and leads to inefficient training.

\cref{fig:lambda_weight} illustrates the lambda weights $\mu_{uij}$ for the Top-20 items in LambdaLoss@$5$. Results show that with the ranking differences $|\pi_{ui} - \pi_{uj}|$ approaching 20, $\mu_{uij}$ nearly vanishes (\ie $\mu_{uij} < 0.005$). This implies that in an RS with $|\mathcal{I}|$ items, at most $40|\mathcal{I}|$ lambda weights exceed $0.005$. This accounts for less than $1\%$ of the total $|\mathcal{I}|^2$ lambda weights in practical RS scenarios, which typically involve over 4K items. This clearly highlights the gradient vanishing issue in LambdaLoss@$K$. On the other hand, a small fraction (\ie $1 / |\mathcal{I}|$) of lambda weights exceed $0.3$, dominating the gradients and disproportionately influencing the optimization direction. This imbalance increases training instability and hinders effective optimization. Moreover, this means that we can not use a larger learning rate to mitigate the issue of gradient vanishing during sampling estimation. Sampling these few instances with large lambda weights occasionally can cause numerical explosions, further complicating optimization. Overall, the extreme long-tailed distribution of lambda weights poses significant challenges to optimization, which cannot be resolved by simply tuning the learning rate.

\begin{figure}[t]
    \centering
    \includegraphics[width=0.49\textwidth]{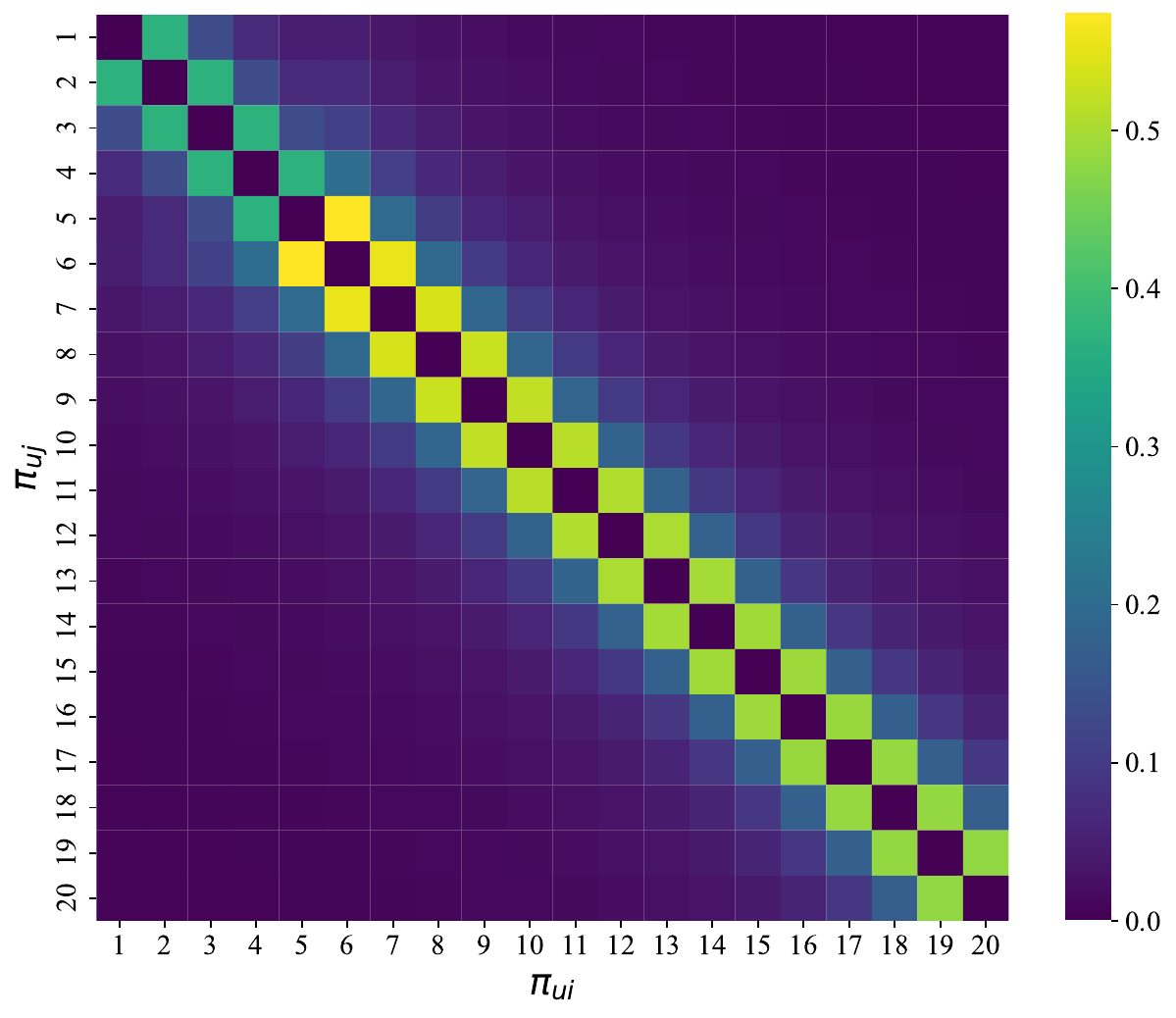}
    \caption{The lambda weight $\mu_{uij}$ ($\geq 0.005$) of Top-20 items in LambdaLoss@$5$.}
    \label{fig:lambda_weight}
    \Description{The lambda weights $\mu_{uij}$ of Top-20 items in LambdaLoss@$5$.}
\end{figure}

Based on the above analysis, we summarize the limitations of LambdaLoss@$K$ in RS as follows:

\begin{itemize}[topsep=3pt,leftmargin=10pt,itemsep=0pt]
    \item \emph{Computational inefficiency}: The calculation of lambda weights $\mu_{uij}$ requires accurate item rankings $\pi_{ui}$ and $\pi_{uj}$, which necessitates a full sorting of items for each user at every iteration, leading to a high computational complexity of $O(|\mathcal U| |\mathcal I| \log |\mathcal I|)$. Therefore, LambdaLoss@$K$ is impractical for large-scale RS with extensive user and item spaces (\cf \cref{fig:execution_time_comparison,tab:slatk_versus_lambda}).
    \item \emph{Gradient instability}: Due to the large item space and the sparsity of positive instances in RS, most lambda weights $\mu_{uij}$ are extremely small, since $|\pi_{ui} - \pi_{uj}|$ tends to be large for negative item $j$ that is ranked far from positive item $i$. In our experiments, we found that 99\% of lambda weights are less than 0.005, suggesting that the gradients of LambdaLoss@$K$ are dominated by a few training instances, while others contribute negligibly (\cf \cref{fig:gradient_distribution_comparison,fig:lambda_weight}). This increases training instability and hampers the data utilization efficiency.
    \item \emph{Sampling-unfriendly}: To address the computational inefficiency, one may resort to Monte Carlo sampling \citep{metropolis1953equation} for LambdaLoss@$K$ (\cf \cref{subapp:sample_ranking_estimation}). Specifically, for each user, we sample and rank $N$ items, and the sampled ranking positions are multiplied by $|\mathcal{I}| / N$ to approximate the rankings $\pi_{ui}$ and $\pi_{uj}$ in lambda weights $\mu_{uij}$. While this strategy reduces the complexity to $O(|\mathcal U| N \log N)$, the performance of the sampling-based LambdaLoss@$K$ is significantly degraded, which can be attributed to its high sensitivity to estimation errors. Specifically, the instances with larger lambda weights $\mu_{uij}$, which contribute significantly to training, tend to have smaller $|\pi_{ui} - \pi_{uj}|$. Therefore, even small estimation errors in rankings can lead to substantial deviations in lambda weights, resulting in a performance degradation of over 30\% in our experiments (\cf \cref{tab:slatk_versus_lambda}).
\end{itemize}

\noindentpar{SONG@$K$ \citep{qiu2022large}.}
SONG@$K$ introduces the bilevel compositional optimization technique \citep{wang2017stochastic} to optimize NDCG@$K$. Specifically, SONG@$K$ first smooths the DCG@$K$ by approximating $\pi_{ui}$ with a continuous function $g_{ui} = \sum_{j \in \mathcal{I}} \sigma_{\text{relu}}(d_{uij})$, where $\sigma_{\text{relu}}(x) = \max(0, x + 1)^2$ is a surrogate for the Heaviside step function $\delta(x)$. Subsequently, SONG@$K$ proposes a stochastic gradient estimator $G_{\text{SONG@}K}(u)$ for $\nabla \mathrm{DCG}@K(u)$ as follows:
\begin{equation}
    G_{\text{SONG@}K}(u) = -\frac{1}{\log 2} \sum_{i \in \mathcal{P}_u} \frac{\sigma_{\text{sigmoid}}(s_{ui} - \beta_{u}^{K})}{\log_2^2(g_{ui} + 1) \cdot (g_{ui} + 1)} \nabla \hat{g}_{ui} ,
\end{equation}
where $g_{ui}$ is maintained as an exponential moving average \citep{kingma2014adam} of $\hat{g}_{ui}$, and $\beta_{u}^{K}$ is the Top-$K$ quantile defined similarly to \cref{eq:quantile} and updated by the quantile regression \citep{koenker2005quantile,hao2007quantile,boyd2012accuracy} (\cf \cref{subapp:quantile_regression}). 

Although SONG@$K$ has theoretical guarantees for optimizing NDCG@$K$, it encounters severe \emph{training instability} due to the highly long-tailed gradients. The underlying reason is that $g_{ui} \in [0, 9|\mathcal{I}|]$ is nearly unbounded in real-world RS, where $d_{uij} \in [-2, 2]$. This results in substantial variance in the gradients, leading to extremely diverse gradient distribution. Empirical results validate this issue, where SONG@$K$ exhibits significantly diverse gradients spanning over seven orders of magnitude (\cf \cref{fig:gradient_distribution_comparison}). This skewed gradient distribution hampers the optimization process, as the model is overwhelmed by the few instances with large gradients, while the majority of instances remain undertrained, which limits the training efficiency and stability. In contrast, SL@$K$ has a much more moderate gradient distribution due to (i) the bounded quantile-based weight $w_{ui}$, and (ii) the Softmax-normalized SL gradients $\nabla \mathcal{L}_{\text{SL}}(u, i)$, thus ensuring training stability.


\section{Additional Theoretical Analysis of \texorpdfstring{SL@$K$}{SL@K}} \label{app:theoretical_analysis}

In this section, we provide a detailed additional theoretical analysis of our proposed SL@$K$ loss. Particularly, \cref{subapp:activation_functions} discusses the rationale behind the selection of the activation functions in SL@$K$ (\cf \cref{subsubsec:smooth_surrogate_for_ndcgatk}). \cref{subapp:proof_slatk} presents the proof of \cref{thm:slatk} (\cf \cref{subsubsec:theoretical_properties_of_slatk}). Finally, \cref{subapp:gradient_analysis} provides a gradient analysis of SL@$K$ to support the noise robustness analysis in \cref{subsubsec:theoretical_properties_of_slatk}.

\subsection{Activation Functions in \texorpdfstring{SL@$K$}{SL@K}} \label{subapp:activation_functions}

In \cref{eq:slatk}, we smooth SL@$K$ by two conventional activation functions, \ie the sigmoid function $\sigma_w(x) = 1 / (1 + \exp(-x / \tau_w))$ and the exponential function $\sigma_d(x) = \exp(x / \tau_d)$, where $\tau_w$ and $\tau_d$ are the temperature parameters. In this section, we will discuss the rationale behind the selection of these activation functions, as summarized in \cref{tab:activation_functions}.

\begin{table}[t]
    \centering
    \caption{Different activation functions choices in SL@$K$.}
    \label{tab:activation_functions}
    \small
    \begin{tabularx}{0.48\textwidth}{c|YY}
        \Xhline{1.2pt}
        $(\sigma_w, \sigma_d)$ & \textbf{Sigmoid} & \textbf{Exponential} \bigstrut\\
        \Xhline{1.0pt}
        \textbf{Sigmoid} & \ding{55} (not upper bound) & \textcolor{red}{\ding{51}} \textbf{\textcolor{red}{(Our SL@$K$ loss)}} \bigstrut[t]\\
        \textbf{Exponential} & \ding{55} (not upper bound) & \ding{55} (not tight enough) \bigstrut[b]\\
        \Xhline{1.2pt}
    \end{tabularx}
\end{table}

\noindentparnoline{Case 1: $(\sigma_w, \sigma_d) = $ (Sigmoid, Sigmoid).}  
To achieve an upper bound of $-\log \mathrm{DCG}@K$ from \cref{eq:slatk-final} to \cref{eq:slatk}, $\sigma_d(\cdot)$ must serve as an upper bound of the Heaviside step function $\delta(\cdot)$, ensuring that the surrogate loss (\ref{eq:slatk}) correctly bounds the target objective. While the sigmoid function provides a close approximation to $\delta(\cdot)$, it does not satisfy the requirement of being an upper bound. As a result, selecting the sigmoid function for $\sigma_d(\cdot)$ would fail to establish the upper bound of $-\log \mathrm{DCG}@K$, thereby undermining the theoretical guarantees.

\noindentparnoline{Case 2: $(\sigma_w, \sigma_d) = $ (Sigmoid, Exponential).}  
This configuration corresponds to our proposed SL@$K$ loss. In this case, SL@$K$ achieves a tight upper bound for $-\log \mathrm{DCG}@K$, as proven in \cref{thm:slatk} and \cref{subapp:proof_slatk}. This tightness arises from the exponential function's ability to serve as an effective upper bound for $\delta(\cdot)$, while the sigmoid function provides a close approximation that ensures stability and smoothness in optimization (\cf Case 4).

\noindentparnoline{Case 3: $(\sigma_w, \sigma_d) = $ (Exponential, Sigmoid).} Similar to Case 1, the sigmoid function as $\sigma_d(\cdot)$ does not serve as an upper bound of $\delta(\cdot)$, hindering the theoretical guarantees of SL@$K$.

\noindentparnoline{Case 4: $(\sigma_w, \sigma_d) = $ (Exponential, Exponential).}  
In this case, SL@$K$ indeed serves as an upper bound of $-\log \mathrm{DCG}@K$, but the exponential function does not approximate the Heaviside step function $\delta(\cdot)$ tightly, resulting in a loose upper bound. Specifically, the difference between the sigmoid function $1 / (1 + \exp(-x / \tau_w))$ and $\delta(x)$ is $1 / (1 + \exp(|x| / \tau_w)) \approx \exp(-|x| / \tau_w)$ when $\tau_w$ is small. In contrast, the difference between the exponential function $\exp(x / \tau_d)$ and $\delta(x)$ is $\exp(x / \tau_d) - 1 \approx x / \tau_d$ when $x > 0$ and $\tau_d$ is large. This shows that the sigmoid function provides a better approximation of $\delta(\cdot)$. Moreover, although the sigmoid function is not an upper bound of $\delta(\cdot)$, it can still be used in $\sigma_w(\cdot)$ with a tight upper bound guarantee, as proven in \cref{thm:slatk}.

\subsection{Proof of \texorpdfstring{\cref{thm:slatk}}{Theorem SL@K}} \label{subapp:proof_slatk}

\RestatableThmSlatk*

\begin{proof}[Proof of \cref{thm:slatk}]
    Recall that we have derived inequality (\ref{eq:slatk-derivation-4}) in \cref{subsec:softmaxlossatk}, \ie
    \begin{equation} \label{eq:slatk-proof-1}
        -\log \mathrm{DCG}@K(u) \leq \sum_{i \in \mathcal{P}_u} \frac{\mathbb{I}(s_{ui} \geq \beta_{u}^{K})}{H_{u}^{K}} \log \pi_{ui} - \log H_{u}^{K} .
    \end{equation}
    By the assumption of $H_{u}^{K} \geq 1$, the last term $-\log H_{u}^{K}$ can be relaxed, resulting in
    \begin{equation} \label{eq:slatk-proof-2}
        -\log \mathrm{DCG}@K(u) \leq \sum_{i \in \mathcal{P}_u} \frac{\mathbb{I}(s_{ui} \geq \beta_{u}^{K})}{H_{u}^{K}} \log \pi_{ui} .
    \end{equation}
    Recall again that
    \begin{equation} \label{eq:slatk-proof-3}
        \pi_{ui} = \sum_{j \in \mathcal{I}} \mathbb{I}(s_{uj} \geq s_{ui}) = \sum_{j \in \mathcal{I}} \delta(d_{uij}) \leq \sum_{j \in \mathcal{I}} \sigma_d(d_{uij}) ,
    \end{equation}
    where $d_{uij} = s_{uj} - s_{ui}$, $\delta(x) = \mathbb{I}(x \geq 0)$ is the Heaviside step function, and $\sigma_d(x) = \exp(x / \tau_d)$ is the exponential function serving as a smooth upper bound of $\delta(x)$ for any $x$ and temperature $\tau_d > 0$. Therefore, \cref{eq:slatk-proof-2} can be further relaxed as
    \begin{equation} \label{eq:slatk-proof-4}
        -\log \mathrm{DCG}@K(u) \leq \sum_{i \in \mathcal{P}_u} \frac{1}{H_{u}^{K}} \delta(s_{ui} - \beta_{u}^{K}) \log \left( \sum_{j \in \mathcal{I}} \sigma_d(d_{uij}) \right) .
    \end{equation}
    In the following, we discuss two cases of $H_{u}^{K}$ to complete the proof.

    \noindentparnoline{Case 1.} In the case of $H_{u}^{K} > 1$, we have 
    \begin{equation} \label{eq:slatk-proof-5}
        \frac{1}{H_{u}^{K}} \delta(s_{ui} - \beta_{u}^{K}) \leq \frac{1}{2} \delta(s_{ui} - \beta_{u}^{K}) \leq \sigma_w(s_{ui} - \beta_{u}^{K}) ,
    \end{equation}
    where $\sigma_w(x) = 1 / (1 + \exp(-x / \tau_w))$ is the sigmoid function with temperature $\tau_w > 0$. The last inequality in \cref{eq:slatk-proof-5} holds due to $\sigma_w(s_{ui} - \beta_{u}^{K}) \geq \frac{1}{2}$ if $s_{ui} > \beta_{u}^{K}$. Therefore, by \cref{eq:slatk-proof-4,eq:slatk-proof-5}, we have
    \begin{equation} \label{eq:slatk-proof-6}
        -\log \mathrm{DCG}@K(u) \leq \sum_{i \in \mathcal{P}_u} \sigma_w(s_{ui} - \beta_{u}^{K}) \log \left( \sum_{j \in \mathcal{I}} \sigma_d(d_{uij}) \right) .
    \end{equation}
    which exactly corresponds to the SL@$K$ loss $\mathcal{L}_{\text{SL@}K}(u)$ in \cref{eq:slatk}. Therefore, SL@$K$ serves as an upper bound of $-\log \mathrm{DCG}@K$ when $H_{u}^{K} > 1$.

    \noindentparnoline{Case 2.} In the case of $H_{u}^{K} = 1$, there only exists one positive item $i^* \in \mathcal{P}_u$ with $s_{ui^*} \geq \beta_{u}^{K}$. In this case, \cref{eq:slatk-proof-1} can be reduced to
    \begin{equation} \label{eq:slatk-proof-8}
        -\log \mathrm{DCG}@K(u) \leq \log \pi_{ui^*} \leq \log \left( \sum_{j \in \mathcal{I}} \sigma_d(d_{ui^*j}) \right) .
    \end{equation}
    Since $s_{ui^*} \geq \beta_{u}^{K}$, we have $\sigma_w(s_{ui^*} - \beta_{u}^{K}) \geq \frac{1}{2}$, which leads to
    \begin{equation} \label{eq:slatk-proof-9}
        -\frac{1}{2}\log \mathrm{DCG}@K(u) \leq \sigma_w(s_{ui^*} - \beta_{u}^{K}) \log \left( \sum_{j \in \mathcal{I}} \sigma_d(d_{ui^*j}) \right) \leq \mathcal{L}_{\textrm{SL@}K}(u) .
    \end{equation}
    Therefore, SL@$K$ serves an upper bound of $-\frac{1}{2}\log \mathrm{DCG}@K$ when $H_{u}^{K} = 1$. This completes the proof.
\end{proof}

\noindentpar{Discussion.} The condition in \cref{thm:slatk} is easy to satisfy in practice. For example, SL@20 achieves $H_{u}^{20} > 1$ for 53.32\%, 81.92\%, and 95.66\% of users within 5, 10, and 20 epochs training on Electronic dataset and MF backbone, respectively.

\subsection{Gradient Analysis and Noise Robustness} \label{subapp:gradient_analysis}

SL@$K$ inherently possesses the denoising ability to resist the false positive noise (\eg misclicks), which is common in RS \citep{wen2019leveraging}. To theoretically analyze the noise robustness of SL@$K$, we conduct a gradient analysis as follows:
\begin{equation} \label{eq:slatk-gradient}
    \nabla_{\mathbf{u}} \mathcal{L}_{\text{SL@}K} = \sum_{i \in \mathcal{P}_u} w_{ui} \nabla_{\mathbf{u}} \mathcal{L}_{\text{SL}}(u, i) + \frac{1}{\tau_w} w_{ui} (1 - w_{ui}) \mathcal{L}_{\text{SL}}(u, i) \nabla_{\mathbf{u}} s_{ui} .
\end{equation}
Therefore, we can derive an upper bound of $\left\| \nabla_{\mathbf{u}} \mathcal{L}_{\text{SL@}K} \right\|$ as
\begin{equation} \label{eq:slatk-gradient-bound}
    \left\| \nabla_{\mathbf{u}} \mathcal{L}_{\text{SL@}K} \right\|
    \leq
    \sum_{i \in \mathcal{P}_u} w_{ui} \left( \left\| \nabla_{\mathbf{u}} \mathcal{L}_{\text{SL}}(u, i) \right\| + \frac{1}{\tau_w} \mathcal{L}_{\text{SL}}(u, i) \left\| \nabla_{\mathbf{u}} s_{ui} \right\| \right) .
\end{equation}
It's evident that the above gradient upper bound of SL@$K$ \wrt the user embedding $\mathbf{u}$ is controlled by the weight $w_{ui}$. For any false positive item $i$ with low preference score $s_{ui}$, $w_{ui}$ will be sufficiently small, which reduces its impact on the gradient. This analysis indicates that SL@$K$ is robust to false positive noise.



\section{Sample Quantile Estimation} \label{app:sample_quantile_estimation}

In \cref{subsubsec:quantile_estimation}, we propose a sample yet efficient Monte Carlo sampling strategy to estimate the Top-$K$ quantile $\beta_{u}^{K}$ for SL@$K$. In this section, we provide additional details on the sample quantile estimation technique. 

In \cref{subapp:quantile_estimation_error_bound}, we provide the proof of the \emph{estimation error bound} of the sample quantile technique, \ie \cref{thm:quantile_estimation} in \cref{subsubsec:quantile_estimation}. In \cref{subapp:sample_quantile_estimation_tricks}, we detail our proposed \emph{negative sampling trick} in \cref{subsubsec:quantile_estimation} to enhance the Top-$K$ quantile estimation efficiency. In \cref{subapp:quantile_regression}, we briefly introduce the \emph{quantile regression} technique, which can be used as an alternative to the sample quantile estimation (though no performance gain is observed in our experiments). In \cref{subapp:sample_ranking_estimation}, we discuss the \emph{sample ranking estimation} technique, which can be used to estimate the ranking position (though less effective in practice). Finally, in \cref{subapp:slatk_optimization}, we provide the detailed optimization algorithm for SL@$K$.

\begin{figure*}[t]
    \centering
    \begin{subfigure}[b]{0.48\textwidth}
        \centering
        \includegraphics[width=\textwidth]{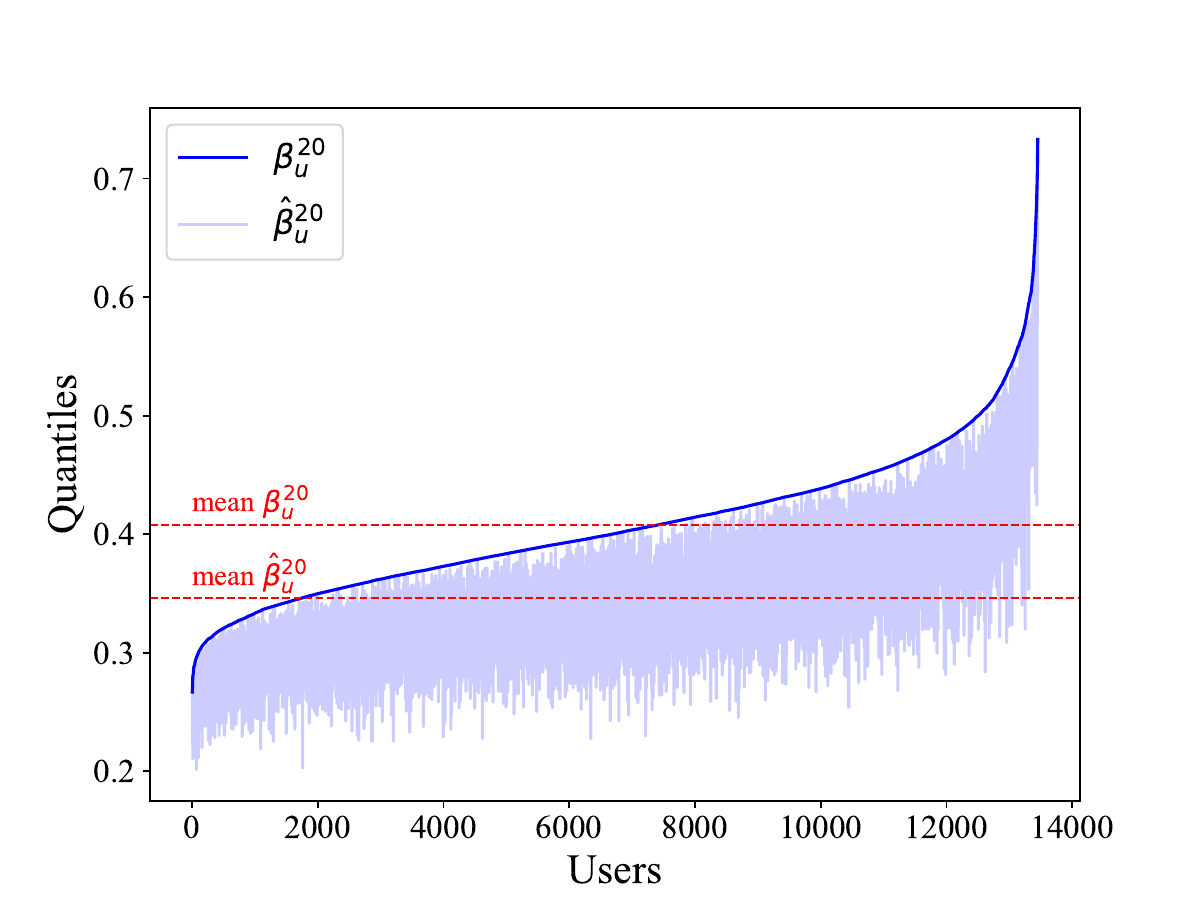}
        \caption{With negative sampling trick.}
        \label{fig:sample_quantile_estimation-trick}
    \end{subfigure}
    \begin{subfigure}[b]{0.48\textwidth}
        \centering
        \includegraphics[width=\textwidth]{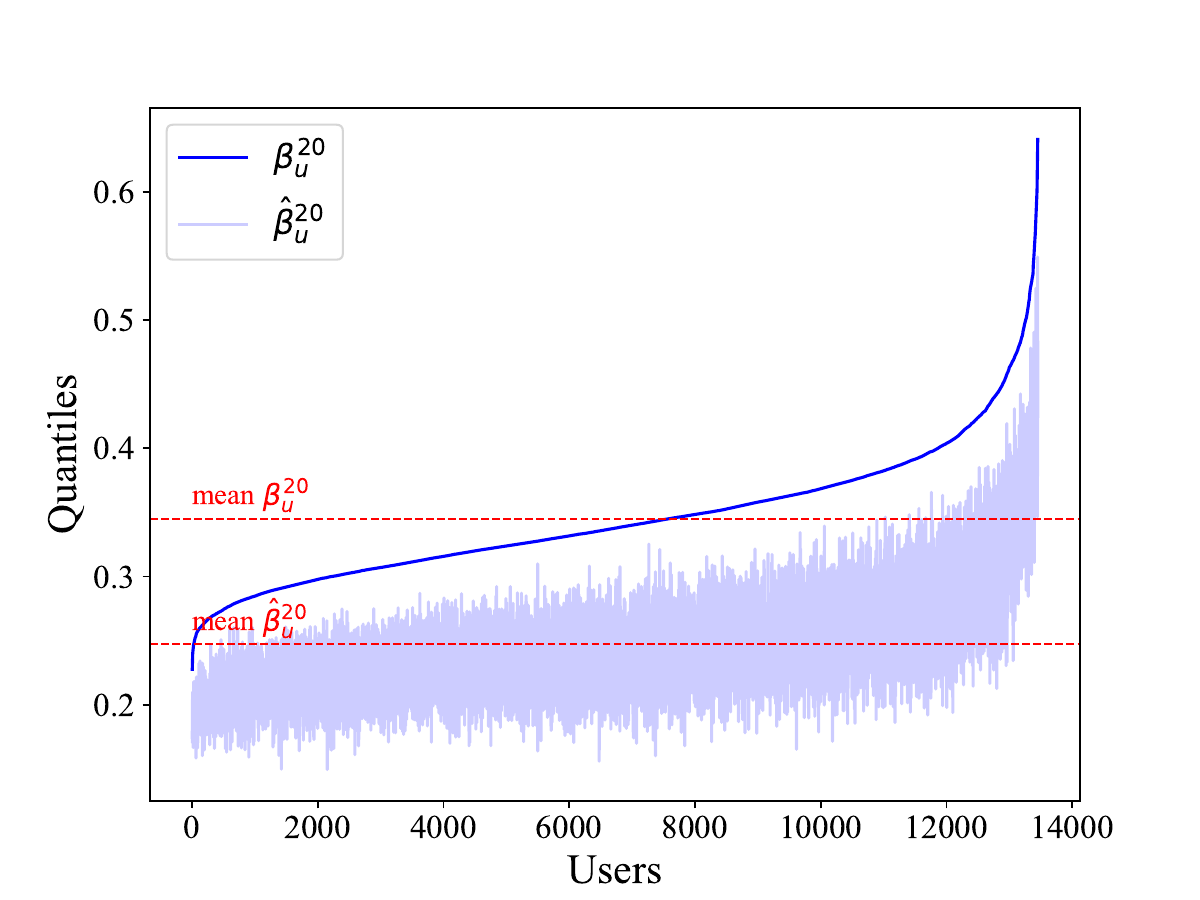}
        \caption{Without negative sampling trick.}
        \label{fig:sample_quantile_estimation-iid}
    \end{subfigure}
    \caption{Comparison of sample quantile estimation with and without the negative sampling trick.}
    \label{fig:sample_quantile_estimation-trick-compare}
    \Description{Comparison of sample quantile estimation with and without the negative sampling trick.}
\end{figure*}

\begin{figure*}[t]
    \centering
    \begin{subfigure}[b]{0.48\textwidth}
        \centering
        \includegraphics[width=\textwidth]{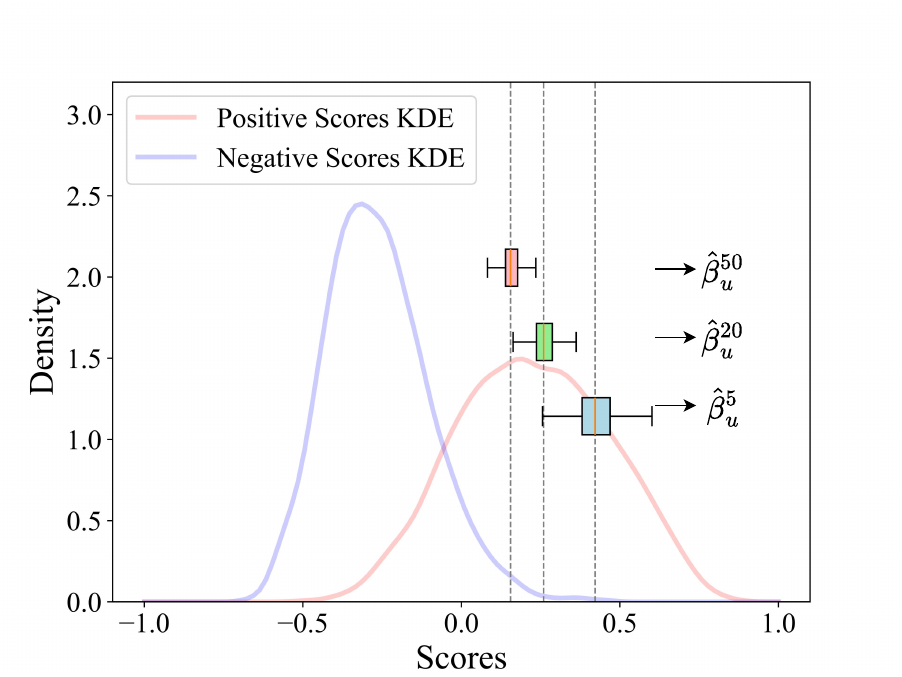}
        \caption{Estimated Top-$K$ quantile $\hat \beta_{u}^{K}$.}
        \label{fig:sample_quantile_estimation-compare-sample}
    \end{subfigure}
    \begin{subfigure}[b]{0.48\textwidth}
        \centering
        \includegraphics[width=\textwidth]{fig/quantile_scores_nosample.drawio.pdf}
        \caption{Ideal Top-$K$ quantile $\beta_{u}^{K}$.}
        \label{fig:sample_quantile_estimation-compare-ideal}
    \end{subfigure}
    \caption{Comparison of the estimated Top-$K$ quantile $\hat \beta_{u}^{K}$ with the ideal Top-$K$ quantile $\beta_{u}^{K}$.}
    \label{fig:sample_quantile_estimation-compare}
    \Description{Comparison of the estimated Top-$K$ quantile $\hat \beta_{u}^{K}$ with the ideal Top-$K$ quantile $\beta_{u}^{K}$.}
\end{figure*}

\subsection{Proof of \texorpdfstring{\cref{thm:quantile_estimation}}{Theorem Quantile Estimation}} \label{subapp:quantile_estimation_error_bound}

\RestatableThmQuantileEstimation*

To proof \cref{thm:quantile_estimation}, we first introduce the following lemma.

\begin{lemma}[Dvoretzky-Kiefer-Wolfowitz (DKW) inequality \citep{massart1990tight,bickel2015mathematical}] \label{lemma:dkw_inequality}
    For any \cdf $F(s)$ and the corresponding empirical \cdf $\hat{F}(s)$ with sample size $N$, given the sup-norm distance between $\hat{F}$ and $F$ defined as $\| \hat{F} - F \|_{\infty} = \sup_{s \in \mathbb{R}} \{|\hat{F}(s) - F(s)|\}$, we have
    \begin{equation} \label{eq:dkw_inequality}
        \Pr\left( \| \hat{F} - F \|_{\infty} > \varepsilon \right) \leq 2e^{-2N\varepsilon^2} .
    \end{equation}
\end{lemma}

The estimation error bound of the sample quantile technique (\cf \cref{thm:quantile_estimation}) can be simply derived from the DKW inequality (\cf \cref{lemma:dkw_inequality}) as follows.

\begin{proof}[Proof of \cref{thm:quantile_estimation}]
    Consider the error between $\hat \theta_u^p$ and $\theta_u^p$ , we have
    \begin{equation} \label{eq:quantile_estimation_proof-1}
        \begin{aligned}
            \Pr(\hat \theta_u^p > \theta_u^p + \varepsilon) 
            &= \Pr(p > \hat{F}_u(\theta_u^p + \varepsilon))	\\
            &= \Pr(F_u(\theta_u^p + \varepsilon) - \hat{F}_u(\theta_u^p + \varepsilon) > F_u(\theta_u^p + \varepsilon) - p)	\\
            &\leq \Pr(\| \hat{F}_u - F_u \|_{\infty} > \delta_{\varepsilon}^+) ,
        \end{aligned}    
    \end{equation}
    where $\delta_{\varepsilon}^+ = F_u(\theta_u^p + \varepsilon) - p$. Analogously, let $\delta_{\varepsilon}^- = p - F_u(\theta_u^p - \varepsilon)$ , we have
    \begin{equation}
        \Pr(\hat \theta_u^p < \theta_u^p - \varepsilon) \leq \Pr(\| \hat{F}_u - F_u \|_{\infty} > \delta_{\varepsilon}^-) .
    \end{equation}
    Therefore, we have the two side error bound (\cf \cref{eq:quantile_estimation_error}) by setting $\delta_{\varepsilon} = \min\{\delta_{\varepsilon}^+, \delta_{\varepsilon}^-\}$, which completes the proof.
\end{proof}

\subsection{Negative Sampling Trick} \label{subapp:sample_quantile_estimation_tricks}

\noindentpar{Negative sampling trick.}
In \cref{subsubsec:quantile_estimation}, we introduce a \emph{negative sampling trick} to estimate the Top-$K$ quantile $\beta_{u}^{K}$ in RS. Specifically, our sampled items will include all positive items $\mathcal{P}_u$ and $N$ ($\ll I$) \iid sampled negative items $\hat{\mathcal{N}}_u = \{j_k : j_k \overset{\iid}{\sim} \operatorname{Uniform}(\mathcal{N}_u)\}_{k=1}^{N}$ from the negative item set $\mathcal{N}_u$ uniformly. Since the Top-$K$ quantile is usually located within the range of positive scores, this trick can estimate the quantile more effectively than directly \iid sampling from all items (\cf \cref{fig:sample_quantile_estimation-trick-compare} for empirical comparison, all experiments are conducted under the same settings as \cref{fig:sample_quantile}).

\noindentpar{Discussions on bias.}
However, applying the negative sampling trick leads to a theoretical gap. Since the sampled items $\hat{\mathcal{I}}_u = \mathcal{P}_u \cup \hat{\mathcal{N}}_u$ are not \iid sampled from the whole item set $\mathcal{I}$, we should not directly take the $(K / I)$-th quantile of $\hat{\mathcal{I}}_u$ as the estimated quantile $\hat \beta_{u}^{K}$, which may introduce bias. Instead, under a reasonable assumption that all Top-$\min(K, P_u)$ items are positive items, we can set the estimated quantile $\hat \beta_{u}^{K}$ as:

\begin{itemize}[topsep=3pt,leftmargin=10pt,itemsep=0pt]
    \item If $K \leq P_u$, $\hat \beta_{u}^{K}$ should be set as the Top-$K$ score of $\{s_{ui}\}$, where $i \in \mathcal{P}_u$.
    \item If $K > P_u$, $\hat \beta_{u}^{K}$ should be set as the $((K - P_u) / I)$-th quantile of $\{s_{uj}\}$, where $j \in \hat{\mathcal{N}}_u$.
\end{itemize}

The sampling strategy above can be seen as unbiased. Nevertheless, this sampling setting is still not practical in RS. In the case of $K > P_u$, the quantile ratio $(K - P_u) / I$ can be too small and even less than $1 / N$ (\eg $K = 20, I = 10^5, N = 10^3$). Therefore, the estimated quantile $\hat \beta_{u}^{K}$ could be theoretically higher than all the negative item scores and can not be estimated by sampling negative items $\hat{\mathcal{N}}_u$.

Given the impracticality of the above unbiased sampling setting, we slightly modify the above sampling to derive our proposed negative sampling trick. Specifically, we set $\hat \beta_{u}^{K}$ as the Top-$K$ score of $\{s_{uk}\}$, where $k \in \mathcal{P}_u \cup \hat{\mathcal{N}}_u$. This sampling trick perfectly fits the above unbiased case when $K \leq P_u$. In the case of $K > P_u$, this setting actually estimates the $(K - P_u) / N$-th quantile of negative item scores, introducing a slight bias but also making the training more stable. Moreover, it's clear that the estimated quantile $\hat \beta_{u}^{K}$ will always be lower than the ideal Top-$K$ quantile $\beta_{u}^{K}$ under this sampling trick (\cf \cref{fig:sample_quantile_estimation-trick-compare}), which leads to a more conservative yet moderate truncation in training SL@$K$, as shown in \cref{fig:sample_quantile_estimation-compare}.


\begin{algorithm*}[t]
    \caption{SL@$K$ Optimization.}
    \label{alg:slatk}
    \begin{algorithmic}[1]
        \Statex \textbf{Input:} user and item sets $\mathcal{U}, \mathcal{I}$; dataset $\mathcal{D} = \{y_{ui} \in \{0, 1\} : u \in \mathcal{U}, i \in \mathcal{I}\}$; score function $s_{ui}: \mathcal{U} \times \mathcal{I} \to \mathbb{R}$ with parameters $\Theta$; negative sampling number $N$; the number of epochs $T$; the number of $K$; temperature parameters $\tau_w, \tau_d$; quantile update interval $T_{\beta}$.
        \State Initialize the estimated Top-$K$ quantiles $\hat \beta_{u}^{K} \gets 0$ for all $u \in \mathcal{U}$.
        \For{$t = 1, 2, \ldots, T$}
            \For{$u \in \mathcal{U}$}
                \State Let $\mathcal{P}_u = \{i : y_{ui} = 1\}$ be the positive items of user $u$.
                \State Let $\mathcal{N}_u = \{i : y_{ui} = 0\}$ be the negative items of user $u$.
                \Statex \Comment{Estimate the Top-$K$ quantiles $\hat \beta_{u}^{K}$ by sample quantile estimation}
                \Statex \Comment{Complexity: $O((|\mathcal{P}_u| + N) \log (|\mathcal{P}_u| + N)) \approx O(N \log N)$}
                \If{$t \equiv 0 \mod T_{\beta}$}
                    \State Sample $N$ negative items $\hat{\mathcal{N}}_u = \{j_k : j_k \overset{\iid}{\sim} \operatorname{Uniform}(\mathcal{N}_u)\}_{k=1}^{N}$, 
                        let $\hat{\mathcal{I}}_u = \mathcal{P}_u \cup \hat{\mathcal{N}}_u$.
                    \State Sort items $\hat i \in \hat{\mathcal{I}}_u$ by scores $\{s_{u\hat{i}}\}$.
                    \State Estimate the Top-$K$ quantile $\hat \beta_{u}^{K} \gets \hat{\mathcal{I}}_u[K]$, where $\hat{\mathcal{I}}_u[K]$ denotes the $K$-th top-ranked item in $\hat{\mathcal{I}}_u$.
                \EndIf
                \Statex \Comment{Optimize $\Theta$ by minimizing the SL@$K$ loss}
                \Statex \Comment{Complexity: $O(|\mathcal{P}_u| N)$}
                \State Sample $N$ negative items $\hat{\mathcal{N}}_u = \{j_k : j_k \overset{\iid}{\sim} \operatorname{Uniform}(\mathcal{N}_u)\}_{k=1}^{N}$.
                \For{$i \in \mathcal{P}_u$}
                    \State Compute the quantile-based weight $w_{ui} = \sigma_w(s_{ui} - \hat \beta_{u}^{K})$,
                        where $\sigma_w = \sigma(\cdot / \tau_w)$, and $\sigma(\cdot)$ is the sigmoid function.
                    \State Compute the (sampled) SL loss $\mathcal{L}_{\text{SL}}(u, i) = \log \sum_{j \in \hat{\mathcal{N}}_u} \sigma_d(d_{uij})$,
                        where $\sigma_d = \exp(\cdot / \tau_d)$.
                \EndFor
                \State Compute the loss $\mathcal{L}_{\text{SL@}K}(u) = \sum_{i \in \mathcal{P}_u} w_{ui} \cdot \mathcal{L}_{\text{SL}}(u, i)$.
                \State Update the parameters $\Theta$ by minimizing $\mathcal{L}_{\text{SL@}K}(u)$.
            \EndFor
        \EndFor
        \Statex \textbf{Output:} the optimized parameters $\Theta$.
    \end{algorithmic}
\end{algorithm*}

\subsection{Quantile Regression} \label{subapp:quantile_regression}

\noindentpar{Quantile regression.}
Quantile regression technique \citep{koenker2005quantile,hao2007quantile} can also be utilized for sample quantile estimation. Specifically, to estimate the $p$-th quantile (see definition in \cref{thm:quantile_estimation}), the \emph{quantile regression loss} $\mathcal{L}_{\text{QR}}$ can be defined as
\begin{equation} \label{eq:quantile_regression-1}
    \mathcal{L}_{\text{QR}}(u) = \mathbb{E}_{i \sim \operatorname{Uniform}(\mathcal{I})} \left[ p(s_{ui} - \hat \beta_u)_+ + (1 - p)(\hat \beta_u - s_{ui})_+ \right] ,
\end{equation}
where $(\cdot)_+ = \max(\cdot, 0)$, $\hat \beta_u$ is the estimated $p$-th quantile, and $\operatorname{Uniform}(\mathcal{I})$ denotes the uniform distribution over the item set $\mathcal{I}$. Note that for any $x \in \mathbb{R}$, $x \cdot \delta(x) = x_+$, and $x_+ - (-x)_+ = x$, we can rewrite the quantile regression loss in \cref{eq:quantile_regression-1} as
\begin{equation} \label{eq:quantile_regression-2}
    \mathcal{L}_{\text{QR}}(u) = \mathbb{E}_{i \sim \operatorname{Uniform}(\mathcal{I})} \left[ (\hat \beta_u - s_{ui})(\delta(\hat \beta_u - s_{ui}) - p) \right] .
\end{equation}

Suppose that $S$ is a random variable representing the preference scores $s_{ui}$ (\wrt user $u$), and $F_S$ is the \cdf of $S$ on $\mathbb{R}$. Since item $i$ is uniformly distributed, the quantile regression loss (\ref{eq:quantile_regression-1}) can be rewritten as the following expectation:
\begin{equation} \label{eq:quantile_regression-3}
    \begin{aligned}
        \mathcal{L}_{\text{QR}}(u) 
        &= \mathbb{E}_{S \sim F_S} \left[ p(S - \hat \beta_u)_+ + (1 - p)(\hat \beta_u - S)_+ \right]    \\
        &= \int_{-\infty}^{\hat \beta_u} (1 - p)(\hat \beta_u - S) \diff F_S(S) + \int_{\hat \beta_u}^{\infty} p(S - \hat \beta_u) \diff F_S(S) .
    \end{aligned}
\end{equation}
Let $\beta_u = \argmin_{\hat \beta_u} \mathcal{L}_{\text{QR}}(u)$, we have
\begin{equation} \label{eq:quantile_regression-solution}
    (1 - p) \int_{-\infty}^{\beta_u} \diff F_S(S) = p \int_{\beta_u}^{\infty} \diff F_S(S) .
\end{equation}
This results in $\int_{-\infty}^{\beta_u} \diff F_S(S) = p$, \ie the optimal $\hat \beta_u$ is precisely the $p$-th quantile of scores $S$.

\noindentpar{Discussion.}
Quantile regression is indeed theoretically unbiased. However, we did not observe any performance gains in our experiments when applying it to SL@$K$. This may be because the quantile regression loss is relatively difficult to optimize due to sparse signals and large variance. In \cref{eq:quantile_regression-1}, as $p = 1 - K / |\mathcal{I}| \approx 1$ in Top-$K$ recommendation scenarios, the samples scored above the quantile $\hat \beta_u$ (high-ranked items) are assigned a weight of $p$, which is close to 1, while the samples scored below the quantile are assigned a nearly vanishing weight of $1 - p$. Given the sparse nature of the high-ranked items, the Monte Carlo estimation of the quantile regression loss is highly unstable -- once the high-ranked items are sampled, the loss will be dominated by these items, leading to large variance and unstable gradients \wrt $\hat \beta_u$. Although importance sampling \citep{wasserman2004all} may reduce the variance, it is difficult to accurately estimate the sampling distribution over the entire item set during sampling-based training. In contrast, our proposed strategy is both practical and effective for sample quantile estimation.



\subsection{Sample Ranking Estimation} \label{subapp:sample_ranking_estimation}

Similar to sample quantile estimation, sample ranking estimation \citep{weston2010large} can also be applied to estimate the ranking position $\pi_{ui}$. Specifically, we can uniformly sample $N$ negative items $\hat{\mathcal{N}}_u = \{j_k : j_k \overset{\iid}{\sim} \operatorname{Uniform}(\mathcal{N}_u)\}_{k=1}^{N}$, then sort the sampled items $i \in \hat{\mathcal{I}}_u = \mathcal{P}_u \cup \hat{\mathcal{N}}_u$ by their scores $\{s_{ui}\}$. Then, for any item $i$, given the sample ranking position $\pi_{ui}^{*}$ in the sampled items $\hat{\mathcal{I}}_u$, the estimated ranking position $\hat \pi_{ui}$ in the entire item set can be rescaled as
\begin{equation} \label{eq:sample_ranking_estimation}
    \hat \pi_{ui} = \pi_{ui}^{*} \cdot \frac{|\mathcal{I}|}{|\hat{\mathcal{I}}_u|} .
\end{equation}
Compared to sample quantile estimation, sample ranking estimation may result in greater errors, primarily because the estimated ranking $\hat \pi_{ui}$ obtained from sample ranking estimation is always discrete and predefined, i.e., $1, 1 + {|\mathcal{I}|}/{|\hat{\mathcal{I}}_u|}, 1 + 2{|\mathcal{I}|}/{|\hat{\mathcal{I}}_u|}, \cdots$. It is evident that sample ranking estimation will result in an expected error of at least $\frac{1}{2}{|\mathcal{I}|}/{|\hat{\mathcal{I}}_u|} \approx \frac{1}{2}{|\mathcal{I}|}/{N}$, which decreases proportionally to $1/N$ when $N$ is large. However, the error in sample quantile estimation decreases exponentially \wrt $N$, which leads to better estimation accuracy (\cf \cref{thm:quantile_estimation}). Therefore, sample ranking estimation is not suitable for recommendation losses that are highly sensitive to ranking positions, such as LambdaLoss \citep{wang2018lambdaloss} and LambdaLoss@$K$ \citep{jagerman2022optimizing}, as discussed in \cref{app:gradient_vanishing}.


\subsection{\texorpdfstring{SL@$K$}{SL@K} Optimization} \label{subapp:slatk_optimization}

The detailed optimization algorithm of SL@$K$ (\ref{eq:slatk}) is presented in \cref{alg:slatk}, which is based on the sample quantile estimation technique discussed in \cref{subapp:sample_quantile_estimation_tricks}. In practical SL@$K$ optimization, to address training difficulties arising from frequent quantile changes due to score variations (especially in the early stages), we introduce a quantile update interval hyperparameter $T_{\beta}$, i.e., where quantiles are updated every $T_{\beta}$ epochs.



\begin{table}[t]
    \centering
    \caption{Statistics of the datasets.}
    \label{tab:datasets-statistics}
    \begin{tabular}{l|rrrr}
    \Xhline{1.2pt}
    \multicolumn{1}{c|}{\textbf{Dataset}} & \textbf{\#Users} & \textbf{\#Items} & \textbf{\#Interactions} & \textbf{Density} \bigstrut\\
    \Xhline{1pt}
    Health \citep{he2016ups} & 1,974  & 1,200  & 48,189  & 0.02034 \bigstrut[t]\\
    Electronic \citep{he2016ups} & 13,455  & 8,360  & 234,521  & 0.00208 \\
    Gowalla \citep{cho2011friendship} & 29,858  & 40,988  & 1,027,464  & 0.00084 \\
    Book \citep{he2016ups}  & 135,109  & 115,172  & 4,042,382  & 0.00026 \\
    MovieLens \citep{harper2015movielens} & 939   & 1,016  & 80,393  & 0.08427 \\
    Food \citep{majumder2019generating}  & 5,875  & 9,852  & 233,038  & 0.00403 \bigstrut[b]\\
    \Xhline{1.2pt}
    \end{tabular}
\end{table}

\begin{table}[t]
    \centering
    \caption{Hyperparameters to be searched for each method.}
    \label{tab:hyperparameters-searching}
    \begin{tabular}{l|l}
        \Xhline{1.2pt}
        \multicolumn{1}{c|}{\textbf{Method}} & \multicolumn{1}{c}{\textbf{Hyperparameters}} \bigstrut\\
        \Xhline{1pt}
        BPR             & lr, wd \bigstrut[t]\\
        GuidedRec       & lr, wd \\
        SONG            & lr, wd, $\gamma_{g}$ \\
        SONG@$K$        & lr, wd, $\gamma_{g}$, $\eta_{\lambda}$ \\
        LLPAUC          & lr, wd, $\alpha$, $\beta$ \\
        SL              & lr, wd, $\tau$ \\
        AdvInfoNCE      & lr, wd, $\tau$ \\
        BSL             & lr, wd, $\tau_1$, $\tau_2$ \\
        PSL             & lr, wd, $\tau$ \\
        LambdaRank      & lr, wd \\
        LambdaLoss      & lr, wd \\
        LambdaLoss@$K$  & lr, wd \\
        SL@$K$          & lr, wd, $\tau_d$, $\tau_w$, $T_{\beta}$ \bigstrut[b]\\
        \Xhline{1.2pt}
    \end{tabular}
\end{table}

\section{Experimental Details} \label{app:experimental_details}

In this section, we provide additional details on the experiments, including the dataset descriptions in \cref{subapp:datasets}, the recommendation scenarios in \cref{subapp:recommendation_scenarios}, the recommendation backbones in \cref{subapp:recommendation_backbones}, the baseline methods and the corresponding hyperparameter settings in \cref{subapp:compared_methods_hyperparameters}. We also provide the optimal hyperparameters for all methods in \cref{subapp:optimal_hyperparameters} for reproducibility. Finally, we describe the additional information retrieval (IR) tasks in \cref{subapp:experimental_setup_ir}.

\noindentpar{Hardware and software.}
All experiments are conducted on 1x NVIDIA GeForce RTX 4090 GPU. The code is implemented in PyTorch \citep{paszke2019pytorch}. Both the datasets and code are available at \url{https://github.com/Tiny-Snow/IR-Benchmark}.

\subsection{Datasets} \label{subapp:datasets}

In our recommendation experiments, we use six benchmark datasets, as summarized in \cref{tab:datasets-statistics}: 
\begin{itemize}[topsep=3pt,leftmargin=10pt,itemsep=0pt]
    \item \textbf{Health / Electronic / Book \citep{he2016ups,mcauley2015image}:} These datasets are collected from the Amazon dataset, a large-scale collection of product reviews from Amazon\footnote{\url{https://www.amazon.com/}}. The 2014 version of the Amazon dataset contains 142.8 million reviews spanning May 1996 to July 2014.
    \item \textbf{Gowalla \citep{cho2011friendship}:} The Gowalla dataset is a check-in dataset from the location-based social network Gowalla\footnote{\url{https://en.wikipedia.org/wiki/Gowalla}}, including 1 million users, 1 million locations, and 6 million check-ins.
    \item \textbf{MovieLens \citep{harper2015movielens}:} The MovieLens dataset is a movie rating dataset from MovieLens\footnote{\url{https://movielens.org/}}. We use the MovieLens-100K version, which contains 100,000 ratings from 1,000 users and 1,700 movies.
    \item \textbf{Food \citep{majumder2019generating}:} The Food dataset consists of 180,000 recipes and 700,000 recipe reviews spanning 18 years of user interactions and uploads on Food.com\footnote{\url{https://www.food.com/}}.
\end{itemize}

\noindentpar{Dataset preprocessing.}
Following the standard practice in \citet{yang2024psl}, we use a 10-core setting \citep{he2016vbpr,wang2019neural}, \ie all users and items have at least 10 interactions. To remove low-quality interactions, we retain only those with ratings (if available) greater than or equal to 3. After preprocessing, we randomly split the datasets into 80\% training and 20\% test sets, and further split 10\% of the training set as a validation set for hyperparameter tuning.


\subsection{Recommendation Scenarios} \label{subapp:recommendation_scenarios}

In this paper, we evaluate the performance of each method primarily in the following two Top-$K$ recommendation scenarios:

\begin{itemize}[topsep=3pt,leftmargin=10pt,itemsep=0pt]
    \setlength{\itemsep}{0pt}
    \item \textbf{IID scenario \citep{he2020lightgcn,yang2024psl}:} The IID scenario is the most common recommendation scenario, where the training and test sets are split in an independent and identically distributed (\iid) manner from the entire dataset, ensuring the same distributions. This follows the setup described in \citet{he2020lightgcn} and \citet{yang2024psl}.
    \item \textbf{False Positive Noise scenario \citep{wu2023bsl,yang2024psl}:} The False Positive Noise scenario is commonly used to assess a method's ability to handle noisy data. Our false positive noise setting is adapted from \citet{wu2023bsl} and \citet{yang2024psl}. Specifically, given a noise ratio $r$, we randomly sample $\lceil r \times P_u \rceil$ negative items for each user $u$, and flip them to positive items to simulate false positive noise. The noise ratio $r$ represents the percentage of false positive noise and takes values from the set $\{ 5\%, 10\%, 15\%, 20\% \}$.
\end{itemize}

\begin{table*}[t]
    \centering
    \begin{minipage}[t]{0.48\linewidth}
    \centering
    \caption{Optimal hyperparameters on Health dataset.}
    \begin{tabular}{c|l|ccccc}
    \Xhline{1.2pt}
    \textbf{Model} & \multicolumn{1}{c|}{\textbf{Loss}} & \multicolumn{5}{c}{\textbf{Hyperparameters}} \bigstrut\\
    \Xhline{1pt}
    \multirow{14}[2]{*}{MF}
        & BPR           & 0.001 & 0.0001 &       &       &  \bigstrut[t]\\
        & GuidedRec     & 0.01  & 0     &       &       &  \\
        & SONG@20       & 0.1   & 0     & 0.9   & 0.001 &  \\
        & LLPAUC        & 0.1   & 0     & 0.7   & 0.01  &  \\
        & SL            & 0.1   & 0     & 0.2   &       &  \\
        & AdvInfoNCE    & 0.1   & 0     & 0.2   &       &  \\
        & BSL           & 0.1   & 0     & 0.2   & 0.2   &  \\
        & PSL           & 0.1   & 0     & 0.1   &       &  \\
        & SL@5          & 0.1   & 0     & 0.2   & 2.5   & 5 \\
        & SL@10         & 0.1   & 0     & 0.2   & 2.5   & 20 \\
        & SL@20         & 0.1   & 0     & 0.2   & 2.5   & 5 \\
        & SL@50         & 0.1   & 0     & 0.2   & 2.25  & 5 \\
        & SL@75         & 0.1   & 0     & 0.2   & 2.25  & 5 \\
        & SL@100        & 0.1   & 0     & 0.2   & 2.5   & 5 \bigstrut[b]\\
    \Xhline{1pt}
    \multirow{9}[2]{*}{LightGCN} 
        & BPR           & 0.001 & 0.000001 &       &       &  \bigstrut[t]\\
        & GuidedRec     & 0.01  & 0     &       &       &  \\
        & SONG@20       & 0.1   & 0     & 0.9   & 0.001 &  \\
        & LLPAUC        & 0.1   & 0     & 0.7   & 0.1   &  \\
        & SL            & 0.1   & 0     & 0.2   &       &  \\
        & AdvInfoNCE    & 0.1   & 0     & 0.2   &       &  \\
        & BSL           & 0.1   & 0     & 0.05  & 0.2   &  \\
        & PSL           & 0.1   & 0     & 0.1   &       &  \\
        & SL@20         & 0.01  & 0     & 0.2   & 2.5   & 20 \bigstrut[b]\\
    \Xhline{1pt}
    \multirow{9}[2]{*}{XSimGCL} 
        & BPR           & 0.1   & 0.000001 &       &       &  \bigstrut[t]\\
        & GuidedRec     & 0.001 & 0.000001 &       &       &  \\
        & SONG@20       & 0.1   & 0     & 0.1   & 0.01  &  \\
        & LLPAUC        & 0.1   & 0     & 0.1   & 0.1   &  \\
        & SL            & 0.1   & 0     & 0.2   &       &  \\
        & AdvInfoNCE    & 0.1   & 0     & 0.2   &       &  \\
        & BSL           & 0.1   & 0     & 0.05  & 0.2   &  \\
        & PSL           & 0.1   & 0     & 0.1   &       &  \\
        & SL@20         & 0.01  & 0     & 0.2   & 1.5   & 20 \bigstrut[b]\\
    \Xhline{1.2pt}
    \end{tabular}
    \label{tab:hyperparameters-health}
    \end{minipage}
    \hfill
    \begin{minipage}[t]{0.48\linewidth}
    \centering
    \caption{Optimal hyperparameters on Electronic dataset.}
    \begin{tabular}{c|l|ccccc}
    \Xhline{1.2pt}
    \textbf{Model} & \multicolumn{1}{c|}{\textbf{Loss}} & \multicolumn{5}{c}{\textbf{Hyperparameters}} \bigstrut\\
    \Xhline{1pt}
    \multirow{14}[2]{*}{MF}
        & BPR           & 0.001 & 0.00001 &       &       &  \bigstrut[t]\\
        & GuidedRec     & 0.01  & 0     &       &       &  \\
        & SONG@20       & 0.1   & 0     & 0.1   & 0.001 &  \\
        & LLPAUC        & 0.1   & 0     & 0.5   & 0.01  &  \\
        & SL            & 0.01  & 0     & 0.2   &       &  \\
        & AdvInfoNCE    & 0.1   & 0     & 0.2   &       &  \\
        & BSL           & 0.1   & 0     & 0.5   & 0.2   &  \\
        & PSL           & 0.01  & 0     & 0.1   &       &  \\
        & SL@5          & 0.01  & 0     & 0.2   & 2.25  & 5 \\
        & SL@10         & 0.01  & 0     & 0.2   & 2.25  & 20 \\
        & SL@20         & 0.01  & 0     & 0.2   & 2.25  & 20 \\
        & SL@50         & 0.01  & 0     & 0.2   & 2.25  & 20 \\
        & SL@75         & 0.01  & 0     & 0.2   & 2.25  & 20 \\
        & SL@100        & 0.01  & 0     & 0.2   & 2.25  & 5 \bigstrut[b]\\
    \Xhline{1pt}
    \multirow{9}[2]{*}{LightGCN} 
        & BPR           & 0.01  & 0.000001 &       &       &  \bigstrut[t]\\
        & GuidedRec     & 0.01  & 0     &       &       &  \\
        & SONG@20       & 0.1   & 0     & 0.1   & 0.01  &  \\
        & LLPAUC        & 0.1   & 0     & 0.5   & 0.01  &  \\
        & SL            & 0.01  & 0     & 0.2   &       &  \\
        & AdvInfoNCE    & 0.01  & 0     & 0.2   &       &  \\
        & BSL           & 0.01  & 0     & 0.2   & 0.2   &  \\
        & PSL           & 0.01  & 0     & 0.1   &       &  \\
        & SL@20         & 0.01  & 0     & 0.2   & 2.25  & 5 \bigstrut[b]\\
    \Xhline{1pt}
    \multirow{9}[2]{*}{XSimGCL} 
        & BPR           & 0.01  & 0     &       &       &  \bigstrut[t]\\
        & GuidedRec     & 0.01  & 0     &       &       &  \\
        & SONG@20       & 0.1   & 0     & 0.1   & 0.001 &  \\
        & LLPAUC        & 0.1   & 0     & 0.3   & 0.01  &  \\
        & SL            & 0.01  & 0     & 0.2   &       &  \\
        & AdvInfoNCE    & 0.1   & 0     & 0.2   &       &  \\
        & BSL           & 0.1   & 0     & 0.1   & 0.2   &  \\
        & PSL           & 0.1   & 0     & 0.1   &       &  \\
        & SL@20         & 0.01  & 0     & 0.2   & 1.25  & 5 \bigstrut[b]\\
    \Xhline{1.2pt}
    \end{tabular}
    \label{tab:hyperparameters-electronic}
    \end{minipage}
\end{table*}

\subsection{Recommendation Backbones} \label{subapp:recommendation_backbones}

Recommendation backbones (\aka recommendation models) are the core components of RS. In the scope of this paper, the recommendation backbones can be formally defined as the preference score function $s_{ui} : \mathcal{U} \times \mathcal{I} \to \mathbb{R}$ with parameters $\Theta$. It is crucial to evaluate the effectiveness of the recommendation loss on different backbones to ensure their generalization and consistency. In our experiments, we implement three popular recommendation backbones with different architectures and properties:
\begin{itemize}[topsep=3pt,leftmargin=10pt,itemsep=0pt]
    \setlength{\itemsep}{0pt}
    \item \textbf{MF \citep{koren2009matrix}:} MF is a fundamental yet effective recommendation model that factorizes the user-item interaction matrix into the user and item embeddings. Many embedding-based recommendation models leverage MF as the initial layer. Specifically, we set the embedding size $d = 64$ for all settings, following \citet{wang2019neural} and \citet{yang2024psl}.
    \item \textbf{LightGCN \citep{he2020lightgcn}:} LightGCN is an effective GNN-based recommendation model. LightGCN performs graph convolution on the user-item interaction graph to aggregate high-order interactions. Specifically, LightGCN simplifies NGCF \citep{wang2019neural} and retains only the non-parameterized graph convolution. In our experiments, we set the number of layers to 2, which aligns with the original settings in \citet{he2020lightgcn} and \citet{yang2024psl}.
    \item \textbf{XSimGCL} \citep{yu2023xsimgcl}: XSimGCL is a novel recommendation model based on contrastive learning \citep{jaiswal2020survey, liu2021self}. Based on a 3-layer LightGCN, XSimGCL adds random noise to the output embeddings of each layer, and introduces contrastive learning between the final layer and the $l^*$-th layer, i.e., adding an auxiliary InfoNCE \citep{oord2018representation} loss between these two layers. Following the original settings in \citet{yu2023xsimgcl} and \citet{yang2024psl}, the magnitude of random noise added to each layer is set to 0.1, the contrastive layer $l^*$ is set to 1 (with the embedding layer considered as the 0-th layer), the temperature of InfoNCE is set to 0.1, and the weight of the auxiliary InfoNCE loss is chosen from the set $\{0.05, 0.1, 0.2\}$.
\end{itemize}

\begin{table*}[t]
    \centering
    \begin{minipage}[t]{0.48\linewidth}
    \centering
    \caption{Optimal hyperparameters on Gowalla dataset.}
    \begin{tabular}{c|l|ccccc}
    \Xhline{1.2pt}
    \textbf{Model} & \multicolumn{1}{c|}{\textbf{Loss}} & \multicolumn{5}{c}{\textbf{Hyperparameters}} \bigstrut\\
    \Xhline{1pt}
    \multirow{9}[2]{*}{MF} 
        & BPR           & 0.001 & 0.000001 &       &       &  \bigstrut[t]\\
        & GuidedRec     & 0.001 & 0     &       &       &  \\
        & SONG@20       & 0.1   & 0     & 0.1   & 0.001 &  \\
        & LLPAUC        & 0.1   & 0     & 0.7   & 0.01  &  \\
        & SL            & 0.1   & 0     & 0.1   &       &  \\
        & AdvInfoNCE    & 0.1   & 0     & 0.1   &       &  \\
        & BSL           & 0.1   & 0     & 0.2   & 0.1   &  \\
        & PSL           & 0.1   & 0     & 0.05  &       &  \\
        & SL@20         & 0.01  & 0     & 0.1   & 1     & 20 \bigstrut[b]\\
    \Xhline{1pt}
    \multirow{9}[2]{*}{LightGCN} 
        & BPR           & 0.001 & 0     &       &       &  \bigstrut[t]\\
        & GuidedRec     & 0.001 & 0     &       &       &  \\
        & SONG@20       & 0.1   & 0     & 0.7   & 0.001 &  \\
        & LLPAUC        & 0.1   & 0     & 0.7   & 0.01  &  \\
        & SL            & 0.1   & 0     & 0.1   &       &  \\
        & AdvInfoNCE    & 0.1   & 0     & 0.1   &       &  \\
        & BSL           & 0.1   & 0     & 0.05  & 0.1   &  \\
        & PSL           & 0.1   & 0     & 0.05  &       &  \\
        & SL@20         & 0.01  & 0     & 0.1   & 0.75  & 5 \bigstrut[b]\\
    \Xhline{1pt}
    \multirow{9}[2]{*}{XSimGCL}
        & BPR           & 0.0001 & 0     &       &       &  \bigstrut[t]\\
        & GuidedRec     & 0.001 & 0     &       &       &  \\
        & SONG@20       & 0.1   & 0     & 0.1   & 0.001 &  \\
        & LLPAUC        & 0.1   & 0     & 0.7   & 0.01  &  \\
        & SL            & 0.01  & 0     & 0.1   &       &  \\
        & AdvInfoNCE    & 0.1   & 0     & 0.1   &       &  \\
        & BSL           & 0.1   & 0     & 0.05  & 0.1   &  \\
        & PSL           & 0.1   & 0     & 0.05  &       &  \\
        & SL@20         & 0.01  & 0     & 0.1   & 0.75  & 5 \bigstrut[b]\\
    \Xhline{1.2pt}
    \end{tabular}
    \label{tab:hyperparameters-gowalla}
    \end{minipage}
    \hfill
    \begin{minipage}[t]{0.48\linewidth}
    \centering
    \caption{Optimal hyperparameters on the Book dataset.}
    \begin{tabular}{c|l|ccccc}
    \Xhline{1.2pt}
    \textbf{Model} & \multicolumn{1}{c|}{\textbf{Loss}} & \multicolumn{5}{c}{\textbf{Hyperparameters}} \bigstrut\\
    \Xhline{1pt}
    \multirow{9}[2]{*}{MF}
        & BPR           & 0.0001 & 0     &       &       &  \bigstrut[t]\\
        & GuidedRec     & 0.001 & 0     &       &       &  \\
        & SONG@20       & 0.1   & 0     & 0.1   & 0.001 &  \\
        & LLPAUC        & 0.1   & 0     & 0.7   & 0.01  &  \\
        & SL            & 0.1   & 0     & 0.05  &       &  \\
        & AdvInfoNCE    & 0.01  & 0     & 0.1   &       &  \\
        & BSL           & 0.1   & 0     & 0.5   & 0.05  &  \\
        & PSL           & 0.1   & 0     & 0.025 &       &  \\
        & SL@20         & 0.01  & 0     & 0.05  & 0.5   & 20 \bigstrut[b]\\
    \Xhline{1pt}
    \multirow{9}[2]{*}{LightGCN}
        & BPR           & 0.001 & 0     &       &       &  \bigstrut[t]\\
        & GuidedRec     & 0.001 & 0     &       &       &  \\
        & SONG@20       & 0.1   & 0     & 0.1   & 0.01  &  \\
        & LLPAUC        & 0.1   & 0     & 0.7   & 0.01  &  \\
        & SL            & 0.1   & 0     & 0.05  &       &  \\
        & AdvInfoNCE    & 0.1   & 0     & 0.1   &       &  \\
        & BSL           & 0.1   & 0     & 0.5   & 0.05  &  \\
        & PSL           & 0.1   & 0     & 0.025 &       &  \\
        & SL@20         & 0.01  & 0     & 0.05  & 0.5   & 20 \bigstrut[b]\\
    \Xhline{1pt}
    \multirow{9}[2]{*}{XSimGCL} 
        & BPR           & 0.0001 & 0.00001 &       &       &  \bigstrut[t]\\
        & GuidedRec     & 0.1   & 0     &       &       &  \\
        & SONG@20       & 0.1   & 0     & 0.1   & 0.001 &  \\
        & LLPAUC        & 0.1   & 0     & 0.7   & 0.01  &  \\
        & SL            & 0.1   & 0     & 0.05  &       &  \\
        & AdvInfoNCE    & 0.1   & 0     & 0.1   &       &  \\
        & BSL           & 0.1   & 0     & 0.05  & 0.05  &  \\
        & PSL           & 0.1   & 0     & 0.025 &       &  \\
        & SL@20         & 0.01  & 0     & 0.05  & 0.5   & 20 \bigstrut[b]\\
    \Xhline{1.2pt}
    \end{tabular}
    \label{tab:hyperparameters-book}
    \end{minipage}
\end{table*}

\subsection{Recommendation Losses} \label{subapp:compared_methods_hyperparameters}

To adequately evaluate the effectiveness of SL@$K$, we reproduce the following SOTA recommendation losses and search for the optimal hyperparameters using grid search. In loss optimization, we use the Adam optimizer \citep{kingma2014adam} with hyperparameters including learning rate (lr) and weight decay (wd). The batch size is set to 1024, and the number of epochs is set to 200, during which all methods are observed to converge. If negative sampling is required, the number of negative samples is set to $N = 1000$, except for the MovieLens dataset, where it is reduced to 200 due to its smaller item set size. The above settings are consistent with the settings in \citet{yang2024psl}. The details of the compared methods and their hyperparameter search spaces are as follows:

\noindentpar{BPR \citep{rendle2009bpr}.}
BPR is a conventional pairwise loss based on Bayesian Maximum Likelihood Estimation (MLE) \citep{casella2024statistical}. The objective of BPR is to learn a partial order among items, \ie positive items should be ranked higher than negative items. Furthermore, BPR is a surrogate loss for AUC metric \citep{rendle2009bpr, silveira2019good}. The score function $s_{ui}$ in BPR is defined as the dot product between user and item embeddings. The hyperparameter search space for BPR is as follows:
\begin{itemize}[topsep=3pt,leftmargin=10pt,itemsep=0pt]
    \setlength{\itemsep}{0pt}
    \item lr $\in \{10^{-1}, 10^{-2}, 10^{-3}, 10^{-4}\}$.
    \item wd $\in \{0, 10^{-4}, 10^{-5}, 10^{-6}\}$. 
\end{itemize}

\noindentpar{GuidedRec \citep{rashed2021guided}.}
GuidedRec is a Binary Cross-Entropy (BCE) \citep{he2017ncf} loss with additional model-based DCG surrogate learning guidance. Rather than being a direct DCG surrogate loss, GuidedRec learns a surrogate loss model to estimate DCG. During training, it maximizes the estimated DCG while minimizing the Mean Squared Error (MSE) \citep{he2017nfm} between the estimated DCG and the true DCG. The score function $s_{ui}$ in GuidedRec is defined as the cosine similarity between user and item embeddings. The hyperparameter search space for GuidedRec is as follows:
\begin{itemize}[topsep=3pt,leftmargin=10pt,itemsep=0pt]
    \item lr $\in \{10^{-1}, 10^{-2}, 10^{-3}\}$.
    \item wd $\in \{0, 10^{-4}, 10^{-5}, 10^{-6}\}$.
\end{itemize}

\noindentpar{LambdaRank \citep{burges2006learning}.}
LambdaRank is a weighted BPR loss \citep{rendle2009bpr}, where the weights are designed heuristically. Though LambdaRank aims to optimize DCG, it is not strictly a DCG surrogate loss, thus lacking theoretical guarantees. The score function $s_{ui}$ in LambdaRank is defined as the dot product between user and item embeddings. The hyperparameter search space for LambdaRank is as follows:
\begin{itemize}[topsep=3pt,leftmargin=10pt,itemsep=0pt]
    \item lr $\in \{10^{-1}, 10^{-2}, 10^{-3}, 10^{-4}\}$.
    \item wd $\in \{0, 10^{-4}, 10^{-5}, 10^{-6}\}$.
\end{itemize}

\noindentpar{LambdaLoss \citep{wang2018lambdaloss}.}
LambdaLoss is a DCG surrogate loss that is formally a weighted BPR loss, similar to LambdaRank \citep{burges2006learning}. \citet{wang2018lambdaloss} demonstrated that LambdaRank does not directly optimize DCG, and proposed LambdaLoss as a strict DCG surrogate loss. The score function $s_{ui}$ in LambdaLoss is defined as the dot product between user and item embeddings. The hyperparameter search space for LambdaLoss is as follows:
\begin{itemize}[topsep=3pt,leftmargin=10pt,itemsep=0pt]
    \item lr $\in \{10^{-1}, 10^{-2}, 10^{-3}, 10^{-4}\}$.
    \item wd $\in \{0, 10^{-4}, 10^{-5}, 10^{-6}\}$.
\end{itemize}

\begin{table*}[t]
    \centering
    \begin{minipage}[t]{0.48\linewidth}
    \centering
    \caption{Optimal hyperparameters on MovieLens dataset.}
    \begin{tabular}{c|l|ccccc}
    \Xhline{1.2pt}
    \textbf{Model} & \multicolumn{1}{c|}{\textbf{Loss}} & \multicolumn{5}{c}{\textbf{Hyperparameters}} \bigstrut\\
    \Xhline{1pt}
    \multirow{6}[2]{*}{MF}
        & LambdaRank                & 0.01  & 0.000001 &       &       &  \bigstrut[t]\\
        & LambdaLoss                & 0.001 & 0.00001 &       &       &  \\
        & LambdaLoss-S              & 0.01  & 0.0001 &       &       &  \\
        & LambdaLoss@20             & 0.001 & 0.00001 &       &       &  \\
        & LambdaLoss@20-S           & 0.01  & 0.00001 &       &       &  \\
        & SL@20                     & 0.01  & 0     & 0.2   & 3     & 5 \bigstrut[b]\\
    \Xhline{1.2pt}
    \end{tabular}
    \label{tab:hyperparameters-movielens}
    \end{minipage}
    \hfill
    \begin{minipage}[t]{0.48\linewidth}
    \centering
    \caption{Optimal hyperparameters on Food dataset.}
    \begin{tabular}{c|l|ccccc}
    \Xhline{1.2pt}
    \textbf{Model} & \multicolumn{1}{c|}{\textbf{Loss}} & \multicolumn{5}{c}{\textbf{Hyperparameters}} \bigstrut\\
    \Xhline{1pt}
    \multirow{6}[2]{*}{MF} 
        & LambdaRank                & 0.001 & 0.00001 &       &       &  \bigstrut[t]\\
        & LambdaLoss                & 0.01  & 0.00001 &       &       &  \\
        & LambdaLoss-S              & 0.001 & 0.0001 &       &       &  \\
        & LambdaLoss@20             & 0.001 & 0.00001 &       &       &  \\
        & LambdaLoss@20-S           & 0.01  & 0.000001 &       &       &  \\
        & SL@20                     & 0.01  & 0     & 0.2   & 2.25  & 5 \bigstrut[b]\\
    \Xhline{1.2pt}
    \end{tabular}
    \label{tab:hyperparameters-food}
    \end{minipage}
\end{table*}

\noindentpar{LambdaLoss@$K$ \citep{jagerman2022optimizing}.}
LambdaLoss@$K$ is a DCG@$K$ surrogate loss that is formally a weighted BPR loss, similar to LambdaRank \citep{burges2006learning} and LambdaLoss \citep{wang2018lambdaloss}. Based on the LambdaLoss framework, \citet{jagerman2022optimizing} proposed LambdaLoss@$K$, which strictly serves as a surrogate for DCG@$K$. The score function $s_{ui}$ in LambdaLoss@$K$ is defined as the dot product between user and item embeddings. The hyperparameter search space for LambdaLoss@$K$ is as follows:
\begin{itemize}[topsep=3pt,leftmargin=10pt,itemsep=0pt]
    \item lr $\in \{10^{-1}, 10^{-2}, 10^{-3}, 10^{-4}\}$.
    \item wd $\in \{0, 10^{-4}, 10^{-5}, 10^{-6}\}$.
\end{itemize}

\noindentpar{SONG \citep{qiu2022large}.}
SONG is an NDCG surrogate loss based on the bilevel compositional optimization technique \citep{wang2017stochastic}. Specifically, SONG first smooths NDCG by surrogating the ranking positions with continuous functions, then calculates the bilevel compositional gradients of NDCG to optimize the model. \citet{qiu2022large} proved that SONG is a lower bound of NDCG, justifying its effectiveness on NDCG optimization. Moreover, \citet{qiu2022large} also establishes the convergence guarantee of SONG under certain assumptions. In practical implementations, SONG utilizes a moving average estimator \citep{kingma2014adam} with update rate $\gamma_{g}$ for ranking positions to stabilize the gradients. The score function $s_{ui}$ in SONG is defined as the cosine similarity between user and item embeddings. The hyperparameter search space for SONG is as follows:
\begin{itemize}[topsep=3pt,leftmargin=10pt,itemsep=0pt]
    \item lr $\in \{10^{-1}, 10^{-2}, 10^{-3}\}$.
    \item wd $\in \{0, 10^{-4}, 10^{-5}, 10^{-6}\}$.
    \item $\gamma_{g} \in \{0.1, 0.3, 0.5, 0.7, 0.9\}$.
\end{itemize}

Since SONG exhibits similar performance to SONG@$K$ \citep{qiu2022large} in our experiments, we only report the results of SONG@$K$.

\noindentpar{SONG@$K$ \citep{qiu2022large}.}
SONG@$K$ is a generalization of SONG \citep{qiu2022large} for NDCG@$K$ optimization. Similar to SONG, SONG@$K$ also utilizes the bilevel compositional optimization technique \citep{wang2017stochastic}. Specifically, it first smooths NDCG@$K$ by surrogating the ranking positions with a moving average estimator \citep{kingma2014adam} with update rate $\gamma_{g}$. To further smooth the Top-$K$ truncation, SONG@$K$ introduces a quantile-based weight similar to SL@$K$. To estimate the Top-$K$ quantile, SONG@$K$ employs a quantile regression loss \citep{koenker2005quantile,hao2007quantile} with learning rate $\eta_{\lambda}$ (\cf \cref{subapp:quantile_regression}). The score function $s_{ui}$ in SONG@$K$ is defined as the cosine similarity between user and item embeddings. The hyperparameter search space for SONG@$K$ is as follows:
\begin{itemize}[topsep=3pt,leftmargin=10pt,itemsep=0pt]
    \item lr $\in \{10^{-1}, 10^{-2}, 10^{-3}\}$.
    \item wd $\in \{0, 10^{-4}, 10^{-5}, 10^{-6}\}$.
    \item $\gamma_{g} \in \{0.1, 0.3, 0.5, 0.7, 0.9\}$.
    \item $\eta_{\lambda} \in \{10^{-1}, 10^{-2}, 10^{-3}, 10^{-4}\}$.
\end{itemize}

\noindentpar{LLPAUC \citep{shi2024lower}.}
LLPAUC is a surrogate loss designed for the lower-left part of AUC. It has been shown to serve as a surrogate loss for Top-$K$ metrics such as Recall@$K$ and Precision@$K$ \citep{shi2024lower,fayyaz2020recommendation}. The score function $s_{ui}$ in LLPAUC is defined as the cosine similarity between user and item embeddings. Following \citet{shi2024lower}'s original settings, the hyperparameter search space for LLPAUC is as follows:
\begin{itemize}[topsep=3pt,leftmargin=10pt,itemsep=0pt]
    \item lr $\in \{10^{-1}, 10^{-2}, 10^{-3}\}$.
    \item wd $\in \{0, 10^{-4}, 10^{-5}, 10^{-6}\}$.
    \item $\alpha \in \{0.1, 0.3, 0.5, 0.7, 0.9\}$.
    \item $\beta \in \{0.01, 0.1\}$.
\end{itemize}

\noindentpar{Softmax Loss (SL) \citep{wu2024effectiveness}.}
SL is a SOTA recommendation loss derived from the listwise Maximum Likelihood Estimation (MLE). Beyond explaining the effectiveness of SL from the perspectives of MLE or contrastive learning, it has been demonstrated that SL serves as a DCG surrogate loss. Specifically, SL is an upper bound of $-\log$ DCG \citep{bruch2019analysis,yang2024psl}, ensuring that optimizing SL is consistent with optimizing DCG. In practice, SL introduces a temperature hyperparameter $\tau$ to control the smoothness of the softmax operator. The score function $s_{ui}$ in SL is defined as the cosine similarity between user and item embeddings. Following the settings in \citet{wu2024effectiveness} and \citet{yang2024psl}, the hyperparameter search space for SL is as follows:
\begin{itemize}[topsep=3pt,leftmargin=10pt,itemsep=0pt]
    \item lr $\in \{10^{-1}, 10^{-2}, 10^{-3}\}$.
    \item wd $\in \{0, 10^{-4}, 10^{-5}, 10^{-6}\}$.
    \item $\tau \in \{0.01, 0.05, 0.1, 0.2, 0.5\}$.
\end{itemize}

\noindentpar{AdvInfoNCE \citep{zhang2024empowering}.}
AdvInfoNCE is a Distributionally Robust Optimization (DRO) \citep{shapiro2017distributionally}-based modification of SL. It introduces adaptive negative hardness into the pairwise score difference $d_{uij}$ in SL. Though this modification may lead to robustness enhancement, it also enlarges the gap between loss and DCG optimization target, which may lead to suboptimal performance \citep{yang2024psl}. In practical implementation, following the original settings in \citet{zhang2024empowering}, the negative weight is fixed at 64, the adversarial learning is performed every 5 epochs, and the adversarial learning rate is set to $5 \times 10^{-5}$. The score function $s_{ui}$ in AdvInfoNCE is defined as the cosine similarity between user and item embeddings. The search space of other hyperparameters for AdvInfoNCE is as follows:
\begin{itemize}[topsep=3pt,leftmargin=10pt,itemsep=0pt]
    \item lr $\in \{10^{-1}, 10^{-2}, 10^{-3}\}$.
    \item wd $\in \{0, 10^{-4}, 10^{-5}, 10^{-6}\}$.
    \item $\tau \in \{0.01, 0.05, 0.1, 0.2, 0.5\}$.
\end{itemize}

\noindentpar{BSL \citep{wu2023bsl}.}
BSL is a DRO-based modification of SL that applies additional DRO to positive instances. It introduces two temperature hyperparameters, $\tau_1$ and $\tau_2$. When $\tau_1 = \tau_2$, BSL is equivalent to SL. The score function $s_{ui}$ in BSL is defined as the cosine similarity between user and item embeddings. Following the settings in \citet{wu2023bsl} and \citet{yang2024psl}, the hyperparameter search space for BSL is as follows:
\begin{itemize}[topsep=3pt,leftmargin=10pt,itemsep=0pt]
    \item lr $\in \{10^{-1}, 10^{-2}, 10^{-3}\}$.
    \item wd $\in \{0, 10^{-4}, 10^{-5}, 10^{-6}\}$.
    \item $\tau_1, \tau_2 \in \{0.01, 0.05, 0.1, 0.2, 0.5\}$.
\end{itemize}

\noindentpar{PSL \citep{yang2024psl}.}
PSL is an NDCG surrogate loss that generalizes SL by substituting the exponential function with a more appropriate activation function. \citet{yang2024psl} proved that PSL establishes a tighter upper bound of $-\log$ DCG than SL, thereby leading to SOTA recommendation performance. Additionally, PSL not only retains the advantages of SL in terms of DRO robustness, but also enhances the noise resistance against false negatives. PSL is also hyperparameter-efficient, requiring only a single temperature hyperparameter $\tau$ to control the smoothness of the gradients. The score function $s_{ui}$ in PSL is defined as \emph{half} the cosine similarity between user and item embeddings. The hyperparameter search space for PSL is as follows:
\begin{itemize}[topsep=3pt,leftmargin=10pt,itemsep=0pt]
    \item lr $\in \{10^{-1}, 10^{-2}, 10^{-3}\}$.
    \item wd $\in \{0, 10^{-4}, 10^{-5}, 10^{-6}\}$.
    \item $\tau \in \{0.005, 0.025, 0.05, 0.1, 0.25\}$.
\end{itemize}

\noindentpar{SL@$K$ (ours).}
SL@$K$ is a DCG@$K$ surrogate loss proposed in this study. Formally, SL@$K$ is a weighted SL with weight $w_{ui} = \sigma_w(s_{ui} - \beta_{u}^{K})$, where $\beta_{u}^{K}$ is the Top-$K$ quantile of user $u$'s preference scores over all items, and $\sigma_w$ is an activation function (\eg the sigmoid function). Intuitively, the weight $w_{ui}$ is designed to emphasize the importance of Top-$K$ items in the gradients, thereby enhancing Top-$K$ recommendation performance. Compared to the conventional SL, SL@$K$ introduces several hyperparameters, including the temperature hyperparameter $\tau_w$ for the quantile-based weight, the temperature hyperparameter $\tau_d$ for the SL loss term, and the quantile update interval $T_{\beta}$. In practice, $\tau_d$ can be set directly to the optimal temperature $\tau$ of SL. The score function $s_{ui}$ in SL@$K$ is defined as the cosine similarity between user and item embeddings. The hyperparameter search space for SL@$K$ is as follows:
\begin{itemize}[topsep=3pt,leftmargin=10pt,itemsep=0pt]
    \item lr $\in \{10^{-1}, 10^{-2}, 10^{-3}\}$.
    \item wd $\in \{0, 10^{-4}, 10^{-5}, 10^{-6}\}$.
    \item $\tau_w \in [0.5, 3.0]$, with a search step of 0.25.
    \item $\tau_d \in \{0.01, 0.05, 0.1, 0.2, 0.5\}$ (set directly to the optimal temperature hyperparameter $\tau$ in SL).
    \item $T_{\beta} \in \{5, 20\}$.
\end{itemize}


\subsection{Optimal Hyperparameters} \label{subapp:optimal_hyperparameters}

We report the optimal hyperparameters of each method on each dataset and backbone in \cref{tab:hyperparameters-health,tab:hyperparameters-electronic,tab:hyperparameters-gowalla,tab:hyperparameters-book,tab:hyperparameters-movielens,tab:hyperparameters-food}. The hyperparameters are listed in the same order as in \cref{tab:hyperparameters-searching}.


\subsection{Information Retrieval Tasks} \label{subapp:experimental_setup_ir}

In addition to the recommendation tasks, we also evaluate SL@$K$ on three additional information retrieval tasks: learning to rank (LTR), sequential recommendation (SeqRec), and link prediction (LP). The experimental setup is as follows:

\begin{itemize}[topsep=3pt,leftmargin=10pt,itemsep=0pt]
    \item \textbf{Learning to rank (LTR)}: LTR aims to order a list of candidate items according to their relevance to a given query. Following \citet{pobrotyn2021neuralndcg}, we compare SL@$K$ with existing LTR losses on a Transformer-based backbone \citep{vaswani2017attention} and three datasets (WEB10K, WEB30K \citep{DBLP:journals/corr/QinL13}, and Istella \citep{dato2016fast}). The baselines include ListMLE \citep{xia2008listwise}, ListNet \cite{cao2007learning}, RankNet \citep{burges2005learning}, LambdaLoss@$K$ \citep{jagerman2022optimizing}, NeuralNDCG \citep{pobrotyn2021neuralndcg}, and SL \citep{wu2024effectiveness}.
    \item \textbf{Sequential recommendation (SeqRec)}: SeqRec focuses on next item prediction in a user's interaction sequence. Following prior work \citep{kang2018self}, we compare SL@$K$ with BCE \citep{kang2018self} and SL \citep{wu2024effectiveness} on Beauty and Games \citep{he2016ups,mcauley2015image} datasets based on SASRec \citep{kang2018self}.
    \item \textbf{Link prediction (LP)}: LP aims to predict links between two nodes in a graph. Following \citet{li2023evaluating}, we compare SL@$K$ with BCE \citep{he2017ncf} and SL \citep{wu2024effectiveness} on a GCN \citep{kipf2016semi} backbone and two datasets (Cora and Citeseer \citep{yang2016revisiting,sen2008collective}).
\end{itemize}



\begin{table}[t]
    \centering
    \caption{Performance comparison between LambdaLoss@$K$ and SL@$K$ on MF backbone. "\textcolor{red}{Imp.}" denotes the improvement of SL@$K$ over LambdaLoss@$K$, while "\textcolor{mydarkgreen}{Degr.}" denotes the degradation of LambdaLoss@$K$ caused by the sample estimation (\ie LambdaLoss@$K$-S). "R@20" and "D@20" denote the Recall@20 and NDCG@20 metrics, respectively. "Time" denotes the average training time (s) per epoch.}
    \label{tab:slatk_versus_lambda}
    \small
    \setlength{\tabcolsep}{2.0pt}
    \begin{tabular}{l|ccc|ccc}
        \Xhline{1.2pt}
        \multicolumn{1}{c|}{\multirow{2}[4]{*}{\textbf{Loss}}} & \multicolumn{3}{c|}{\textbf{MovieLens}} & \multicolumn{3}{c}{\textbf{Food}} \bigstrut\\
        \cline{2-7}          & \textbf{R@20} & \textbf{D@20} & \textbf{Time/s} & \textbf{R@20} & \textbf{D@20} & \textbf{Time/s} \bigstrut\\
        \Xhline{1.0pt}
        LambdaLoss@20 & 0.3418  & 0.3466 & 26  & 0.0530  & 0.0382 & 494  \bigstrut[t]\\
        LambdaLoss@20-S & 0.1580  & 0.1603  & 6  & 0.0335  & 0.0238  & 36  \\
        \textbf{SL@20 (Ours)} & \textbf{0.3580} & \textbf{0.3677} & \textbf{2} & \textbf{0.0635} & \textbf{0.0465} & \textbf{8} \bigstrut[b]\\
        \hline
        \textcolor{mydarkgreen}{\textbf{Degr. \%}} & \textcolor{mydarkgreen}{\textbf{-53.77\%}} & \textcolor{mydarkgreen}{\textbf{-53.75\%}} & --  & \textcolor{mydarkgreen}{\textbf{-36.79\%}} & \textcolor{mydarkgreen}{\textbf{-37.70\%}} & -- \bigstrut[t]\\
        \textcolor{red}{\textbf{Imp. \%}} & \textcolor{red}{\textbf{+4.74\%}} & \textcolor{red}{\textbf{+6.09\%}} & -- & \textcolor{red}{\textbf{+19.81\%}} & \textcolor{red}{\textbf{+21.73\%}} & -- \bigstrut[b]\\
        \Xhline{1.2pt}
    \end{tabular}
\end{table}

\begin{figure}[t]
    \centering
    \includegraphics[width=0.47\textwidth]{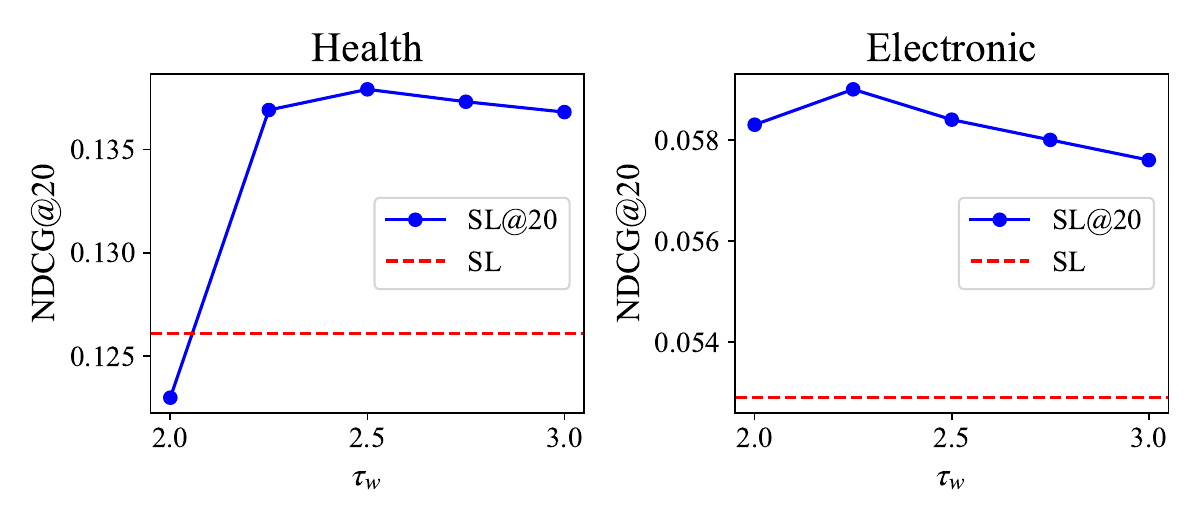}
    \caption{Sensitivity analysis of SL@$K$ on $\tau_w$.}
    \label{fig:sensitivity}
    \Description{Hyperparameter sensitivity of SL@$K$.}
\end{figure}

\section{Supplementary Experimental Results} \label{app:supplementary_experimental_results}


\noindentpar{SL@$K$ \versus Lambda Losses.}
In \cref{tab:slatk_versus_lambda,tab:slatk_versus_lambda-supplementary}, we compare the performance of SL@$K$ against three Lambda losses, including LambdaRank \citep{burges2006learning}, LambdaLoss \citep{wang2018lambdaloss}, and LambdaLoss@$K$ \citep{jagerman2022optimizing}. Results show that SL@$K$ significantly outperforms Lambda losses in terms of accuracy and efficiency. As discussed in \cref{subsec:performance_comparison,app:gradient_vanishing}, Lambda losses suffer from unstable and ineffective optimization processes, leading to suboptimal performance. Additionally, they incur significantly higher computational costs compared to SL@$K$. While sampling strategies (\ie sample ranking estimation in \cref{subapp:sample_ranking_estimation}) could be employed to accelerate the Lambda losses, these approaches result in substantial performance degradation.



\noindentpar{Hyperparameter sensitivity.}
\cref{fig:sensitivity} depicts the performance with varying hyperparameter $\tau_w$ in SL@$K$. Initially, performance improves as $\tau_w$ increases, but beyond a certain point, further increases lead to performance degradation. This indicates an inherent trade-off. When $\tau_w$ is small, the surrogate for NDCG@$K$ is tighter, potentially improving NDCG@$K$ alignment but increasing the training difficulty due to reduced Lipschitz smoothness. As $\tau_w$ increases, the approximation becomes looser, also impacting performance. 

\noindentpar{Noise robustness.}
\cref{fig:slatk_noise_2} illustrates the false positive robustness of SL@$K$ compared to SL, as a supplement to \cref{fig:slatk_noise}.

\begin{table*}[htbp]
    \centering
    \caption{Supplementary results of \cref{tab:slatk_versus_lambda}: Performance comparison of SL@$K$ against three Lambda losses on MF backbone, including LambdaRank, LambdaLoss, and LambdaLoss@$K$. The best results are highlighted in bold, and the best baselines are underlined. "\textcolor{red}{Imp.}" denotes the improvement of SL@$K$ over the best Lambda loss, while "\textcolor{mydarkgreen}{Degr.}" denotes the degradation of Lambda losses caused by the sample ranking estimation (\ie LambdaLoss-S and LambdaLoss@$K$-S, \cf \cref{subapp:sample_ranking_estimation}).}
    \label{tab:slatk_versus_lambda-supplementary}
    \small
    \begin{tabularx}{0.65\textwidth}{l|YY|YY}
        \Xhline{1.2pt}
        \multicolumn{1}{c|}{\multirow{2}[4]{*}{\textbf{Loss}}} & \multicolumn{2}{c|}{\textbf{MovieLens}} & \multicolumn{2}{c}{\textbf{Food}} \bigstrut\\
        \cline{2-5}          & \textbf{Recall@20} & \textbf{NDCG@20} & \textbf{Recall@20} & \textbf{NDCG@20} \bigstrut\\
        \Xhline{1.0pt}
        LambdaRank & 0.3077  & 0.3043  & 0.0520  & 0.0377  \bigstrut\\
        \Xhline{1.0pt}
        LambdaLoss & \uline{0.3425}  & 0.3460  & 0.0515  & 0.0374  \bigstrut[t]\\
        LambdaLoss-S & 0.1497  & 0.1523  & 0.0333  & 0.0243  \bigstrut[b]\\
        \hline
        \textcolor{mydarkgreen}{\textbf{Degr. \%}} & \textcolor{mydarkgreen}{\textbf{-56.29\%}} & \textcolor{mydarkgreen}{\textbf{-55.98\%}} & \textcolor{mydarkgreen}{\textbf{-35.34\%}} & \textcolor{mydarkgreen}{\textbf{-35.03\%}} \bigstrut\\
        \Xhline{1.0pt}
        LambdaLoss@20 & 0.3418  & \uline{0.3466}  & \uline{0.0530}  & \uline{0.0382}  \bigstrut[t]\\
        LambdaLoss@20-S & 0.1580  & 0.1603  & 0.0335  & 0.0238  \bigstrut[b]\\
        \hline
        \textcolor{mydarkgreen}{\textbf{Degr. \%}} & \textcolor{mydarkgreen}{\textbf{-53.77\%}} & \textcolor{mydarkgreen}{\textbf{-53.75\%}} & \textcolor{mydarkgreen}{\textbf{-36.79\%}} & \textcolor{mydarkgreen}{\textbf{-37.70\%}} \bigstrut\\
        \Xhline{1.0pt}
        \textbf{SL@20} & \textbf{0.3580} & \textbf{0.3677} & \textbf{0.0635} & \textbf{0.0465} \bigstrut\\
        \hline
        \textcolor{red}{\textbf{Imp. \%}} & \textcolor{red}{\textbf{+4.53\%}} & \textcolor{red}{\textbf{+6.09\%}} & \textcolor{red}{\textbf{+19.81\%}} & \textcolor{red}{\textbf{+21.73\%}} \bigstrut\\
        \Xhline{1.2pt}
    \end{tabularx}
\end{table*}

\begin{figure*}[t]
    \centering
    \includegraphics[width=\textwidth]{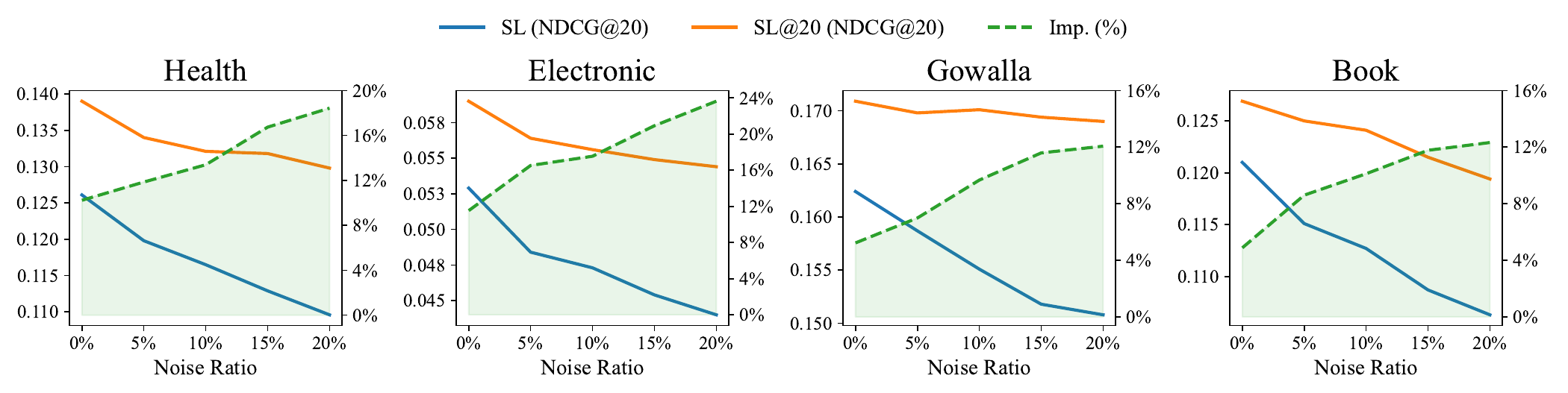}
    \caption{NDCG@20 performance of SL@$K$ compared with SL under varying ratios of imposed false positive instances. "Noise Ratio" denotes the ratio of false positive instances. "Imp." indicates the improvement of SL@$K$ over SL.}
    \label{fig:slatk_noise_2}
    \Description{Noise robustness of SL@$K$.}
\end{figure*}



\end{document}